\documentclass[12pt]{article}
\usepackage{amssymb}
\usepackage{amsmath}
\usepackage{amsfonts}
\usepackage{amsthm}
\usepackage{graphicx}
\usepackage{epstopdf}
\usepackage{color}
\usepackage[onehalfspacing]{setspace}
\usepackage{hyperref}
\usepackage{multirow}
\usepackage{float}
\usepackage{bm}
\usepackage{bbm}
\usepackage{mathabx}
\usepackage{booktabs}
\usepackage{enumitem}
\usepackage{algorithm}
\usepackage{morefloats}
\usepackage{algorithmic}
\usepackage{tikz}
\usetikzlibrary{arrows.meta,positioning,fit, calc, backgrounds,shapes.geometric}

\usepackage[flushleft]{threeparttable}

\def\E{\mathbb{E}}
\def\P{\mathbb{P}}

\numberwithin{equation}{section}

\newtheorem{theorem}{Theorem}[section]
\newtheorem{lemma}{Lemma}[section]
\newtheorem{corollary}{Corollary}[section]
\newtheorem{remark}{Remark}[section]

\newtheorem{proposition}{Proposition}[section]

\theoremstyle{definition}
\newtheorem{assumption}{Assumption}%[section]

\theoremstyle{definition}

\theoremstyle{definition}

\makeatletter
\newcommand{\FORNN}[2][default]{\algorithmicfor\ #2\ \algorithmicdo%
  \ALC@com{#1}\begin{ALC@for}}
\makeatother

\usepackage{geometry}
\usepackage{color, hyperref}
\usepackage{natbib}

\title{
Semiparametric Bayesian Difference-in-Differences\thanks{We thank the associate editor and two anonymous reviewers, as well as
Joachim Freyberger, Björn Höppner, Toru Kitagawa, Soonwoo Kwon, Andriy Norets, Aureo de Paula, Jonathan Roth, Liangjun Su and numerous seminar and conference participants for helpful comments and illuminating discussions. Breunig gratefully acknowledges the support  of the Deutsche Forschungsgemeinschaft (DFG, German Research Foundation) under Germany’s Excellence Strategy – EXC-2047/1 – 390685813. Yu gratefully acknowledges the support of JSPS KAKENHI Grant Number 25K05032.
}
}

		\author{Christoph Breunig\thanks{Department of Economics, University of Bonn. Email: \url{cbreunig@uni-bonn.de}}\and Ruixuan Liu\thanks{CUHK Business School, Chinese University of Hong Kong. Email: \url{ruixuanliu@cuhk.edu.hk}}\and Zhengfei Yu\thanks{Graduate School of Economics, University of Osaka. Email: \url{yu.zhengfei@econ.osaka-u.ac.jp}}}

\begin{document}
  \maketitle

	\begin{abstract}
This paper studies semiparametric Bayesian inference for the average treatment effect on the treated (ATT) within the difference-in-differences (DiD) research design. We propose two new Bayesian methods with frequentist validity. The first one is the semiparametric Bayesian outcome regression, where we place a Gaussian process prior on the conditional mean function of the control group. The second method is a doubly robust Bayesian procedure that adjusts the prior distribution of the conditional mean function and subsequently corrects the posterior distribution of the resulting ATT. We prove new semiparametric Bernstein-von Mises (BvM) theorems for both proposals. Monte Carlo simulations and an empirical application demonstrate that the proposed Bayesian DiD methods exhibit strong finite-sample performance. We also present extensions of the canonical DiD approach, incorporating clustered data and staggered entry with multiple periods.
	\end{abstract}
	{\small \noindent \textsc{Keywords}: Difference-in-differences, conditional parallel trends, semiparametric Bayesian inference, Bernstein–von Mises theorem, double robustness. \bigskip\ }

\section{Introduction}
The Difference-in-Differences (DiD) method is widely used in causal inference for evaluating policy interventions while accounting for unobserved time-invariant heterogeneity. The primary parameter of interest in this context is the average treatment effect on the treated (ATT). One of its key identifying conditions is the (conditional) parallel trends assumption, i.e., treated and control groups would exhibit similar trends absent treatment after adjusting for covariates \citep{abadie2005,sant2020doubly}. While the related literature is largely frequentist, this paper introduces a Bayesian framework under conditional parallel trends, avoiding parametric assumptions about how covariates contribute to treatment status and potential outcomes. Our approach yields point estimators and credible sets in a unified manner.

We propose two novel Bayesian methods for inference on the ATT in the DiD framework. The first is a semiparametric Bayesian outcome regression approach. This method places a Gaussian process (GP) prior on the conditional mean function of the control group, while modeling the remaining components of the likelihood nonparametrically using a Dirichlet process prior. As a result, the approach avoids the need to specify Bayesian models for either the conditional mean of the treated group or the propensity score.
The proposed method can be viewed as a Bayesian counterpart to frequentist outcome-regression imputation. The Gaussian process prior allows the control-group conditional mean function to be modeled flexibly without imposing a finite-dimensional functional form \citep{williams2006gaussian}. The semiparametric Bayesian formulation then propagates uncertainty about this function into the posterior distribution of the ATT, yielding inference directly from the posterior.
We prove that the resulting posterior satisfies a Bernstein–von Mises (BvM) theorem under classical semiparametric regularity conditions. Hence, this Bayesian outcome regression procedure is asymptotically equivalent to the semiparametrically efficient frequentist estimator, although its validity relies on stronger smoothness requirements than the doubly robust Bayesian procedure introduced below.

We also provide a doubly robust version of our Bayesian outcome regression procedure. This is well suited for more complex models, either due to a larger number of continuous covariates or when the underlying conditional mean functions are not smooth. Our \textit{doubly robust Bayesian procedure} adjusts the prior and posterior distributions closely associated with the efficient influence function of estimating the ATT, in the spirit of \cite{BLY2022}. Based on prior adjustment alone, the Bayesian procedure satisfies the  semiparametric BvM theorem, albeit with a ``bias term" in the posterior,  under double-robust smoothness conditions. The resulting posterior distribution depends on the unknown true conditional mean and propensity score functions. Our double-robust Bayesian approach addresses this ``bias term" by incorporating an explicit posterior correction. The general purpose of introducing such robustification shares a similar motivation with the frequentist counterpart of \cite{sant2020doubly}. We stress that the construction of both the prior adjustment and the posterior correction is unique to the semiparametric Bayesian structure. In addition, the new BvM theorem allows one to consider the non-Donsker class.

Our Monte Carlo results demonstrate competitive finite-sample performance for both proposed Bayesian methods. The nonparametric flexibility of the Gaussian process prior over the conditional mean function offers greater robustness than many alternatives relying on parametric nuisance estimates, particularly when potential outcomes and treatment assignment are complex functions of the covariates.  %\textcolor{blue}{I would delete this sentence, as I am not sure whether this is the main advantage and we stress this also at other points}
 Meanwhile, Bayesian approaches also produce shorter credible intervals relative to the frequentist double machine learning (DML) methods that model the nuisance functions nonparametrically. 
 We conjecture that this efficiency gain stems from uncertainty quantification through averaging over the Gaussian process posterior of the conditional mean function, thereby avoiding the variance inflation that sometimes occurs in DML estimators in finite-samples. Between the two Bayesian methods, the doubly robust version performs more stably as the model complexity increases. 
  Our results further suggest that the strong performance of Bayesian procedures is driven primarily by uncertainty quantification through the posterior distribution, rather than the choice of pilot estimators in our correction steps. 

Our Bayesian methodology requires specifying a likelihood function for the control arm, which we model using an exponential family structure. For example, in the Gaussian case, this leads to computational efficiency since the posterior remains multivariate Gaussian, allowing us to avoid computationally intensive methods like MCMC. If the likelihood function is misspecified, our approach naturally transitions to a quasi-Bayesian procedure.
We provide a new posterior contraction result for the conditional mean function with respect to the quasi-posterior. This is in line with our Monte Carlo results, which show that our Bayesian procedures are resilient to likelihood misspecification caused by non-normal errors. 
We also adapt the Bayesian proposals to clustered data settings, where inference is conducted at a more aggregate level than the cross-sectional units.
Finally, we extend our methods to study staggered treatment entry with multiple time periods.

Nonparametric Bayesian causal inference has recently received considerable interest. We refer readers to \cite{daniels2024BNP} for numerous applications. \cite{ray2020causal} develop the comprehensive theory for establishing the BvM theorem in the missing data framework, employing Gaussian process priors for the conditional mean function. Extending their methodology to the ATT would require nonparametric Bayesian modeling of both the propensity score and the conditional mean function for the treated group, as discussed in Remark \ref{rem:dist:ate}. A key innovation of our proposal is to circumvent this route by building our Bayesian procedure on a reparametrization that is particularly convenient for the analysis of the ATT.

 Building on the prior adjustment of \cite{ray2020causal}, \cite{BLY2022} introduce a debiasing step to further correct the posterior and establish the BvM theorem for the average treatment effect (ATE) under double robustness. Although they outline an extension to general semiparametric models where the parameter of interest can be written as the linear functional of conditional means, this approach does not cover the case of the ATT, because of its ratio form. The correction steps in our Bayesian method share the same motivation as in \cite{BLY2022}, but differ in their functional form. Additionally, we find that the bias term in the BvM for ATT is substantially simpler than that for ATE, which also explains the favorable finite-sample behavior of Bayesian ATT estimators even without posterior corrections. To the best of our knowledge, our proposed doubly robust BvM theorem for the ATT is the first to relax the Donsker property assumption for the conditional mean. See also Remark \ref{remarkYiu}. The robust performance of various Bayesian procedures, judged by frequentist criteria, has been well-documented across a range of settings, from Bayesian Additive Regression Trees (BART) to high-dimensional linear regression models. We refer readers to \cite{ray2019debiased}, \cite{hahn2020bart}, \cite{BLY2025RoBART}, \cite{yiu2023corrected}, as well as \cite{DiTraglia2025Bayesian}. 

 There are also scattered results exploring Bayesian methodology for the study of ATT or DiD. In an earlier paper, \cite{chib2002treat} developed a semiparametric Bayesian model for the ATT in both cross-sectional and panel data settings. Their semiparametric model differs from our setup in that covariates enter the outcome equation linearly, while the error term is modeled using flexible Dirichlet process mixtures. Recently, in the context of assessing sensitivity to the parallel trends assumption, \cite{kwon2024empirical} proposed a Bayesian approach. Our adoption of the exponential family on the conditional distribution of the control group shares some similarity with \cite{wooldridge2023simple}. Considering the DiD setup with multiple periods, \cite{wooldridge2023simple} has proposed a generalized linear model with a parametric link function to handle limited dependent outcomes $Y_{it}$ in \textit{levels}. In these cases, the parallel trends assumption is also falsifiable \citep{roth2023trends}.

The remainder of this paper is organized as follows. Section \ref{sec:model} presents the setup and introduces the Bayesian framework in the DiD setup.
Section \ref{sec:method_outline} outlines our Bayesian methods.
In Section \ref{sec:inference}, we establish inference via semiparametric  BvM theorems for our first method.
In Section \ref{sec:asympt:prior}, we derive a doubly robust, semiparametric BvM theorem for our second method.
Section \ref{sec:numerical} presents finite-sample results via simulations and an empirical illustration. Towards the end, we outline two extensions of our methodology in
Section \ref{sec:extend}.
Proofs of main theoretical results are collected in Appendix \ref{appendix:main:proofs}.  Appendix \ref{sec:sq:exp} provides BvM results under primitive conditions when using squared exponential process priors. 
Supplementary Appendices \ref{appendix:lfd}--\ref{appendix:simu} provide additional technical results and further simulation evidence.

	\section{Setup and Implementation}\label{sec:model}
This section provides the main setup of the average treatment effect on the treated (ATT) in the difference-in-differences (DiD) design. We first list standard conditions for the identification of the ATT and introduce the Bayesian formulation of the problem.

\subsection{Setup}\label{sec:setup}
We focus on the canonical DiD design case, where there are two treatment periods and two treatment groups. Let $Y_{i t}$ be the observable outcome of interest for unit $i$ at time $t$. We assume that researchers have access to outcome data in a pre-treatment period $t=1$ and in a post-treatment period $t=2$. Let $D_{i t}=1$ if unit $i$ is treated before time $t$ and $D_{i t}=0$ otherwise. Note that $D_{i 1}=0$ for every $i$ and thus we may write $D_i=D_{i 2}$. Using the potential outcome notation, $Y_{it}(0)$ or $Y_{it}(1)$ denotes the outcome of unit $i$ at time $t$ if it does not receive or receives treatment by time $t$, respectively. Thus, the realized outcome for unit $i$ at time $t=1$ is $Y_{i 1}= Y_{i 1}(0)$, and at time $t=2$ it is $Y_{i 2}=D_i Y_{i 2}(1)+\left(1-D_i\right) Y_{i 2}(0)$. 
Below, $P_0$ denotes the frequentist distribution generating the observed data.

A vector of $p$-dimensional pre-treatment covariates $X_i$ is also available, with its true cumulative distribution function denoted by $F_X$.\footnote{If $X_i$ does not have a density we can simply consider the conditional density of $(\Delta Y_i, D_i)$ given $X_i=x$ instead of the joint density of $(\Delta Y_i, D_i, X_i)$.}
Let 
$\pi_0(x)=P_0(D_i=1\mid X_i=x)$
denote the propensity score, $\bar\pi_0=P_0(D_i=1)$ the proportion of being treated, and $m_0(x)= \mathbb E_0[\Delta Y_i\mid D_i=0, X_i=x]$ the conditional mean of the differenced outcome across two periods, where $\Delta Y_{i}:=Y_{i2}-Y_{i1}$ and where $\mathbb{E}_0[\cdot]$ denotes the expectation under $P_0$.
The researcher observes an independent and identically distributed (\textit{i.i.d.}) sample of 
$\left\{(Y_{i1},Y_{i2},D_i,X_i^\top)\right\}_{i=1}^n$. In the discussion of identification, we will henceforth suppress the unit index $i$ for notational simplicity.

Regarding the causal effect in the canonical DiD setup, the related literature primarily focuses on the average treatment effect on the treated (ATT) given by
\begin{align*}
\tau_0=\mathbb E_0[Y_{2}(1)-Y_{2}(0)|D=1].
\end{align*}
 For its identification, we follow the literature \citep{abadie2005,sant2020doubly} and assume no anticipation,  conditional parallel trends (PTA) given covariates $X$, and overlap.
 		\begin{assumption}\label{Ass:unconfounded}
 		 For all $x$ in the support of $F_X$ we have:\\
 		 (i) $\E_0\left[Y_1(0) \mid D=1, X=x\right]=\E_0\left[Y_1(1) \mid D=1, X=x\right]$ \hfill (No Anticipation), \\
		(ii) $\E_0\left[Y_2(0)-Y_1(0) \mid D=1, X=x\right]=\E_0\left[Y_2(0)-Y_1(0) \mid D=0, X=x\right]$ \hfill(PTA),\\
 (iii) $P_0(D=1) > \epsilon$ and $P_0(D=1 \mid X=x) \leq 1-\epsilon$ for some $\epsilon>0$ \hfill (Overlap).
		\end{assumption}
		
Under Assumption \ref{Ass:unconfounded} the ATT is identified by
\begin{equation}\label{att_0}
	\tau_0
	=\E_0\left[\Delta Y - m_0(X)\mid D=1\right]=\frac{\mathbb E_0\left[D\left(\Delta Y-m_0(X)\right)\right]}{\mathbb  E_0[D]}.
\end{equation}
One can construct an estimator that replaces the conditional mean function $m_0$ with an estimator, known as the outcome regression approach, as described in \cite{heckman1997matching}. Alternatively, the inverse propensity score weighted estimator is proposed by \cite{abadie2005} and the doubly robust version is developed by \cite{sant2020doubly}. It has also been noted in the recent literature that, in the presence of heterogeneous treatment effects in $X$, i.e., when $\mathbb E_0[Y_{2}(1)-Y_{2}(0)\mid X=x, D=1]$ varies with $x$, the two-way fixed effects estimator (TWFE) is generally not consistent for the ATT. See Remark 1 of \cite{sant2020doubly} for an explicit discussion.

\subsection{A Bayesian Framework}\label{sec:bayes_frame}
We now provide the formal Bayesian setup to the ATT in the DiD context. We consider a family of probability distributions $\{P_\eta:\eta\in\mathcal H\}$ for some parameter space $\mathcal H$. The (possibly infinite dimensional) parameter $\eta$ characterizes the probability model. Let $\eta_0$ be the true value of the parameter and denote $P_0=P_{\eta_0}$, which corresponds to the frequentist distribution generating the observed data. 
Under $P_\eta$ where $\eta=(\pi,f_{X}, f_{\Delta Y|D,X})$, the joint density function of $Z=(\Delta Y,D, X^\top)^\top$ can thus be written as
\begin{equation}\label{condpdf:eta}
	p_{\eta}(y,d,x)=	\underbrace{f_X(x)\pi^d(x)(1-\pi(x))^{1-d}f^d_{\Delta Y|D,X}(y\mid 1,x)}_{=:f(dy,d,x)}f^{1-d}_{\Delta Y|D,X}(y\mid 0,x),
\end{equation}
where $f_{\Delta Y \mid D,X}(y \mid d,x)$ is the conditional density of $\Delta Y$ given $(D, X)$. Here $f$ denotes the joint density of $(D\Delta Y, D, X^\top)^\top$ under $P_\eta$, and the corresponding cumulative distribution function is denoted by $F$. For the control group, we assume that $f_{\Delta Y \mid D,X}(y \mid d=0,x)$ follows the exponential family condition specified by \eqref{condpdf} in the sequel. A central insight of this paper is to show that a nonparametric process prior specification on the conditional mean function $m$ and the distribution $F$ is sufficient for the Bayesian inference on the ATT parameter. Specifically, there is no need to additionally parameterize the propensity score $\pi$, the marginal density of $X$, or the conditional density for the treated group,  $f_{\Delta Y| D, X}(y\mid 1,x)$.

We consider the following reparametrization of $(m, f)$ given by $\eta=(\eta^m,\eta^{f})$. We index the probability model by $P_{\eta}$, where
\begin{eqnarray*}
m_\eta= q^{-1}(\eta^m) \,\,\text{ and }\,\, f_\eta=\exp(\eta^{f})
\end{eqnarray*}
for some known, invertible function $q(\cdot)$, which we specify below. We can write the ATT depending on $\eta$ as
\begin{equation}\label{att}
		\tau_\eta:=
%		\frac{\int \big(y-m_{\eta}(x)\big)\mathrm{d}F_{\eta}^1(y,x)}{\int \mathrm{d}F_{\eta}^1(y,x)}=:
\frac{\mathbb{E}_\eta[D\Delta Y-Dm_{\eta}(X)]}{\mathbb{E}_\eta[D]},
	\end{equation}
	where $\mathbb{E}_\eta$ denotes the expectation under $P_\eta$ and in this case, is the integral with respect to the density $f_\eta$.

As we saw above, we only need to impute the conditional mean of the outcome in the control group. We assume that the distribution of $\Delta Y$, conditional on $D=0$ and $X$, belongs to the ``single-parameter" exponential family, where the only unknown component is the nonparametric conditional mean function $m(x)=\mathbb{E}[\Delta Y\mid D=0,X=x]$. The conditional density function is given by
\begin{equation}\label{condpdf}
	f_{\Delta Y|D,X}(y\mid 0,x) = c(y)\exp\left[q(m(x))ay-A(m(x))\right],
\end{equation}
where $A(m)= \log\int c(y)\exp\left[q(m) ay\right]\mathrm{d}y$, some constant $a>0$, and the function $q(\cdot)$ links the conditional mean to the ``natural parameter'' of the exponential family. We also restrict the sufficient statistic to be linear in $y$. The exponential family assumption implies the conditional mean equation $\E[ \Delta Y\mid D=0, X=x]=A'(m(x))/(a\, q'(m(x)))=m(x)$. We emphasize that \eqref{condpdf} does not impose functional form assumptions on the conditional mean function $m$ and, in particular, does not restrict the ATT parameter. 

The family (\ref{condpdf}) allows for count and continuous outcomes, which can be convenient for modeling various types of outcomes. It constitutes a comprehensive class in nonparametric Bayesian analysis; see Chapter 2 of \cite{ghosal2017fundamentals}. For instance, when $a=1$,  the Poisson distribution  corresponds to the choices $c(y)= 1/(y!)$, $q(m)=\log m$, and $A(m)=m$, while the  exponential distribution is represented by $c(y)=1$, $q(m)=-1/m$, and $A(m)=\log m$. Furthermore, the normal distribution\footnote{For the normal case, we treat $\sigma$ as a hyperparameter and estimate it together with hyperparameters in the Gaussian process prior by maximizing marginal likelihood. We note that while a generalization to multinomial outcomes, as in \cite{BLY2022}, is possible, we do not consider this case explicitly in this paper.} with  $\text{Var}(\Delta Y|D=0,X)=\sigma^2$ for some $\sigma>0$, is captured by $c(y)=\exp(-y^2/(2\sigma^2))/\sqrt{2\pi\sigma^2}$, $q(m)=m/\sigma$, $A(m)=m^2 / (2\sigma^2)$, and $a=1/\sigma$.

Our BvM results require posterior contraction of $m_{\eta}$ toward the true conditional mean $m_0$. Under mild conditions, this can hold even if the likelihood is misspecified.
Generally, the posterior of $m_{\eta}$ will contract near the point (pseudo-true value) in the support of the prior that minimizes the Kullback–Leibler (KL) divergence with respect to the true data generating probability. We formalize this generalization in Section \ref{sec:misspecify}. This aligns with our finite-sample results, which are not sensitive to deviations from exponential family distributions, as shown in Appendix \ref{appendix:simu}. Beyond the exponential family, one can also consider the quasi-likelihood in \cite{linero2025Quasi} or the flexible nonparametric Bayesian approach in \cite{norets2022} to model the conditional distribution of the outcome for the control group.
 
\begin{remark}[ATT in Cross-Sectional Setting]
The results of our paper contribute to the literature of ATT using cross-sectional data, i.e.,  where an i.i.d. sample of $(Y_i, D_i, X_i^\top)^\top$ for $i=1,\ldots, n$ is available.
In this case, a	 specific example captured by the single-parameter exponential family is when the outcome variable is binary, where  $q(m)=\log(m/(1-m))$, $A(m)=-\log(1-m)$, and $c(y)=a=1$. This binary outcome case does not require any distributional assumptions. 
Interestingly, the sample ATT, given by $(\sum_{i=1}^nD_i)^{-1}\sum_{i=1}^nD_i(m_0(1,X_i)-m_0(0,X_i))$, requires only a prior on the conditional mean functions, without the need to specify a Dirichlet process prior.  On the other hand,  a prior for the conditional mean function of the treatment group is also necessary in this case. We do not address a Bayesian approach for the sample ATT in this paper.
\end{remark}

				\section{Proposed Bayesian Methodology}\label{sec:method_outline}
			We now present two Bayesian procedures that build on flexible prior processes, enabling semiparametric inference on the ATT. The first corresponds to a semiparametric Bayesian outcome regression based on Gaussian process priors. The second involves Bayesian methods with frequentist modifications, incorporating both an adjustment to the prior and a posterior correction.
		
				\subsection{Semiparametric Bayesian Outcome Regression}\label{sec:method_outline:GP}
			We first propose the semiparametric Bayesian analog to \cite{heckman1997matching}, which builds on a standard Gaussian process prior for the conditional mean function combined with an independent Dirichlet process prior for the conditional expectation in \eqref{att}. The Bayesian outcome regression may not achieve double robustness, as we see in the next section. In our simulation results, however, we find that the proposed method is robust even in cases of near overlap failure.

The use of Gaussian process priors for the conditional mean has the following motivation. The mode of a posterior stemming from Gaussian process priors can be derived by a minimization problem involving the corresponding norm of a so-called reproducing kernel Hilbert space (RKHS). Gaussian process (GP) priors share close ties with spline regression \citep{wahba1990spline}. Their strong finite-sample performance has fueled their popularity in machine learning \citep{williams2006gaussian,murphy2023pml}. The Dirichlet process is a default prior on spaces of probability measures. By definition of the ATT $\tau_{\eta}$, we assign a Dirichlet process prior to model the distribution $F_{\eta}$, which induces the so-called Bayesian bootstrap when the base measure of the Dirichlet process is taken to be zero; see \cite{rubin1981bayesian} and \cite{chamberlain2003bayesian}.
\begin{algorithm}[H]
\caption{Bayesian Procedure Using Standard Gaussian Process Priors}
\label{algorithm_1}
\begin{algorithmic}
    \STATE \textbf{Input:} Data $Z_i=(\Delta Y_i,D_i,X_i^\top)^\top$ for $i=1,\dots,n$,  and number of posterior draws $S$.
    \STATE \textbf{Prior Specification:}  Select a Gaussian process prior $W^m$ and  set the prior for
    \begin{align}\label{prior:uc}
m_\eta(x) = q^{-1}\left(\eta^m(x)\right)\qquad \text{and}\qquad \eta^m(x)=W^m(x).
\end{align}	
\textbf{Posterior Computation:} \FORNN{$s=1,\ldots, S$}
        \STATE 
  (a) Generate the $s$-th draw of the posterior of $(m_\eta(X_i))_{i=1}^n$ using the Gaussian process prior and the data from the control arm; denote it as $(m^s_\eta(X_i))_{i=1}^n$.
            \STATE (b) Draw Bayesian bootstrap weights $M^s_{i}=e^s_i/\sum_{j=1}^n e^s_j$ where $e_i^s \stackrel{iid}{\sim} \textup{Exp}(1)$, $1\leq i\leq n$.\hskip-3cm
        \STATE (c) Calculate a posterior draw for the ATT: 
         \begin{equation}\label{NpBayes}
\tau_\eta^s=\frac{\sum_{i=1}^n M_{i}^sD_i \big(\Delta Y_i-m_\eta^s(X_i)\big)}{\sum_{i=1}^n M_{i}^s D_i}.
\end{equation}
    \ENDFOR
    
 \STATE \textbf{Output: $\{\tau_\eta^{s}:s=1,\ldots,S\}$} 
\end{algorithmic}
\end{algorithm}
Algorithm \ref{algorithm_1} allows for simultaneous point estimation and uncertainty quantification. The $100\cdot(1-\alpha)\%$ credible set $\mathcal{C}_n(\alpha)$ is computed by
\begin{equation}\label{CS:def}
\mathcal{C}_n(\alpha)=\big\{\tau: q(\alpha/2)\leq \tau \leq q(1-\alpha/2)\big\},
\end{equation}
where $q(a)$ denotes the $a$ quantile of $\{\tau_\eta^s:s=1,\ldots,S\}$. We also obtain the Bayesian point estimator (the posterior mean) by averaging the simulation draws: $\overline{\tau}_{\eta}=S^{-1}\sum_{s=1}^S \tau_\eta^s$.

For the choice of the prior process $W^m$, we use a Gaussian process with mean $\mu$ and the squared exponential (SE) covariance function $K\left(\cdot,\cdot\right)$ \citep[p.83]{williams2006gaussian} given by
\begin{equation}\label{SE_cov}
K\left(x,x^\prime\right):= \nu^2 \exp\left(-\sum_{l=1}^{p}a_{ln}^{2}(x_{l}-x^\prime_{l})^2/2\right),
\end{equation}
where the hyperparameter $\nu^2$ is the kernel variance and $a_{1n},\ldots,a_{pn}$ are rescaling parameters that reflect the relevance of  each covariate in predicting $\eta^m$. In practice, the hyperparameters $\mu$, $\nu$, and $a_{1n},\ldots,a_{pn}$ can be chosen by maximizing the marginal likelihood. When the exponential family specification in (\ref{condpdf}) takes the Gaussian form, Step (a) of posterior computation in Algorithm \ref{algorithm_1} is analytically tractable and computationally very efficient. See Supplementary Appendix \ref{details:implementation} for details. For non-Gaussian cases, one can use Laplace approximation or MCMC for Step (a) as suggested by \cite{williams2006gaussian}. The construction of the Gaussian process prior naturally accommodates discrete covariates. It would be interesting to establish the posterior contraction rate under the framework of \cite{norets2022}. However, the asymptotic results for our Bayesian approaches under Gaussian process priors (Appendix \ref{sec:sq:exp}) focus on continuous covariates.

				\subsection{Doubly robust Bayesian DiD }\label{sec:method_outline:DR}
%\textcolor{red}{alternative name: doubly robust Bayesian DiD }
Our doubly robust approach begins with prior adjustment via inverse propensity score weighting (IPW) in the least favorable direction, which is also related to the efficient influence function. This can be compared with \cite{sant2020doubly} who combine outcome regression (OR) and IPW to achieve  doubly robust estimation in the frequentist setting. Following \cite{hahn1998role,hirano2003efficient} or, in the DiD setup,  \cite{sant2020doubly}, the efficient influence function for the ATT is given by 
\begin{align}\label{eif_att}
	\widetilde \tau_{\eta}(\Delta Y,D,X) 
%	&=\gamma_\eta(D,X) (\Delta Y-m_\eta(D,X)) + \frac{D}{\pi_\eta} \left(m_\eta(1,X)-m_\eta(0,X)-\tau_\eta \right)\nonumber\\
	&= \gamma_\eta(D,X)(\Delta Y-m_\eta(X))-\frac{D}{\bar\pi_\eta}\tau_\eta,
\end{align}
with its Riesz representer $\gamma_\eta$ given by
\begin{align}\label{riesz:def}
	\gamma_\eta(d,x)= \frac{d}{\bar\pi_\eta} -\frac{1-d}{\bar\pi_\eta} \frac{\pi_\eta(x)}{1-\pi_\eta(x)}.
\end{align}
We show in the Supplemental Appendix \ref{appendix:lfd} that the Riesz representer $\gamma_\eta$ determines the \emph{least favorable direction} (LFD) associated with the Bayesian submodel with the largest variance. Our prior adjustment using this Riesz representer provides exact invariance under shifts in nonparametric components along this direction. This extends the work of \cite{ray2020causal} on unconditional average treatment effects to the ATT case, where the LFD takes a different functional form. In a similar vein to \cite{BLY2022}, we use the Riesz representer to correct for posterior bias under double-robust smoothness conditions.

Our prior and posterior adjustments depend on a preliminary estimator of $\gamma_0$. A pilot estimator for the propensity score $\pi_0(\cdot)$,  denoted by $\widehat{\pi}(\cdot)$,  is based on an auxiliary sample. So is the estimator of the treated proportion $\bar\pi_0$, which is taken to be the sample mean of the treatment indicators, denoted by $\widehat{\bar{\pi}}$.  In practice,  $\widehat{\pi}(\cdot)$ can be obtained by a parametric,  such as a logistic regression,  or by a nonparametric approach such as random forest.  As our simulation shows,  the latter is more robust to possibly nonlinear data-generating processes.  The Riesz representer $\gamma_0$ is estimated by plugging in $\widehat{\pi}(x)$ and  $\widehat {\bar \pi}$:
\begin{align}\label{riesz:est}
	\widehat{\gamma}(d,x)=\frac{d}{\widehat{\bar{ \pi}}} - \frac{1-d}{\widehat{\bar{\pi}}}\frac{\widehat\pi(x)}{1-\widehat\pi(x)}=\frac{d-\widehat {\pi}(x)}{(1-\widehat\pi(x))\widehat{\bar{ \pi}}}.
\end{align}
Our posterior adjustment also involves a pilot estimator $\widehat{m}$ for the conditional mean function $m_0$.  We recommend simply using the posterior mean of conditional mean function:
\begin{equation}\label{m:est}
\widehat m(x)=\frac{1}{S}\sum_{s=1}^S m_\eta^{s}(x),
\end{equation} 
where the posterior draws $\{m_\eta^{s}: s=1,\dots,S\}$ are obtained following Step (a) of the posterior computation in Algorithm \ref{algorithm_1}.

Algorithm \ref{algorithm_2} describes our doubly robust Bayesian procedure that approximates the posterior distribution of $\tau_{\eta}$ given in equation \eqref{att}. Let $n_c$ denote the number of observations in the control arm.  Algorithm \ref{algorithm_2} also leads to simultaneous point estimation and uncertainty quantification. The $100\cdot(1-\alpha)\%$ credible set $\mathcal{C}_n(\alpha)^{DR}$ for the ATT parameter $\tau_0$ is as in \eqref{CS:def}, but here $q(a)$ denotes the $a$ quantile of $\{ \check{\tau}_\eta^s:s=1,\ldots,S\}$. The Bayesian point estimator is given by $\overline{\tau}_{\eta}^{DR}=S^{-1}\sum_{s=1}^S \check{\tau}_\eta^{s}$.

The use of auxiliary data for the estimation of unknown functional parameters simplifies the technical analysis; see \cite{ray2020causal} for propensity score adjusted priors in the case of missing data. To make efficient use of the data, we propose a $K$-fold cross-fitting that is inspired by the DML literature.  The next remark outlines the key steps.  A complete description of the procedure can be found in Algorithm \ref{algorithm_2_cf} of Appendix \ref{appendix:crossfit}.
In essence,  posterior draws of the conditional mean for the control group are obtained for each chosen fold, while the remaining folds are used to obtain the pilot estimators such as $\widehat{\pi}(\cdot)$ and $\widehat{m}(\cdot)$. Then, we rotate across folds and average all posterior draws. 
\begin{algorithm}[H]
\caption{Doubly Robust Bayesian Procedure}
\label{algorithm_2}
\begin{algorithmic}
    \STATE \textbf{Input:} Data $Z_i=(\Delta Y_i,D_i,X_i^\top)^\top$ for $i=1,\dots,n$, number of posterior draws $S$,  pilot estimators $\widehat \gamma$ and $\widehat m$. Let $n_c$ be the number of control units. 
   
    \STATE \textbf{Prior Specification:}  Set the adjusted prior:
\begin{align}\label{prior:ps}
m_{\eta}(x) = q^{-1}\left(\eta^m(x)\right),  \quad \eta^m(x)=W^m(x) + \lambda\,\widehat \gamma(0,x), 
\end{align}    
where $W^m$ is a Gaussian process independent of $\lambda \sim N(0,\varsigma_n^2)$,  $\varsigma_n=\nu\log n_c/(\sqrt{n_c}\, \Gamma_n)$, $\nu$ is the hyperparameter in (\ref{SE_cov}), and $\Gamma_n=\sum_{i=1}^n\vert\widehat\gamma(0,X_i)(1-D_i)\vert/n_c$.

\textbf{Posterior Computation:} \FORNN{$s=1,\ldots, S$}
        \STATE  (a)
        Generate the $s$-th draw of the posterior of $(m_\eta(X_i))_{i=1}^n$ using the adjusted prior in (\ref{prior:ps}) and the data from the control arm; denote it as $(m^s_\eta(X_i))_{i=1}^n$.
     \STATE (b) Draw Bayesian bootstrap weights $M^s_{i}=e^s_i/\sum_{j=1}^n e^s_j$ where $e_i^s \stackrel{iid}{\sim} \textup{Exp}(1)$, $1\leq i\leq n$.\hskip-3cm
      
        \STATE (c) Calculate the ATT draws by 
       $\check{\tau}_\eta^s=\tau_\eta^s-\widehat{b}^s_{\eta}$, where $\tau_{\eta}^s$ is given in \eqref{NpBayes} but using the adjusted prior. The posterior correction $\widehat{b}^s_{\eta}$ is given by
\begin{align}\label{recentering_term}
\widehat{b}^s_{\eta}=\frac{1}{n}\sum_{i=1}^n \widehat{\gamma}(D_i,X_i) (\widehat m-  m_{\eta}^s)(X_i).
\end{align}
    \ENDFOR
    \STATE \textbf{Output: $\{\check{\tau}_\eta^s:s=1,\ldots,S\}$} 
\end{algorithmic}
\end{algorithm}

\begin{remark}[Extension to Cross-fitting]\label{rmk:cross_fit}
We also consider an extension of Algorithm~\ref{algorithm_2} to cross-fitting based on $K$-fold sample splitting. For the exact implementation, we refer to Algorithm \ref{algorithm_2_cf} in Appendix~\ref{appendix:crossfit}; here, we present the key steps. Let $(I_k)_{k=1}^K$ be a partition of $\{1,\dots,n\}$, each with size $n_k$.
    For each fold $k$, we obtain the off-fold pilot estimators
$(\widehat \gamma^{(-k)},\widehat m^{(-k)})$ using observations in
$I_k^c$. We then apply Algorithm~\ref{algorithm_2} to observations in $I_k$, replacing
$(\widehat \gamma,\widehat m)$ by the off-fold pilot estimators and drawing
$m_\eta^{k,s}$ based on the control units in $I_k$. The fold-specific
bias-corrected draws are
    \begin{equation}\label{post_adj_cf}
\check{\tau}_\eta^{k,s}
=
\tau_\eta^{k,s}
-
\frac{1}{n_k}
\sum_{i \in I_k}
\widehat{\gamma}^{(-k)}(D_i,X_i)
\big(\widehat m^{(-k)} - m_\eta^{k,s}\big)(X_i),
\qquad s=1,\dots,S,
    \end{equation}
where $\tau_\eta^{k,s}$ denotes the $k$-th fold version of $\tau_{\eta}^s$ in (\ref{NpBayes}) but using the adjusted prior.
    We then average across folds to obtain the output $\check{\tau}_\eta^s
    =
    \frac{1}{K}
    \sum_{k=1}^K
    \check{\tau}_\eta^{k,s}$ for $s=1,\ldots, S$. 
The prior specification is modified analogously by replacing $\widehat{\gamma}$ with $\widehat{\gamma}^{(-k)}$ and using fold-specific quantities (e.g., $n_{c,k}$= the number of control units in the $k$th fold). 
\end{remark}

\begin{remark}[Distinction with ATE]\label{rem:dist:ate}
With cross-sectional \textit{i.i.d.} data on $(Y_i,D_i,X_i)$, \cite{BLY2022} study the Bayesian inference for the ATE. The posterior of the ATE builds on
%\begin{equation*}
$	\int[m_{\eta}(1,x)-m_{\eta}(0,x)]\,\mathrm{d}F_{X,\eta}(x)$,
%\end{equation*}
where one assigns Gaussian process priors on the conditional means $(m_{\eta}(1,\cdot),m_{\eta}(0,\cdot))$ and places a Dirichlet process prior on $F_{X, \eta}(\cdot)$. An adoption of their framework for our analysis of the ATT would lead to an alternative Bayesian method based on
\begin{equation*}
	\tau_\eta=\frac{\int \pi_{\eta}(x)[m_{\eta}(1,x)-m_{\eta}(0,x)]\,\mathrm{d}F_{X,\eta}(x)}{\int \pi_{\eta}(x)\,\mathrm{d}F_{X,\eta}(x)},
\end{equation*}
which requires prior specification for each component of $(m_{\eta}(0,\cdot),m_{\eta}(1,\cdot),\pi_{\eta}(\cdot),F_{X, \eta}(\cdot))$. Fortunately, this is not necessary. The key observation of our current approach is that the last three components are all contained in $F_{\eta}(\cdot)$.  As a result, we do not need to specify them separately when analyzing the ATT.
\end{remark}
\begin{remark}[Posterior Recentering]\label{rem:debias:ate}
 A posterior debiasing step for the posterior is required by Theorem \ref{thm:BvM} and, as shown in Theorem \ref{thm:Debias}, our posterior correction  term in \eqref{recentering_term} indeed allows for a derivation of the BvM result under double-robust smoothness assumptions.
On the other hand, the posterior correction is not  required if one is willing to impose Donsker type smoothness conditions on the conditional mean function $m_\eta$, i.e., if the smoothness of $m_\eta$ exceeds $\dim(X_i)/2$, which is an implication of Corollary \ref{cor:BvM:single}. Posterior corrections were also proposed by \cite{BLY2022} in the context of average treatment effects (ATEs) using cross-sectional data. 
In their case, the bias correction term is given by $\widehat{b}^{ATE, s}_{\eta} = n^{-1} \sum_{i=1}^n \boldsymbol\tau\left[m_\eta^s-\widehat{m}\right]\left(Z_i\right)$, 
where $\boldsymbol{\tau}[m](z):=m(1, x)-m(0, x)+\widehat{\gamma}^{ATE}(d, x)(y-m(d, x))$. See also Remark \ref{rem:bias:ate} for a more explicit comparison of the biases in both cases. We observe that double-robust Bayesian inference for the ATT, as given in \eqref{recentering_term}, involves a simpler form of posterior correction compared to that for the ATE.
\end{remark}

\begin{remark}[Comparison with Frequentist Estimators]
Our approach is inspired by existing frequentist methods to conduct inference on the ATT. \cite{heckman1997matching} propose the following outcome imputed estimator for the ATT:
\begin{equation*}
	\widehat{\tau}_n=\frac{\sum_{i=1}^n D_i \big(\Delta Y_i-\widehat{m}(X_i)\big)}{\sum_{i=1}^n D_i},
\end{equation*}
where $\widehat{m}(\cdot)$ stands for the kernel smoothing estimator of the conditional mean in the control group. The doubly robust version from \cite{sant2020doubly} is
	\begin{equation*}
		\widehat{\tau}^{\text{DR}}_n=\frac{1}{n}\sum_{i=1}^n \widehat{\gamma}_n(D_i,X_i) \big(\Delta Y_i-\widehat{m}(X_i)\big),
	\end{equation*}
for some pilot estimators of the propensity score and the conditional mean of the control group. In contrast, our Bayesian estimator does not merely recenter the point estimator via the estimated Riesz representer $\widehat\gamma$; rather, it enters indirectly via prior and posterior adjustments.
\end{remark}

\section{Theory for Semiparametric Bayesian Outcome Regression} \label{sec:inference}
In this section, we establish a Bernstein-von Mises (BvM) Theorem using standard Gaussian process priors as considered in our Bayesian procedure in Algorithm \ref{algorithm_1}.

		\subsection{High-level Assumptions}\label{sec:high}
We now provide additional notations used for the derivation of our semiparametric Bernstein-von Mises Theorem. Recall that we restrict the joint density for the control arm only, imposing the exponential family restriction as in \eqref{condpdf}. We denote the observed data corresponding to the treated part as $Z_{\text{Treated}}^{(n)}:=(D_i\Delta Y_i, D_i, X_i^{\top})_{i=1}^n$. We express the posterior as follows:
\begin{align}\label{PosteriorFormula}
\Pi\left(m \in A, F \in B \mid Z^{(n)}\right)=\int_B \frac{\int_A \prod_{i=1}^n f_{\Delta Y \mid D, X}^{1-D_i}\left( \Delta Y_i \mid 0, X_i\right) \mathrm{d} \Pi(m)}{\int \prod_{i=1}^n f_{\Delta Y \mid D, X}^{1-D_i}\left( \Delta Y_i \mid 0, X_i\right) \mathrm{d} \Pi(m)} \,\mathrm{d} \Pi(F \mid Z_{\text {Treated }}^{(n)}),
\end{align}
where the conditional density $f_{\Delta Y|D,X}$ is a function of the conditional mean $m$ by the exponential family restriction given in \eqref{condpdf}.
Here, we used  the fact that independent priors are placed on the conditional mean $m$ and the distribution function $F$.

We first introduce high-level assumptions. Below, we consider some measurable sets $\mathcal H^m_n$ of functions $\eta^m$, which is understood only for the control arm such that $\Pi(\eta^m\in\mathcal{H}^m_n\mid Z^{(n)})\to_{P_0} 1$. We also denote $\mathcal{H}_n=\{\eta:\eta^m\in\mathcal{H}_n^m\}$ when we index the conditional mean function $m_{\eta}$ by its subscript $\eta$. We write $F_0:= F_{\eta_0}$ for the distribution of $(D\Delta Y, D, X^{\top})$ under $P_0$, and define $\|\phi\|_{2,F_0}:= \sqrt{\int \phi^2(z)\,\mathrm{d}F_0(z)}$ for all $\phi\in L^2(F_0):=\{\phi:\|\phi\|_{2,F_0}<\infty\}$. When we consider the conditional moment function $m$ below, the integral simplifies to one that depends only on the true marginal distribution of $X$, denoted by $F_X$ under $P_0$.
\begin{assumption}[Rates of Convergence]\label{Assump:NonDRRate}
Suppose there exist some measurable sets $\mathcal{H}_n$ such that $\sup_{\eta\in\mathcal{H}_n}\Vert m_\eta-m_0\Vert_{2,F_0}\leq \varepsilon_n$ for some $\varepsilon_n\to 0$ and $\Pi(\eta\in\mathcal{H}_n\mid Z^{(n)})\to_{P_0} 1$. 
\end{assumption}
The posterior contraction rate for the conditional mean can be derived by modifying the classical results of \cite{ghosal2000rates}. In the semiparametric Bayesian literature, the requirement $\varepsilon_n=o(n^{-1/4})$ is stated explicitly in order to eliminate second-order remainder terms; see Condition (C) in \cite{castillo2012gaussian}. This also aligns with the usual cut-off rate of the nonparametric components in frequentist semiparametric models \citep{newey1994var}. Note that the ATT $\tau_{\eta}$ is linear in $m_{\eta}$, so that we do not need to deal with these second-order terms. Nevertheless, the posterior contraction rate also plays a crucial role in the next two assumptions related to the stochastic equicontinuity and prior stability. For the concrete example involving the H\"older class for the conditional mean function, we need to impose sufficient smoothness so that this contraction rate indeed satisfies $\sqrt n\varepsilon^2_n=o(1)$. 

We adopt the standard empirical process notation as follows. For a function $h$ of a random vector $Z=(\Delta Y,D, X^\top)^\top$ that follows distribution $P_0$, we let $P_0[h]=\int h(z)\mathrm{d}P_0(z)$, $\mathbb{P}_n[h]=n^{-1}\sum_{i=1}^{n}h(Z_i)$, and $\mathbb{G}_n[h]=\sqrt n\left(\mathbb{P}_n-P_0\right)[h]$. The next set of assumptions restrict the complexity of the conditional mean functions. The first part requires the class $\{m_\eta:\eta\in\mathcal{H}_n\}$ to be Glivenko-Cantelli, plus some mild moment conditions on its envelope function. The second part imposes the stochastic equicontinuity, which holds when the conditional mean function belongs to a Donsker class. For the H\"older class considered in Section \ref{sec:sq:exp}, this enforces the sufficient smoothness of those functions relative to the dimensionality of covariates.
		\begin{assumption}[Complexity]\label{Assump:NonDRSE}
For measurable sets $\mathcal{H}_n$ in Assumption \ref{Assump:NonDRRate}, we assume that
(i) 
	$\sup _{\eta\in \mathcal{H}_n}\left|(\mathbb{P}_n-P_0) m_{\eta}\right| =o_{P_0}(1)$
and $\{m_\eta:\eta\in\mathcal{H}_n\}$ has an envelope function $M(\cdot)$ with  $P_0M^{2+\delta}<\infty$ for some constant $\delta>0$, and (ii)
			$\sup_{\eta\in\mathcal{H}_n}\left|\mathbb{G}_n\left[m_\eta-m_0\right]\right|=o_{P_0}(1)$.
	\end{assumption}
	
The next assumption concerns the prior stability condition, which is common to semiparametric Bayesian inference \citep{ghosal2017fundamentals}. This facilitates the technical proof for which we need to consider the perturbation along the least favorable direction. For that purpose, we introduce some necessary terminologies related to the general Gaussian process. Such a process determines the reproducing kernel Hilbert space (RKHS) $\left(\mathbb{H}^m,\|\cdot\|_{\mathbb{H}^m}\right)$. Our Bayesian method in Algorithm \ref{algorithm_1} does not include any correction involving the Riesz representer $\gamma_0$ as defined in \eqref{riesz:def}. Yet to establish prior stability, an approximation condition for $\gamma_0$ is imposed, requiring sufficient regularity of the propensity score $\pi_0(\cdot)$. We introduce the ball in $\mathbb{H}^m$ centered at the true Riesz representer $\gamma_0$ given by 
\begin{align*}
\mathbb{H}^m(r_n):=\left\{h \in \mathbb{H}^m : \Vert h-\gamma_0\Vert_{\infty}\leq r_n\text{ and }\Vert h\Vert_{\mathbb{H}^m}\leq \sqrt{n}r_n\right\}
\end{align*}
for some rate $r_n$, where $\Vert \cdot\Vert_{\infty}$ denotes the supremum norm. 
	\begin{assumption}[Prior Stability]\label{Assump:NonDRPS}
(i)	There exists $\overline\gamma_n\in\mathbb{H}^m(\zeta_n)$ for  a sequence $\zeta_n=o(1)$ with $\sqrt{n}\,\varepsilon_n\zeta_n=o(1)$ where $\varepsilon_n$ is the posterior contraction rate in Assumption \ref{Assump:NonDRRate}. (ii) Further, 
	$\Pi(\eta^m\in \mathcal{H}^m_n-t\overline\gamma_n n^{-1/2}|Z^{(n)})\to_{P_0} 1$ for every $t\in\mathbb{R}$.		
	\end{assumption}
 Based on this assumption, we provide the proof of this prior stability in Supplementary Appendix \ref{sec:GPStable}. Some comments follow the above technical conditions. For standard parametric models, the absolute continuity of the prior density suffices. However, for nonparametric priors, the very notion of a Radon-Nikodym density is non-trivial, and one needs to apply the Cameron-Martin theorem; see Proposition I.20 in \cite{ghosal2017fundamentals}. Assumption \ref{Assump:NonDRPS} (i) imposes an approximation condition to the Riesz representer $\gamma_0$ via the restriction $\overline\gamma_n\in\mathbb{H}^m(\zeta_n)$. Assumption \ref{Assump:NonDRPS} (ii) requires the posterior contraction to a shifted set $\mathcal{H}^m_n-t\overline\gamma_n n^{-1/2}$. Given the regularity in Assumption \ref{Assump:NonDRPS} (i), this set hardly differs from $\mathcal{H}^m_n$ for concrete Gaussian processes. See Lemma 5 for Riemann-Liouville process and Lemma 6 for finite Gaussian series prior in \cite{ray2020causal}, as well as our Lemma \ref{lemma:prior:stability} for the squared exponential process. In comparison, the prior correction in the doubly robust procedure weakens the requirement with the help of a pilot estimator of the propensity score. 
		\subsection{A Semiparametric BvM Theorem}
We now establish a Bernstein-von Mises Theorem for our Bayesian outcome regression method based on Gaussian process priors. When it comes to the centering point of the posterior, we consider an asymptotically efficient estimator $\widehat{\tau}$ with the following linear representation:
\begin{equation}\label{def:est:chi}
	\widehat{\tau}=\tau_0+\frac{1}{n}\sum_{i=1}^n \widetilde{\tau}_0(Z_i)+o_{P_0}(n^{-1/2}),
\end{equation}
where $\widetilde{\tau}_0=\widetilde{\tau}_{\eta_0}$ is the	efficient influence function given in \eqref{eif_att}. 
Below, we write $\mathcal{L}_{\Pi}(\sqrt{n}(\tau_\eta-\widehat{\tau})|Z^{(n)})$ for the marginal posterior law of $\sqrt{n}(\tau_\eta-\widehat{\tau})$.

The distributional convergence (conditional on the observed data) is established using the so called \textit{bounded Lipschitz distance}.
For two probability measures $P,Q$ defined on a metric space $\mathcal{Z}$, we define the bounded Lipschitz distance as
\begin{equation}
	d_{BL}(P,Q)=\sup_{f\in BL(1)}\left| \int_{\mathcal{Z}}f(\mathrm{d}P-\mathrm{d}Q)\right|,
\end{equation}
where 
$BL(1)=\left\{f:\mathcal{Z}\mapsto\mathbb{R}, \sup_{z\in\mathcal{Z}}|f(z)|+\sup_{z\neq z'}\frac{|f(z)-f(z')|}{\|z-z'\|_{\ell_2}}\leq 1 \right\}$.
Here,  $\|\cdot\|_{\ell_2}$ denotes the vector $\ell_2$ norm. 
Below is our main statement about the asymptotic behavior of the posterior distribution of $\tau_{\eta}$, that is derived from the Bayes rule given the prior specification and the observed data $Z^{(n)}$. 

\begin{theorem}\label{BvM:thm:standard}
	Let Assumptions \ref{Ass:unconfounded}--\ref{Assump:NonDRPS} hold. 
	Then, using  standard Gaussian process priors \eqref{prior:uc} on $\eta^m$ and an independent Dirichlet process prior on $F$, we have
	\begin{equation*}
		d_{BL}\left(\mathcal{L}_{\Pi}(\sqrt{n}(\tau_\eta-\widehat{\tau})\mid Z^{(n)}), N(0,\textsc v_0) \right)\to_{P_0} 0.
	\end{equation*}
	As a result, the Bayesian credible set $\mathcal{C}_n(\alpha)$ given in Section \ref{sec:method_outline} satisfies
	$P_0\big(\tau_0\in  \mathcal{C}_n(\alpha)\big) \to 1-\alpha$ for any $\alpha\in(0,1)$.
\end{theorem}

Theorem  \ref{BvM:thm:standard} establishes the BvM result for our Bayesian procedure using standard Gaussian process priors.
The entropy condition uniformly over $\eta\in\mathcal{H}_n$ is satisfied if $m_\eta$ is sufficiently smooth, that is, if $m_\eta$ belongs to a fixed $F_0$-Donsker class and, in particular, rules out double robustness. On the other hand, note that the asymptotic equivalence is obtained without any adjustment of prior or correction to posterior distributions, so the full Bayesian flavor is preserved. 

\subsection{Misspecified Likelihood and Posterior Contraction}\label{sec:misspecify}
 We now address possible misspecification of the control-arm likelihood in \eqref{condpdf} within our Bayesian framework. For the resulting quasi-posterior distribution, we establish contraction around the true conditional mean function.
 
 We assume a Gaussian \textit{working} likelihood of the conditional density function of outcome given covariates for the control units, but the true data generating distribution can differ from it. As argued by \cite{bissiri2016update}, one can start with a sample criterion function and take its exponentiated version as a working likelihood in Bayesian analysis. In our context, the key ingredient is the conditional mean function, which is naturally connected to the squared distance from the Gaussian likelihood. As we show in the sequel, this nonparametric Bayesian method possesses the same robustness to misspecification as frequentist M-estimation using least squares.

Referring to the (quasi-)posterior in (\ref{PosteriorFormula}), we specify the working model as $f_{\Delta Y \mid D, X}\left( \Delta Y \mid 0, X\right)=\frac{1}{\sqrt{2\pi\sigma^2}}\exp\left(-\left(\Delta Y-m_{\eta}(X)\right)^2/2\sigma^2\right)$ for some $\sigma>0$. Under the true DGP, we write the differenced potential outcome under control in the following regression form:
\begin{equation}\label{model:potential:control}
	\Delta Y(0)=m_0(X)+U\quad \text{and} \quad \E_0[U\mid D=0,X]=0,
\end{equation}
where $\Delta Y(0):=Y_2(0)-Y_1(0)$. In the next result, we will impose proper tail restrictions on the latent error term $U$, whose distribution may deviate from the normal case. Below, let $\mathcal{M}$ denote the class of uniformly bounded functions mapping $\mathcal{X}$ to $\mathbb{R}$. We further define $	N\left(\varepsilon,\mathcal{M},\Vert\cdot \Vert_{2,F_0}\right)$ as the covering number for any $\varepsilon>0$.

	\begin{lemma}\label{lemma:MisRate}
		Assume model \eqref{model:potential:control}, where $m_0\in\mathcal{M}$.  Suppose that the unobservable $U$ satisfies, for all $t\geq 0$,
		\begin{align}\label{ErrorMoment}
			\sup_{x\in\mathcal{X}} 
			P_0\left( |U|>t \mid D=0, X=x\right)
			\lesssim e^{-\underline{c} t^\beta}
		\end{align}	
		for some constants $\underline{c}>0$ and $\beta>1$.
		Let $\varepsilon_n$ be a sequence of positive numbers satisfying
$\varepsilon_n \to 0$ and $n\varepsilon_n^2 \to \infty$ such that
		\begin{align*}
\Pi\left(m_{\eta}\in\mathcal{M}:\Vert m_{\eta}-m_0\Vert_{2,F_0}^2\leq \varepsilon_n^2\right)
\geq e^{-c n\varepsilon_n^2}
\quad\text{and}\quad
N\left(\varepsilon_n,\mathcal{M},\Vert\cdot \Vert_{2,F_0}\right)
\leq e^{n\varepsilon_n^2}
		\end{align*}
for some 
constant $c>0$ and all $n$. Then we have
		\begin{align*}
			\Pi\left(
			m_{\eta}\in\mathcal{M}:\Vert m_{\eta}-m_0\Vert_{ 2,F_0}^2\geq M\varepsilon_n^2
			\mid Z^{(n)}
			\right)
			\to_{P_0}0 .
		\end{align*}
	for a sufficiently large constant $M>0$.
\end{lemma}
The above lemma serves as an important step in establishing the BvM theorem by providing the right posterior rate of contraction, stated as in the high-level Assumption \ref{Assump:NonDRRate}. For general misspecified models, the complication in nonparametric Bayesian analysis is due to the covering number for testing under misspecification \citep{kleijn2006misInf}. Tailored to economics applications, we relax the independent error assumption in \cite{kleijn2006misInf}. Given the posterior contraction rate, we conjecture the semiparametric BvM theorem remains valid for this quasi-posterior distribution under a proper set of assumptions\footnote{For parametric models, \cite{chernozhukov2003mcmc} demonstrated that the BvM theorem still guarantees frequentist coverage provided the generalized information identity holds. In our case, the asymptotic variance for the ATT coincides with the semiparametric efficiency bound without requiring the normality of the conditional distribution for the potential outcome in the control group.}. Appendix \ref{appendix:sensivity:error} presents simulation evidence that showcases the robustness of our procedure where the latent error term $U$ is non-normal and exhibits heteroskedasticity. We leave the theoretical development for future work.

\section{Theory for Doubly Robust Bayesian Method}\label{sec:asympt:prior}
This section establishes the main theory for our doubly robust Bayesian procedure. We build on double robustness within the modern machine learning framework \citep{chernozhukov2017double,benkeser2017doubly}: rather than merely hedging against parametric misspecification, it represents the ability to trade off estimation accuracy between nuisance functions\footnote{A fundamental reason why this robustness applies to the ATT stems from the mixed-bias property of the target functional, as studied by \cite{rotnitzky2021bias}.}.

For notational clarity, we focus on the auxiliary-data-based procedure in Algorithm \ref{algorithm_2}. To use the full data, we recommend the $K$-fold cross-fitted version described in Remark \ref{rmk:cross_fit}. Appendix \ref{appendix:crossfit} provides a complete description of this version and shows that its validity hinges on the fold-specific BvM result established in Theorem \ref{thm:Debias}.
 
\subsection{High-level Assumptions}\label{sec:dr_high}
Below, we present the assumptions that deliver double-robust inference through adjustments to the prior and posterior distributions.
	\begin{assumption}[DR Rates of Convergence]\label{Assump:Rate}
The estimators $\widehat \pi$ and $\widehat m$, which are based on an auxiliary sample independent of $Z^{(n)}$, satisfy $\Vert \widehat{\pi}-\pi_0\Vert_{2, F_0}=O_{P_0}(r_n) $, 
		\begin{equation*}
\Vert \widehat{m}-m_0\Vert_{ 2,F_0}=O_{P_0}(\varepsilon_n), ~\text{ and }~\sup_{\eta\in\mathcal{H}_n}\Vert m_\eta-m_0\Vert_{2,F_0}\leq \varepsilon_n,
	\end{equation*}
	where $\max\{\varepsilon_n, r_n\}\to 0$, $\sqrt{n}\varepsilon_nr_n\to 0$ and the measurable sets $\mathcal{H}_n$ satisfy $\Pi(\eta\in\mathcal{H}_n\mid Z^{(n)})\to_{P_0} 1$. Further, $\Vert \widehat\gamma\Vert_{\infty}=O_{P_0}(1)$ and $\widehat{\bar{\pi}}-\bar{\pi}_0=O_{P_0}(n^{-1/2})$. 
	\end{assumption}	
	Assumption \ref{Assump:Rate} imposes sufficiently fast convergence rates for the estimators for the conditional mean function $m_0$ and the propensity score $\pi_0$. 
 The posterior convergence rate for the conditional mean can be derived by modifying the classical results of \cite{ghosal2000rates} by accommodating the propensity score-adjusted prior, in the same spirit of \cite{ray2020causal}. 
 We refer to \cite{BLY2022} who showed that this assumption allows for double robustness under H\"older type smoothness assumptions. 
	\begin{assumption}[DR Stochastic Equicontinuity] \label{Assump:Donsker}
$		\sup_{\eta\in\mathcal{H}_n}\left|\mathbb{G}_n\left[\left(\gamma_0-\widehat \gamma\right)(m_\eta-m_0)\right]\right|=o_{P_0}(1).$
	\end{assumption}

	Assumption \ref{Assump:Donsker} restricts the functional class $\mathcal{H}_n$ to form a $P_0$-Glivenko-Cantelli class; see Section 2.4 of \cite{van1996empirical} and imposes a stochastic equicontinuity condition on a  product structure involving $\widehat\gamma$ and $m_\eta$. 
	Hence, the complexity of the functional class $(m_{\eta}-m_0)$ can be compensated by certain high regularity of the corresponding Riesz representer and vice versa.
	This condition adapts the complexity requirement of \cite{BLY2022} by only restricting the control arm.

Recall the propensity score-dependent prior on $m$ given in \eqref{prior:ps}, i.e., $m(\cdot) = q^{-1}\left(W^m(\cdot) + \lambda\widehat \gamma(\cdot)\right)$. Below, we restrict the behavior for $\lambda$ through its hyperparameter $\varsigma_n>0$. For two sequences $\{a_n\}$ and $\{b_n\}$ of positive numbers, we write $a_n \lesssim b_n$ if $\limsup_{n\to\infty} (a_n / b_n)<\infty$, and $a_n \sim b_n$ if $a_n \lesssim b_n$ and $b_n \lesssim a_n$.

	\begin{assumption}[DR Prior Stability]\label{Assump:Prior}
		$W^m$ is a continuous stochastic process independent of the normal random variable $\lambda\sim N(0,\varsigma_n^2)$, where $\varsigma_n\lesssim 1$, $n\varsigma^2_{n}\to\infty$ and that satisfies (i)
$		\Pi\left(\lambda:|\lambda|\leq u_n\varsigma_n^2\sqrt{n}\mid Z^{(n)}\right)\to_{P_0}1$, 	for some deterministic sequence $u_n\to 0$	and (ii) 
$		\Pi\left((w,\lambda):w+(\lambda+tn^{-1/2})\widehat{\gamma}\in\mathcal{H}_n^m\mid Z^{(n)} \right)\to_{P_0}1$	 for any $t\in\mathbb R$. 
	\end{assumption}

		 Assumption \ref{Assump:Prior} incorporates Conditions (3.9) and (3.10) from Theorem 2 in \cite{ray2020causal}, and it is imposed to establish the stability property of the adjusted prior distribution. 
		 We will provide sufficient conditions for Assumption \ref{Assump:Prior} in Section \ref{sec:sq:exp}. 
		\subsection{Doubly robust BvM Theorems}
	We now establish a semiparametric BvM theorem for our doubly robust Bayesian procedure given in Algorithm \ref{algorithm_2}.
To signify the dependence of the posterior law on the pilot estimation of the Riesz representer, we write the posterior as $\Pi_{\widehat{\gamma}}(\cdot|Z^{(n)})$ in the sequel.
	\begin{theorem}\label{thm:BvM}
		Let Assumptions \ref{Ass:unconfounded}, \ref{Assump:NonDRSE}(i), \ref{Assump:Rate}, \ref{Assump:Donsker}, and \ref{Assump:Prior} hold. Consider the propensity score adjusted prior \eqref{prior:ps} on $\eta^m$ and an independent Dirichlet process prior on $F$. Then we have
	\begin{equation*}
		d_{BL}\left(\mathcal{L}_{\Pi_{\widehat{\gamma}}}(\sqrt{n}\left((\tau_{\eta}-\widehat{\tau})-b_{0,\eta}\right)\mid Z^{(n)}), N(0,\textsc v_0) \right)\to_{P_0} 0,
	\end{equation*}
		where $	b_{0,\eta}:=	\mathbb{P}_n[(m_0-m_{\eta})\gamma_0]$. 
	\end{theorem}
	Theorem \ref{thm:BvM} shows that, under double-robust smoothness conditions, the BvM theorem holds only up to a ``bias term" $b_{0,\eta}$, which depends on the unknown conditional mean $m_0$. This biased posterior makes the BvM infeasible in practice. We also emphasize that the derivation of this result is different from the results in \cite{BLY2022}, as we need to control the denominator in the asymptotic expansions.
	
\begin{remark}[Comparison of Bias in ATE/ATT Posteriors]\label{rem:bias:ate}
\cite{BLY2022} showed that, for inference on the ATE in the cross-sectional case, the BvM holds for a biased posterior under double-robust smoothness conditions, see also Remark \ref{rem:debias:ate}. This ``bias term'' is closely related to the influence function of the ATE, which takes the following form
\begin{align*}
 b_{0,\eta}^{ATE}
%&=\mathbb{E}_0\left[\gamma^{ATE}_0(D,X)
%\left(m_0(D, X)-m_\eta(D, X)\right)\right] -\mathbb{E}_0\left[\bar{m}_0(X) - \bar{m}_\eta(X)\right]\\
&=\frac{1}{n}\sum_{i=1}^n \Bigg\{\Bigg(\underbrace{\frac{D_i}{\pi_0(X_i)}-\frac{1-D_i}{1-\pi_0(X_i)}}_{=:\gamma_0^{ATE}(D_i,X_i)}\Bigg)
\left(m_0(D_i, X_i)-m_\eta(D_i, X_i)\right) -(\bar{m}_0(X_i) - \bar{m}_\eta(X_i))\Bigg\},
\end{align*}
where $\bar{m}_0(\cdot)=m_0(1, \cdot)-m_0(0, \cdot)$,  $\bar{m}_\eta(\cdot)=m_\eta(1, \cdot)-m_\eta(0, \cdot)$, and the Riesz representer $\gamma_0^{ATE}$ as given in the ATE case, see \cite{BLY2022}. Referring to the influence function of the ATT, we can also express it in terms of the conditional mean $m_{0}(D,X)$ involving both treated and control groups, cf. Equation (8.5) in \cite{vanderLaan2011tl}. 
Therefore, we have the following expression for the bias term in the ATT case:
\begin{align*}
	&\frac{1}{n}\sum_{i=1}^n \Bigg\{\Bigg(\underbrace{\frac{D_i}{\bar\pi_0}-\frac{1-D_i}{\bar\pi_0}\frac{\pi_0(X_i)}{1-\pi_0(X_i)}}_{=\gamma_0(D_i,X_i)}\Bigg)
	\left(m_0(D_i, X_i)-m_\eta(D_i, X_i)\right) -\frac{D_i}{\bar\pi_0}\left(\bar{m}_0(X_i) - \bar{m}_\eta(X_i)\right)\Bigg\}\\
	&=	\frac{1}{n}\sum_{i=1}^n\gamma_0(D_i,X_i)
	\left(m_0(0, X_i)-m_\eta(0, X_i)\right)= 	b_{0,\eta}^{ATT},
\end{align*}
where the simplification occurs because the term $(D_i/\bar\pi_0)(m_0(1,X_i)-m_{\eta}(1,X_i))$ cancels out in the difference. The resulting simplification of the bias term aligns with our simulation results, which show that standard Gaussian process priors also provide accurate coverage for the ATT in many cases.
\end{remark}		
	
The next result is an immediate implication of Theorem \ref{thm:BvM}. Specifically, it provides a BvM Theorem for Bayesian procedures that do not rely on posterior correction. This can be achieved if the bias term is asymptotically negligible uniformly over the underlying functional class, which requires more restrictive  smoothness conditions on the conditional mean function $m_0$. 
	\begin{corollary}\label{cor:BvM:single}
	Let Assumptions \ref{Ass:unconfounded}, \ref{Assump:NonDRSE}(i), \ref{Assump:Rate}, \ref{Assump:Donsker}, and \ref{Assump:Prior} hold.  Consider  the propensity score adjusted prior \eqref{prior:ps} on $\eta^m$ and an independent Dirichlet process prior on $F$. If, in addition, $	b_{0,\eta}=o_{P_0}(n^{-1/2})$ uniformly for $\eta\in\mathcal{H}_n$,  then we have 
	\begin{equation*}
		d_{BL}\left(\mathcal{L}_{\Pi_{\widehat{\gamma}}}(\sqrt{n}(\tau_{\eta}-\widehat{\tau})\mid Z^{(n)}), N(0,\textsc v_0) \right)\to_{P_0} 0.
		\end{equation*}
\end{corollary}
While Corollary \ref{cor:BvM:single} allows for arbitrarily low regularity of propensity scores, it requires the conditional mean function to be sufficiently smooth; specifically, the smoothness of $m$ must be greater than $\dim(X_i)/2$ (also referred to as the Donsker property). This condition is also called \textit{single robustness} by \cite{ray2020causal}, and indeed, this corollary extends their findings to the inference on the ATT. Also, as they point out, propensity score adjusted priors \eqref{prior:ps} relax the uniformity condition $\sup_{\eta\in\mathcal H_n}|\mathbb{G}_n\left[m_\eta-m_0\right]|=o_{P_0}(1)$ used in Theorem \ref{BvM:thm:standard} under standard Gaussian process priors.

Under double-robust assumptions,  however,  the Bayesian procedure that achieves the BvM equivalence in Theorem \ref{thm:BvM} is not feasible, because it depends on the term $b_{0,\eta}$, which is a function of the unknown conditional mean $m_0$.
Our objective is to maintain double-robust conditions, while considering pilot estimators for the unknown functional parameters in $b_{0,\eta}$. The correction term $\widehat{b}_{\eta}$, as introduced in \eqref{recentering_term}, results in a feasible Bayesian procedure that satisfies the BvM theorem, as demonstrated below.
	\begin{theorem}\label{thm:Debias}
Let Assumptions \ref{Ass:unconfounded}, \ref{Assump:NonDRSE}(i), \ref{Assump:Rate}, \ref{Assump:Donsker}, and \ref{Assump:Prior} hold.  Consider  the propensity score adjusted prior \eqref{prior:ps} on $\eta^m$ and an independent Dirichlet process prior on $F$. 	Then we have
	\begin{equation*}
		d_{BL}\left(\mathcal{L}_{\Pi_{\widehat{\gamma}}}(\sqrt{n}(\tau_{\eta}-\widehat{\tau} - 	\widehat b_\eta)\mid Z^{(n)}), N(0,\textsc v_0) \right)\to_{P_0} 0,
	\end{equation*}
	where $	\widehat b_\eta =	\mathbb{P}_n[(\widehat m-m_{\eta})\widehat \gamma]$. As a result, the Bayesian credible set $\mathcal{C}_n^{DR}(\alpha)$ given in Section \ref{sec:method_outline:DR} satisfies	$P_0\big(\tau_0\in  \mathcal{C}_n^{DR}(\alpha)\big) \to 1-\alpha$ for any $\alpha\in(0,1)$.
\end{theorem}

Theorem \ref{thm:Debias} shows that the Bayesian method proposed in Algorithm \ref{algorithm_2}, $\check\tau_\eta = \tau_{\eta} - \widehat{b}_\eta$, achieves the BvM result under double-robust smoothness conditions. The following remark clarifies the relationship when considering posterior correction alone, in which case BvM results are available only under more restrictive smoothness assumptions on the propensity score and the conditional mean function.

\begin{remark}\label{remarkYiu}
Building on the idea of a one-step update in frequentist semiparametric estimation, \cite{yiu2023corrected} propose a different method of posterior correction (without prior adjustment) that involves the efficient influence function. When applying their methodology to the ATT, it is evident that both the conditional mean function and the propensity score must satisfy the Donsker property, cf. Assumption 4(c) therein. In contrast,  the relaxation of the Donsker property is one of the key technical innovations of our doubly robust Bayesian inference.
\end{remark}

The $K$-fold cross-fitted Bayesian procedure proposed in  Remark \ref{rmk:cross_fit}  is designed to make efficient use of the entire sample. A complete description of the algorithm is given in Appendix \ref{appendix:crossfit}. Our construction is inspired by the fold-specific multiplier bootstrap procedure of \cite{bach2024R}, which uses pilot estimators computed from the remaining folds. The key difference is that uncertainty propagation in our approach is primarily induced by the posterior distribution for each fold. Nonetheless, Theorem \ref{thm:BvMCF} shows that the validity of this cross-fitted method mainly relies on the BvM result for the fold-specific posterior distribution, established in Theorem \ref{thm:Debias}.

\section{Finite Sample Results}\label{sec:numerical}
This section investigates the finite-sample performance of the proposed Bayesian approaches and then applies them to a real dataset.
\subsection{Simulation Evidence}\label{sec:simulations}
We present Monte Carlo simulation results comparing our proposed semiparametric Bayesian methods with existing frequentist approaches. We adopt the four DGPs of \cite{sant2020doubly} and additionally consider a higher-dimensional DGP.  In each DGP, we generate samples of size $n \in \{1000, 2000\}$ of the observable variables $(Y_{1}, Y_{2}, D, X^\top)$ as follows:
Let $V:=\left(V_{1}, V_{2}, V_{3}, V_{4}\right)^\top \sim \mathcal{N}(0, I_4)$,  where $I_p$ denotes the $p$-dimensional identity matrix.  The covariate vector $X:=\left(X_{1}, X_{2}, X_{3}, X_{4}\right)^\top$ is constructed from $V$ as follows: 
$X_{j} = (\tilde{X}_{j} - \mathbb{E}[\tilde{X}_{j}])/\sqrt{\mathrm{Var}[\tilde{X}_{j}]}$,
where
$\tilde{X}_{1} = \exp(V_{1}/2)$,
$\tilde{X}_{2} = 10 + V_{2}/(1 + \exp(V_{1}/2))$,
$\tilde{X}_{3} = (0.6 + V_{1}V_{3}/25)^3$,
and $\tilde{X}_{4} = (20 + V_{2} + V_{4})^2$.
For a vector $w = (w_1, w_2, w_3, w_4)^\top$,  define
\begin{equation*}
  g(w) = 0.75(-w_1 + 0.5w_2 - 0.25w_3 - 0.1w_4),
  \quad
  \mu(w) = 210 + 27.4w_1 + 13.7(w_2 + w_3 + w_4).
\end{equation*}
Let  $\Psi(t) = 1/(1 + e^{-t})$ and $(\alpha, \epsilon_{1}, \epsilon_{2}(0), \epsilon_{2}(1))^\top
\sim \mathcal{N}(0, I_4)$.  
For $d \in \{0,1\}$, consider the following four DGPs, which generate the treatment indicator $D$ and the outcomes $Y_{1}$ and $Y_{2}=D Y_{2}(1) + (1 - D)Y_{2}(0)$:
\begin{enumerate}[leftmargin=0.1in]
 \item[]
\[
\begin{aligned}
\text{DGP1: } \quad
&D\sim \mathrm{Bernoulli}\left(\Psi[g(X)]\right),  \quad
Y_{1} = \mu(X) + D\mu(X) + \alpha + \epsilon_{1}, \\
&Y_{2}(d)= 2\mu(X) + D\mu(X) + \alpha + \epsilon_{2}(d).
\end{aligned}
\]
\item[]
 \[
\begin{aligned}
\text{DGP2: } \quad
&D\sim \mathrm{Bernoulli}\left(\Psi[g(V)]\right),  \quad
Y_{1} = \mu(X) + D\mu(X) + \alpha + \epsilon_{1}, \\
&Y_{2}(d)= 2\mu(X) + D\mu(X) + \alpha + \epsilon_{2}(d).
\end{aligned}
\]
\item[]
\[
\begin{aligned}
\text{DGP3: } \quad
&D\sim \mathrm{Bernoulli}\left(\Psi[g(X)]\right),  \quad
Y_{1} = \mu(V) + D\mu(V) + \alpha+ \epsilon_{1}, \\
&Y_{2}(d)= 2\mu(V) + D\mu(V) + \alpha + \epsilon_{2}(d).
\end{aligned}
\]
\item[]
\[
\begin{aligned}
\text{DGP4: } \quad
&D\sim \mathrm{Bernoulli}\left(\Psi[g(V)]\right),  \quad
Y_{1} = \mu(V) + D\mu(V) + \alpha + \epsilon_{1}, \\
&Y_{2}(d)= 2\mu(V) + D\mu(V) + \alpha+ \epsilon_{2}(d).
\end{aligned}
\]
\end{enumerate}
 The true ATT is zero in these DGPs. All the simulation results reported here are based on 1000 Monte Carlo replications for each design.

Our semiparametric Bayesian outcome regression (\textbf{OR Bayes}) and doubly robust Bayesian
(\textbf{DR Bayes}) methods follow Algorithms~\ref{algorithm_1}
and~\ref{algorithm_2}, implemented via the MATLAB package \texttt{GPML}
with $5000$ posterior draws.
DR Bayes is implemented in two variants according to the pilot estimator used for
the propensity score: logistic regression
(\textbf{DR Bayes$^{\text{Logit}}$}) or random forest
(\textbf{DR Bayes$^{\text{RF}}$}).
Although the BvM result for DR Bayes relies on sample splitting, we report the finite-sample performance with and without it. The default DR Bayes is  the sample-split version that 
uses $K$-fold cross-fitting described in Remark \ref{rmk:cross_fit} with $K=5$.  The version labeled with \textbf{FS} does not split the sample,  and instead uses the full sample for both nuisance estimation and posterior computation.

We compare the Bayesian methods with a variety of frequentist DiD
estimators.
\textbf{DR} is the improved doubly robust DiD estimator of
\cite{sant2020doubly}.
\textbf{OR} (outcome regression) is the sample analog of~(\ref{att_0})
with $m_0$ estimated by OLS on the control arm.  \textbf{IPW} represents the H\'{a}jek-type IPW DiD estimator.\footnote{Its expression is given by equation~(4.1) of \cite{sant2020doubly}.  The Horvitz--Thompson-type IPW DiD estimator  \citep{abadie2005} leads to confidence intervals twice as long as the H\'{a}jek-type,  so it is omitted in the following tables.  }\textsuperscript{,}\footnote{As \cite{sant2020doubly} documented,  the standard two-way fixed effects estimator performs very badly in all DGPs because the time trends are correlated with covariates. Thus, our tables omit it. } \textbf{DML} represents the double/debiased machine learning ATT estimator
\citep{chernozhukov2017double,chang2020double}, with $\pi_0(x)$ estimated by random forest (RF) and $m_0(x)$ estimated by random forests,  neural nets (NN), or Gaussian process regression (GP).\footnote{Frequentist DiD estimators other than DML are implemented
using the R package \texttt{DRDID}; DML estimators use the R package
\texttt{DoubleML}.  Frequentist confidence intervals are computed using analytical (influence-function-based) standard errors.   We also experimented with DML estimators that use Logit to estimate $\pi_0(x)$ and find that they perform similarly to their counterparts using RF as the propensity score estimator in DGPs 1 to 3,  but with smaller bias.  In DGP 4,  however, they perform considerably worse than their RF-based counterparts,  exhibiting large bias and substantial undercoverage.  We therefore omit them in the tables to save space.} All tables in this section report the finite-sample bias, RMSE, coverage probability (CP), and average interval length (CIL) of 95\% credible/confidence intervals.
 
 % ------------------------------------------------------------------
%  TABLE 1
% ------------------------------------------------------------------
\begin{table}[H]
\centering
\caption{Simulation results under DGPs~1 to~4.  Sample size $n=1000$.}
\vskip.15cm
{\small 
\label{tab:simu_n1K}
\setlength{\tabcolsep}{3pt}
\renewcommand{\arraystretch}{0.75}
\begin{tabular}{l cccc| cccc}
\toprule
Method
  & \multicolumn{4}{c}{DGP 1}
  & \multicolumn{4}{c}{DGP 2} \\
\cmidrule(lr){2-5}\cmidrule(lr){6-9}
  & Bias & RMSE & CP & CIL
  & Bias & RMSE & CP & CIL \\
\midrule
OR Bayes
  & $-0.002$ & $0.104$ & $0.951$ & $0.406$
  & $-0.000$ & $0.105$ & $0.950$ & $0.407$ \\
DR Bayes$^{\text{Logit}}$
  & $-0.003$ & $0.111$ & $0.955$ & $0.447$
  & $-0.000$ & $0.110$ & $0.954$ & $0.433$ \\
DR Bayes$^{\text{RF}}$
  & $-0.002$ & $0.147$ & $0.965$ & $0.596$
  & $-0.001$ & $0.148$ & $0.961$ & $0.611$ \\
  DR Bayes$^{\text{Logit}}$(FS)
  & $-0.003$ & $0.108$ & $0.945$ & $0.420$
  & $0.000$ & $0.108$ & $0.942$ & $0.405$ \\
DR Bayes$^{\text{RF}}$(FS)
  & $-0.002$ & $0.147$ & $0.971$ & $0.591$
  & $-0.001$ & $0.143$ & $0.961$ & $0.598$ \\
\cline{1-9}
DR
  & $-0.003$ & $0.107$ & $0.941$ & $0.407$
  & $-0.000$ & $0.107$ & $0.944$ & $0.403$ \\
OR
  & $-0.002$ & $0.102$ & $0.947$ & $0.395$
  & $-0.001$ & $0.103$ & $0.938$ & $0.393$ \\
IPW
  & $-0.040$ & $1.140$ & $0.949$ & $4.287$
  & $-0.784$ & $1.244$ & $0.838$ & $3.620$ \\
DML$^{\text{RF-RF}}$
  & $-0.160$ & $0.719$ & $0.999$ & $4.417$
  & $-0.480$ & $0.941$ & $0.992$ & $4.846$ \\
DML$^{\text{RF-NN}}$
  & $-0.017$ & $0.153$ & $0.962$ & $0.641$
  & $-0.027$ & $0.197$ & $0.951$ & $0.737$ \\
DML$^{\text{RF-GP}}$
  & $-0.338$ & $0.946$ & $0.971$ & $4.679$
  & $-0.108$ & $1.022$ & $0.989$ & $5.688$ \\
\midrule
Method
  & \multicolumn{4}{c}{DGP 3}
  & \multicolumn{4}{c}{DGP 4} \\
\cmidrule(lr){2-5}\cmidrule(lr){6-9}
  & Bias & RMSE & CP & CIL
  & Bias & RMSE & CP & CIL \\
\midrule
OR Bayes
  & $-0.053$ & $0.408$ & $0.993$ & $2.257$
  & $-0.558$ & $0.745$ & $0.981$ & $3.242$ \\
DR Bayes$^{\text{Logit}}$
  & $-0.012$ & $0.779$ & $0.967$ & $3.313$
  & $-1.037$ & $1.355$ & $0.777$ & $3.390$ \\
DR Bayes$^{\text{RF}}$
  & $-0.001$ & $1.068$ & $0.972$ & $3.951$
  & $-0.688$ & $1.441$ & $0.915$ & $4.195$ \\
  DR Bayes$^{\text{Logit}}$(FS)
  & $-0.057$ & $0.409$ & $0.967$ & $1.783$
  & $-0.553$ & $0.741$ & $0.832$ & $2.084$ \\
DR Bayes$^{\text{RF}}$(FS)
  & $-0.052$ & $0.437$ & $0.987$ & $2.327$
  & $-0.494$ & $0.705$ & $0.929$ & $2.629$ \\
\cline{1-9}
DR
  & $-0.060$ & $0.997$ & $0.944$ & $3.856$
  & $-2.506$ & $2.709$ & $0.288$ & $3.850$ \\
OR
  & $-1.386$ & $1.848$ & $0.811$ & $4.812$
  & $-5.144$ & $5.305$ & $0.018$ & $5.037$ \\
IPW
  & $-0.011$ & $1.439$ & $0.944$ & $5.499$
  & $-3.885$ & $4.182$ & $0.249$ & $5.733$ \\
DML$^{\text{RF-RF}}$
  & $-0.278$ & $0.831$ & $0.991$ & $4.470$
  & $-0.853$ & $1.204$ & $0.964$ & $5.029$ \\
DML$^{\text{RF-NN}}$
  & $-0.066$ & $0.822$ & $0.985$ & $3.924$
  & $-1.043$ & $1.463$ & $0.848$ & $4.555$ \\
DML$^{\text{RF-GP}}$
  & $-0.256$ & $1.011$ & $0.992$ & $5.313$
  & $-0.912$ & $1.597$ & $0.942$ & $7.274$ \\
\bottomrule
\end{tabular}}
\end{table}

According to Tables~\ref{tab:simu_n1K} and~\ref{tab:simu_n2K}, 
in DGPs~1 and~2, where the conditional mean $m_0$ is linear, both OR Bayes and
DR Bayes$^{\text{Logit}}$ achieve near-nominal coverage and short
intervals, performing comparably to the correctly specified parametric OR and DR estimators.
DR Bayes$^{\text{RF}}$ also attains correct coverage,  but yields intervals wider than its logit-based counterpart.  The IPW estimators exhibit large RMSEs/CILs and, for DGP~2, substantially
degraded coverage due to propensity score misspecification.
 DML variants with RF or GP for $m_0$ produce CILs considerably wider  than the Bayesian methods despite similar or lower coverage; DML with NN is more competitive and performs similarly to DR Bayes$^{\text{RF}}$.

% ------------------------------------------------------------------
%  TABLE 2
% ------------------------------------------------------------------
\begin{table}[H]
\centering
\caption{Simulation results under DGPs~1 to~4.  Sample size $n=2000$.}
\vskip.15cm
{\small 
\label{tab:simu_n2K}
\setlength{\tabcolsep}{3pt}
\renewcommand{\arraystretch}{0.75}
\begin{tabular}{l cccc| cccc}
\toprule
Method
  & \multicolumn{4}{c}{DGP 1}
  & \multicolumn{4}{c}{DGP 2} \\
\cmidrule(lr){2-5}\cmidrule(lr){6-9}
  & Bias & RMSE & CP & CIL
  & Bias & RMSE & CP & CIL \\
\midrule
OR Bayes
  & $-0.004$ & $0.072$ & $0.949$ & $0.287$
  & $-0.002$ & $0.071$ & $0.956$ & $0.288$ \\
DR Bayes$^{\text{Logit}}$
  & $-0.004$ & $0.075$ & $0.948$ & $0.302$
  & $-0.002$ & $0.073$ & $0.952$ & $0.291$ \\
DR Bayes$^{\text{RF}}$
  & $-0.004$ & $0.106$ & $0.973$ & $0.405$
  & $-0.003$ & $0.097$ & $0.972$ & $0.414$ \\
  DR Bayes$^{\text{Logit}}$(FS)
  & $-0.004$ & $0.074$ & $0.947$ & $0.295$
  & $-0.002$ & $0.073$ & $0.949$ & $0.285$ \\
DR Bayes$^{\text{RF}}$(FS)
  & $-0.003$ & $0.098$ & $0.974$ & $0.408$
  & $-0.004$ & $0.097$ & $0.970$ & $0.416$ \\
\cline{1-9}
DR
  & $-0.004$ & $0.074$ & $0.949$ & $0.290$
  & $-0.002$ & $0.072$ & $0.948$ & $0.287$ \\
OR
  & $-0.003$ & $0.071$ & $0.951$ & $0.280$
  & $-0.002$ & $0.070$ & $0.953$ & $0.278$ \\
IPW
  & $-0.008$ & $0.821$ & $0.944$ & $3.067$
  & $-0.809$ & $1.045$ & $0.754$ & $2.551$ \\
DML$^{\text{RF-RF}}$
  & $-0.244$ & $0.482$ & $0.998$ & $2.913$
  & $-0.434$ & $0.648$ & $0.989$ & $3.276$ \\
DML$^{\text{RF-NN}}$
  & $-0.001$ & $0.088$ & $0.959$ & $0.351$
  & $-0.008$ & $0.102$ & $0.961$ & $0.400$ \\
DML$^{\text{RF-GP}}$
  & $-0.402$ & $0.621$ & $0.941$ & $2.554$
  & $-0.052$ & $0.531$ & $0.993$ & $3.173$ \\
\midrule
Method
  & \multicolumn{4}{c}{DGP 3}
  & \multicolumn{4}{c}{DGP 4} \\
\cmidrule(lr){2-5}\cmidrule(lr){6-9}
  & Bias & RMSE & CP & CIL
  & Bias & RMSE & CP & CIL \\
\midrule
OR Bayes
  & $-0.001$ & $0.234$ & $0.996$ & $1.393$
  & $-0.311$ & $0.436$ & $0.993$ & $2.360$ \\
DR Bayes$^{\text{Logit}}$
  & $-0.016$ & $0.370$ & $0.978$ & $1.759$
  & $-0.598$ & $0.723$ & $0.804$ & $1.912$ \\
DR Bayes$^{\text{RF}}$
  & $-0.013$ & $0.464$ & $0.989$ & $2.123$
  & $-0.412$ & $0.697$ & $0.937$ & $2.339$ \\
  DR Bayes$^{\text{Logit}}$(FS)
  & $-0.003$ & $0.234$ & $0.968$ & $1.036$
  & $-0.309$ & $0.434$ & $0.886$ & $1.341$ \\
DR Bayes$^{\text{RF}}$(FS) 
  &$-0.004$ & $0.242$ & $0.991$ & $1.302$
  & $-0.284$ & $0.416$ & $0.961$ & $1.599$ \\
\cline{1-9}
DR
  & $-0.030$ & $0.693$ & $0.959$ & $2.739$
  & $-2.482$ & $2.576$ & $0.053$ & $2.734$ \\
OR
  & $-1.347$ & $1.604$ & $0.657$ & $3.415$
  & $-5.149$ & $5.227$ & $0.000$ & $3.579$ \\
IPW
  & $-0.009$ & $1.006$ & $0.942$ & $3.860$
  & $-3.911$ & $4.045$ & $0.040$ & $4.014$ \\
DML$^{\text{RF-RF}}$
  & $-0.316$ & $0.601$ & $0.987$ & $2.923$
  & $-0.922$ & $1.067$ & $0.871$ & $3.193$ \\
DML$^{\text{RF-NN}}$
  & $-0.005$ & $0.492$ & $0.982$ & $2.322$
  & $-0.942$ & $1.126$ & $0.733$ & $2.854$ \\
DML$^{\text{RF-GP}}$
  & $-0.090$ & $0.589$ & $0.989$ & $3.056$
  & $-1.070$ & $1.325$ & $0.856$ & $4.256$ \\
\bottomrule
\end{tabular}}
\end{table}

In DGP~3 where $m_0$ is nonlinear but the index of propensity score $\Psi^{-1}(\pi_0)$ is linear,  Bayesian methods with and without DR adjustment continue to deliver desirable coverage,  though with larger CILs.  Parametric frequentist DR leads to a similar coverage performance,  with CIL comparable to DR Bayes$^{\text{RF}}$ and larger than OR Bayes and DR Bayes$^{\text{Logit}}$.  Parametric frequentist OR is markedly biased  and  undercovers,  confirming the cost of misspecifying $m_0$.  DML estimators lead to correct coverage.  Among them,  those using NN to estimate $m_0$ are more precise than those using RF- or GP-based estimators.
When it comes to DGP~4 where both nuisance functions are nonlinear,  parametric DR becomes inconsistent and undercovers severely. OR Bayes, which models $m_0$ nonparametrically, maintains high coverage.
DR Bayes$^{\text{Logit}}$ undercovers because the logistic linear propensity score model
is misspecified, yet it still substantially outperforms parametric DR.
Using the more flexible random forest for $\pi_0$, DR Bayes$^{\text{RF}}$
restores the coverage especially when $n = 2000$.
Among DML variants,  DML$^{\text{RF-RF}}$ and DML$^{\text{RF-GP}}$  are competitive in coverage
but they (especially the latter) yield substantially wider intervals than Bayesian methods.  
Tables~\ref{tab:simu_n1K} and~\ref{tab:simu_n2K} also show that DR Bayes with and without sample-split deliver comparable results. 

\begin{table}[H]
\centering
\caption{Simulation results for DR~Bayes using a frequentist pilot
  estimator $\widehat m(x)$.  Sample size $n=1000$.}
\vskip.15cm
{\small 
\label{tab:simu_freq_M}
\setlength{\tabcolsep}{3pt}
\renewcommand{\arraystretch}{0.75}
\begin{tabular}{ll cccc| cccc}
\toprule
Method & $\widehat m(x)$
  & \multicolumn{4}{c}{DGP 1}
  & \multicolumn{4}{c}{DGP 2} \\
\cmidrule(lr){3-6}\cmidrule(lr){7-10}
  & & Bias & RMSE & CP & CIL
    & Bias & RMSE & CP & CIL \\
\midrule
DR Bayes$^{\text{Logit}}$ & RF
  & $-0.002$ & $0.116$ & $0.942$ & $0.447$
  & $-0.000$ & $0.117$ & $0.933$ & $0.433$ \\
DR Bayes$^{\text{Logit}}$ & NN
  & $-0.004$ & $0.122$ & $0.927$ & $0.447$
  & $-0.000$ & $0.120$ & $0.926$ & $0.433$ \\
DR Bayes$^{\text{RF}}$    & RF
  & $-0.001$ & $0.161$ & $0.941$ & $0.596$
  & $-0.002$ & $0.162$ & $0.943$ & $0.611$ \\
DR Bayes$^{\text{RF}}$    & NN
  & $-0.000$ & $0.168$ & $0.939$ & $0.597$
  & $-0.003$ & $0.164$ & $0.949$ & $0.611$ \\
\midrule
Method & $\widehat m(x)$
  & \multicolumn{4}{c}{DGP 3}
  & \multicolumn{4}{c}{DGP 4} \\
\cmidrule(lr){3-6}\cmidrule(lr){7-10}
  & & Bias & RMSE & CP & CIL
    & Bias & RMSE & CP & CIL \\
\midrule
DR Bayes$^{\text{Logit}}$ & RF
  & $-0.061$ & $0.812$ & $0.950$ & $3.351$
  & $-1.358$ & $1.688$ & $0.639$ & $3.489$ \\
DR Bayes$^{\text{Logit}}$ & NN
  & $-0.029$ & $0.876$ & $0.943$ & $3.352$
  & $-0.563$ & $1.145$ & $0.887$ & $3.491$ \\
DR Bayes$^{\text{RF}}$    & RF
  & $-0.173$ & $1.047$ & $0.959$ & $4.012$
  & $-0.120$ & $1.500$ & $0.888$ & $4.229$ \\
DR Bayes$^{\text{RF}}$    & NN
  & $-0.047$ & $1.166$ & $0.935$ & $4.013$
  & $-0.260$ & $1.568$ & $0.865$ & $4.223$ \\
\bottomrule
\end{tabular}}
\end{table}

The advantageous performance of DR Bayes relative to DML estimators does not appear to be driven by the choice of pilot estimators in the prior and posterior adjustments.  As Tables~\ref{tab:simu_n1K} and~\ref{tab:simu_n2K} show,  DR Bayes$^{\text{RF}}$ has generally smaller bias and RMSE,  as well as shorter confidence intervals than DML$^{\text{RF-GP}}$,  despite the fact that both use random forest as the point estimator for $\pi_0$ and GP posterior mean as the point estimator for $m_0$.   On the other hand,  although we propose to use the posterior mean as the pilot estimator $\widehat m$ in the posterior correction (\ref{recentering_term}) or (\ref{post_adj_cf}),  we can replace it with some frequentist estimators.  As Table~\ref{tab:simu_freq_M} reveals, DR Bayes maintains its good performance if we instead use random forests or neural net based pilot estimator $\widehat m$ in the posterior correction.  Therefore, the strong performance of DR Bayes is more likely a result of the Bayesian inferential structure which quantifies uncertainty through the posterior.

In Tables~\ref{tab:simu_n1K} and~\ref{tab:simu_n2K}, OR
Bayes performs better than DR Bayes under DGPs~1--4. To examine their performance in a higher-dimensional setting,  we consider DGP~5 which contains eight covariates.   Let  $\left(V_{1}, \dots, V_{8}\right)^\top \sim \mathcal N\left(0, I_{8}\right)$ and $X_{j} = (\tilde{X}_{j}-\mathbb{E}[\tilde{X}_{j}])/\sqrt{\text{Var}[\tilde{X}_{j}]}$, where $\tilde{X}_{1+4k}=\exp(0.5V_{1+4k})$, $\tilde{X}_{2+4k}=10 + V_{2+4k}/(1+\exp(0.5V_{1+4k}))$,   $\tilde{X}_{3+4k}= (0.6 + V_{1+4k}V_{3+4k}/25)^{3}$ and $\tilde{X}_{4+4k}= (20+V_{2+4k}+V_{4+4k})^{2}$, for $k=0,1$.   
Variables $D$ and $Y_{t}$ are generated similarly to DGP~4,  except that  $Y_{t}$ now depends on eight covariates:
\begin{align*}
& D\sim  \text{Bernoulli}\left(\Psi\left[g(V_{1},\dots, V_{4})\right]\right),\\
&Y_{1} = \tilde \mu(V_{1},\dots,  V_{8}) + D\tilde \mu(V_{1},\dots,  V_{8}) +\alpha+ \epsilon_{1},\\
&Y_{2}(d) =  2\tilde \mu(V_{1},\dots,  V_{8})+ D\tilde \mu(V_{1},\dots,  V_{8}) +\alpha+  \epsilon_{2}(d), \text{for} ~ d=0,1,\\
& \text{where} ~ \tilde \mu(w_1, \dots, w_8)= 0.5\mu(w_1, w_2, w_3, w_4) + 0.5\mu(w_5, w_6, w_7, w_8).
\end{align*}
In Table \ref{tab:simu_more_X},  DR Bayes (true PS) plugs the true propensity score into the Riesz representer estimator $\hat\gamma$ and uses the full sample to draw the posterior of $m_0$.  This infeasible estimator sets the highest standard DR Bayes can potentially achieve under DGP~5.
As Table \ref{tab:simu_more_X} shows, OR Bayes exhibits undercoverage in this higher-dimensional setting.  The prior and posterior adjustments,  especially when the true propensity score is available,  restore the empirical coverage rate towards the nominal level.
Estimated propensity scores also yield meaningful improvement,
demonstrating the robustness of DR Bayes with respect to model complexity.

\begin{table}[H]
\centering
\caption{Simulation results under DGP~5 (more covariates)}
\vskip.15cm
{\small 
\label{tab:simu_more_X}
\setlength{\tabcolsep}{3pt}
\renewcommand{\arraystretch}{0.75}
\begin{tabular}{l cccc| cccc}
\toprule
Method 
& \multicolumn{4}{c}{$n=1000$} 
& \multicolumn{4}{c}{$n=2000$} \\
\cmidrule(lr){2-5}\cmidrule(lr){6-9}
& Bias & RMSE & CP & CIL 
& Bias & RMSE & CP & CIL \\
\midrule

OR Bayes
  & $-0.747$ & $0.984$ & $0.801$ & $2.639$
  & $-0.253$ & $0.312$ & $0.732$ & $0.752$ \\
DR Bayes (true PS)
  & $-0.191$ & $0.682$ & $0.943$ & $2.650$
  & $-0.040$ & $0.190$ & $0.956$ & $0.758$ \\
DR Bayes$^{\text{RF}}$
  & $-0.783$ & $1.063$ & $0.868$ & $3.174$
  & $-0.217$ & $0.284$ & $0.878$ & $0.863$ \\
DR Bayes$^{\text{RF}}$(FS)
  & $-0.541$ &  $0.836$ &   $ 0.879$ &   $ 2.609$ 
  &  $-0.171$ &  $0.247$ &    $0.859$ &   $ 0.753$ \\
\cline{1-9}
DR
  & $-1.295$ & $1.508$ & $0.589$ & $2.951$
  & $-0.503$ & $0.554$ & $0.401$ & $0.877$ \\
OR
  & $-2.618$ & $2.752$ & $0.129$ & $3.345$
  & $-1.037$ & $1.068$ & $0.020$ & $0.980$ \\
IPW
  & $-1.954$ & $2.225$ & $0.534$ & $4.107$
  & $-0.789$ & $0.844$ & $0.259$ & $1.174$ \\
DML$^{\text{RF-RF}}$
  & $-1.063$ & $1.323$ & $0.895$ & $4.009$
  & $-0.414$ & $0.471$ & $0.742$ & $1.105$ \\
DML$^{\text{RF-NN}}$
  & $-1.139$ & $1.410$ & $0.772$ & $3.622$
  & $-0.429$ & $0.484$ & $0.599$ & $0.987$ \\
DML$^{\text{RF-GP}}$
  & $-1.479$ & $1.679$ & $0.646$ & $3.630$
  & $-0.513$ & $0.555$ & $0.409$ & $0.920$ \\
\bottomrule
\end{tabular}}
\end{table}

Appendix \ref{appendix:simu} in the supplementary materials provides additional simulation evidence.  Our Bayesian procedures behave differently from applying Bayesian bootstrap to frequentist estimators.
The finite-sample performance of DR Bayes is largely insensitive to the choice of $\varsigma_n$, the tuning parameter governing the strength of the prior adjustment. The Bayesian procedures also demonstrate robustness to non-normality and heteroskedasticity in the error terms. Moreover, OR Bayes continues to perform well as the overlap condition approaches violation. While the coverage performance of DR Bayes deteriorates in DGP~4 under weaker overlap, it still compares favorably with the frequentist alternatives.

\subsection{Empirical Application}

We apply Bayesian DiD methods to the well-known minimum wage study of  \cite{card1994minimum}.\footnote{The data are available at \href{https://davidcard.berkeley.edu/data_sets}{https://davidcard.berkeley.edu/data\_sets}. }
The outcome variables are full-time equivalent (FTE) employment at fast-food stores in New Jersey and Pennsylvania before and after New Jersey’s minimum wage increase. The treatment variable equals one for fast-food stores in New Jersey and zero otherwise. The set of covariates $X$ includes seven store characteristics measured prior to the policy change: an indicator for company ownership, three chain type dummies, hours open per weekday, the prices of medium soda and a main course, yielding a sample of 
338 observations.
We investigate whether the findings of \cite{card1994minimum} are robust to additional covariates and more flexible model specifications. Since the data contains a non-negligible proportion of units with estimated propensity scores very close to $1$,  we follow \cite{crump2009dealing} in discarding observations with estimated propensity scores outside the range $(0, 1 - t]$ and set the trimming threshold $t=0.025$ and $0.01$.  Based on our simulation results,  we implement the DR Bayes using the RF-based propensity score estimator,  which is robust to nonlinear models. For frequentist alternatives, we consider the parametric DR,  OR,  IPW,  nonparametric DML$^{\text{RF-RF}}$ and DML$^{\text{RF-NN}}$ in the simulation exercises, as well as the standard two-way fixed effects (TWFE) estimator.

\begin{figure}[H]\caption{ATT estimates and $95\%$ credible/confidence intervals for the impact of a minimum wage increase on employment.  Trimming thresholds for the estimated propensity score are $t=0.025$ (left) and $t = 0.01$ (right). }\label{fig:mw_att}
  \includegraphics[width=1\textwidth]{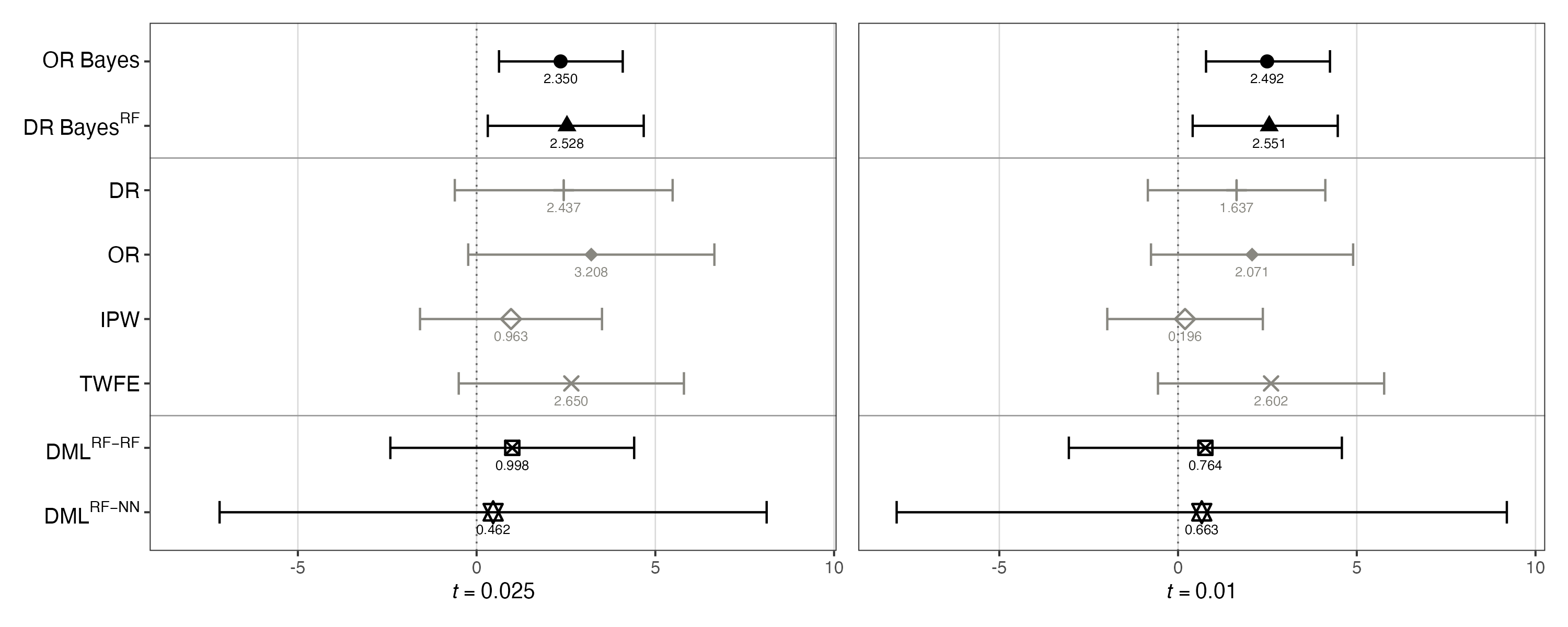}
\end{figure}

Figure \ref{fig:mw_att} presents the ATT estimates and $95\%$ credible/confidence intervals. All methods produce positive point estimates, and are either statistically insignificant or modestly significant at the $95\%$ level. DR Bayes$^{\text{RF}}$ yields point estimates around $2.5$ and credible intervals marginally on the right hand side of zero.  Bayesian credible intervals are shorter  than frequentist confidence intervals except for the IPW estimator which leads to smaller estimates than most other methods. These results are consistent with \cite{card1994minimum}'s finding of no evidence that an increase in the minimum wage reduced employment. The Bayesian estimates are also quantitatively close to \cite{card1994minimum}'s results: a canonical DiD estimate of $2.76$ (standard error $1.36$),  and a regression-adjusted estimate of $2.30$ (standard error $1.20$) when 
controlling for the chain and ownership dummies. Therefore,  our Bayesian DiD methods confirm the robustness of findings in the classic literature against model specifications.

\section{Extensions}\label{sec:extend}
We provide extensions to the canonical DiD panel data setup and show that our Bayesian DiD methods, described in Section \ref{sec:method_outline}, can be conveniently extended to cases such as clustered data and multiple periods with staggered entry.

\subsection{Clustered Data}\label{cluster}
This section considers Bayesian inference at a more aggregate level than the cross-sectional units $i$'s. We focus on the case where there are many clusters and each cluster contains a finite number of units.
Suppose that the observations $Z_{i}$, $i=1,\dots,n$ consist of $J$ mutually independent clusters, indexed $j=1,\dots,J$ where the $j$-th cluster has $n_j$ observations, with $\sum_{j=1}^Jn_j=n$. Using the double-index, we write the clustered data as $Z_{jk}=(\Delta Y_{jk},D_{jk},X_{jk}^\top)^\top$ for clusters $j=1,\dots,J$ and units $k=1,\dots,n_j$ in the $j$-th cluster. Algorithm \ref{algorithm_1_c} modifies Algorithm \ref{algorithm_1} to account for the clustering structure. Analogously, the DR Bayes in Algorithm \ref{algorithm_2} can also be modified for clustered data.\footnote{For clustered data, when DR Bayes is implemented with cross-fitting, the partition of folds is conducted at the cluster level.}
\begin{algorithm}[H]
\caption{Bayesian Procedure for Clustered Data}
\label{algorithm_1_c}
\begin{algorithmic}
    \STATE \textbf{Input:} Data $Z_{jk}$ for $j=1,\dots,J$ and $k=1,\dots,n_j$,  and number of posterior draws $S$.
    \STATE \textbf{Prior Specification:}  Set the prior for $m_\eta(x) =W^m(x)$, where $W^m$ is a Gaussian process.

\textbf{Posterior Computation:} 
    \FORNN{$s=1,\ldots, S$}
        \STATE 
  (a) Generate the $s$-th draw of the posterior of $\left(m_\eta(X_{jk}),j=1,\dots,J, k=1,\dots,n_j\right)$ using the Gaussian process prior and the data from the control arm; denote it as $\left(m^s_\eta(X_{jk}),j=1,\dots,J, k=1,\dots,n_j\right)$.
        \STATE (b) Draw Bayesian bootstrap weights: Let $V^s_{j}=e^s_j/\sum_{j=1}^J e^s_j$ where $e_j^s \stackrel{i.i.d.}{\sim} \textup{Exp}(1)$, for $j=1,\dots,J$. Set $M_{jk}^s=V_j^s/n_{j}$.\hskip-3cm
        \STATE (c) Calculate a posterior draw for the ATT: 
         \begin{equation}\label{NpBayes_c}
\tau_\eta^s=\frac{\sum_{j=1}^J\sum_{k=1}^{n_j} M_{jk}^sD_{jk} \big(\Delta Y_{jk}-m_\eta^s(X_{jk})\big)}{\sum_{j=1}^J\sum_{k=1}^{n_j} M_{jk}^s D_{jk}}.
\end{equation}
    \ENDFOR
    
 \STATE \textbf{Output: $\{\tau_\eta^{s}:s=1,\ldots,S\}$} 
\end{algorithmic}
\end{algorithm}
The main change relative to Algorithm \ref{algorithm_1} occurs in the Bayesian bootstrap weights $M_{jk}^s$, which take the same value for all units within cluster $j$ and thus preserve the within-cluster dependence. If each cluster contains the same number of units $n_j=n/J$ and the observed vectors $(D_{jk}\Delta Y_{jk},D_{jk},X_{jk})_{k=1}^{n_j}$ are \textit{i.i.d.} across $j$, one can view the Bayesian bootstrap weights in Algorithm \ref{algorithm_1_c} as arising from a DP prior on the cluster-level observed vector with zero base measure. 
To accommodate possibly different cluster sizes, our procedure normalizes by $n_j$ when averaging the conditional mean function over units within each cluster.

For posterior draws of the conditional mean function, we adopt a misspecified 
\textit{i.i.d.} Gaussian working likelihood that ignores the clustering structure, in 
the spirit of the framework developed in Section~\ref{sec:misspecify}. In particular, 
we construct the quasi-posterior from a product Gaussian working likelihood that 
treats $\Delta Y_{jk}(0)\mid D_{jk}=0,X_{jk}=x$ as normally distributed with mean 
$m_{\eta}(x)$ and unit variance, independently over all 
$j=1,\ldots,J$ and $k=1,\ldots,n_j$. 
This working specification is justified by Lemma~\ref{lemma:MisRateCluster}, which 
establishes posterior consistency for the conditional mean function while allowing 
for intra-cluster dependence, provided that clusters are independent.

We present Monte Carlo simulation evidence to support our proposal. We modify DGPs~1 to~4 by making the outcome error terms $(\epsilon_{1}, \epsilon_{2}(0), \epsilon_{2}(1))$ equi-correlated within clusters \citep{moulton1986random}, where each cluster consists of $10$ units.  Specifically,  the error $\epsilon_{1}$ for $n$ units $(\epsilon_{11}, \dots, \epsilon_{n1})^\top \sim \mathcal{N}(0, \Sigma_n)$, where $\Sigma_n = \mathrm{diag}\left\{\Sigma_1, \dots, \Sigma_J\right\}$ is block-diagonal with $J = n/10$ blocks, each of the form $\Sigma_j = (1-\rho)I_{10} + \rho\mathbf{1}_{10}\mathbf{1}_{10}^\top$, where $\rho = 0.5$ is the intra-cluster correlation and $\mathbf{1}_{10}$ denotes a length-10 vector of ones. The post-treatment potential outcome errors $(\epsilon_{12}(d), \dots, \epsilon_{n2}(d))^\top \sim \mathcal{N}(0, \Sigma_n)$ for $d = 0, 1$ follow the same block-diagonal structure.
Table \ref{tab:simu_n1K_cl} reports finite-sample performance,
comparing the proposed Bayesian methods modified for clustered data, the
cluster-robust DR estimator based on multiplier bootstrap, and DML estimators
with cluster-robust standard errors.\footnote{The cluster-robust DR is implemented using the R package \texttt{did}. DML standard errors are obtained
via the sandwich variance estimator in \texttt{DoubleML}, which aggregates
orthogonal score contributions at the cluster level.}
All methods in Table \ref{tab:simu_n1K_cl}  achieve satisfactory coverage in DGPs~1 to~3 under clustered
errors.  In DGP 4 with clustered errors,  the parametric DR once again exhibits large bias and low coverage probability. OR Bayes maintains a high coverage without relying on a propensity
score model.
DR Bayes$^{\text{Logit}}$ undercovers due to its misspecified
logit model for $\pi_0$, while DR Bayes$^{\text{RF}}$, using the more flexible
random forest for $\pi_0$,  improves the coverage.  
Overall, the relative performance between Bayesian and frequentist methods in Table~\ref{tab:simu_n1K_cl} closely
mirrors that in Table~\ref{tab:simu_n1K}. This demonstrates that the proposed cluster-robust modifications effectively account for within-cluster dependence while preserving the desirable properties of the Bayesian methods.

% ------------------------------------------------------------------
%  tab:simu_n1K_cl
% ------------------------------------------------------------------
\begin{table}[H]
\centering
\caption{Simulation results under DGPs 1 to 4 with clustered error terms. Sample size $n=1000$.}
\vskip.15cm
{\small 
\label{tab:simu_n1K_cl}
\setlength{\tabcolsep}{3pt}
\renewcommand{\arraystretch}{0.75}

\begin{tabular}{l cccc| cccc}
\toprule
Method
& \multicolumn{4}{c}{DGP 1}
& \multicolumn{4}{c}{DGP 2} \\
\cmidrule(lr){2-5}\cmidrule(lr){6-9}
& Bias & RMSE & CP & CIL
& Bias & RMSE & CP & CIL \\
\midrule
OR Bayes
  & $-0.002$ & $0.141$ & $0.948$ & $0.545$
  & $-0.000$ & $0.143$ & $0.943$ & $0.546$ \\
DR Bayes$^{\text{Logit}}$
  & $-0.004$ & $0.145$ & $0.949$ & $0.559$
  & $-0.003$ & $0.149$ & $0.932$ & $0.548$ \\
DR Bayes$^{\text{RF}}$
  & $-0.002$ & $0.175$ & $0.951$ & $0.686$
  & $-0.003$ & $0.181$ & $0.945$ & $0.705$ \\
\cline{1-9}
DR
  & $-0.003$ & $0.144$ & $0.942$ & $0.555$
  & $-0.001$ & $0.145$ & $0.939$ & $0.549$ \\
DML$^{\text{RF-RF}}$
  & $-0.160$ & $0.701$ & $0.997$ & $4.451$
  & $-0.549$ & $0.969$ & $0.982$ & $4.916$ \\
DML$^{\text{RF-NN}}$
  & $-0.017$ & $0.180$ & $0.960$ & $0.753$
  & $-0.022$ & $0.198$ & $0.959$ & $0.844$ \\
DML$^{\text{RF-GP}}$
  & $-0.379$ & $0.961$ & $0.968$ & $4.681$
  & $-0.196$ & $1.049$ & $0.995$ & $5.844$ \\
\midrule
Method
& \multicolumn{4}{c}{DGP 3}
& \multicolumn{4}{c}{DGP 4} \\
\cmidrule(lr){2-5}\cmidrule(lr){6-9}
& Bias & RMSE & CP & CIL
& Bias & RMSE & CP & CIL \\
\midrule
OR Bayes
  & $-0.030$ & $0.416$ & $0.990$ & $2.288$
  & $-0.534$ & $0.733$ & $0.976$ & $3.220$ \\
DR Bayes$^{\text{Logit}}$
  & $-0.026$ & $0.810$ & $0.949$ & $3.244$
  & $-0.997$ & $1.304$ & $0.804$ & $3.309$ \\
DR Bayes$^{\text{RF}}$
  & $-0.003$ & $1.023$ & $0.965$ & $3.911$
  & $-0.734$ & $1.431$ & $0.887$ & $4.166$ \\
\cline{1-9}
DR
  & $-0.004$ & $1.202$ & $0.942$ & $4.645$
  & $-3.149$ & $3.426$ & $0.325$ & $5.121$ \\
DML$^{\text{RF-RF}}$
  & $-0.312$ & $0.870$ & $0.985$ & $4.481$
  & $-0.811$ & $1.222$ & $0.961$ & $5.086$ \\
DML$^{\text{RF-NN}}$
  & $-0.026$ & $0.806$ & $0.977$ & $3.870$
  & $-1.028$ & $1.439$ & $0.840$ & $4.508$ \\
DML$^{\text{RF-GP}}$
  & $-0.266$ & $1.013$ & $0.986$ & $5.299$
  & $-0.856$ & $1.642$ & $0.937$ & $7.348$ \\
\bottomrule
\end{tabular}}
\end{table}

\subsection{Multiple Periods and Staggered Entry}\label{sec:staggered}
The Bayesian DiD methods described in Section \ref{sec:method_outline} can be conveniently extended to the cases with multiple periods and staggered intervention \citep{de2020two,callaway2021difference,sun2021estimating,borusyak2024revisiting}. The related literature focuses on the identification of disaggregated causal parameters and some proper aggregation of these parameters. This section extends our Bayesian method for inference on the disaggregated ATT, specifically the group-time ATT proposed by \cite{callaway2021difference}. 

 Suppose the available panel data consists of $T$ periods indexed by $t=1,\dots,T$ and the earliest treatment intervention occurs at period $S$. We assume that the treatment intervention remains once a unit gets treated. As a result, the entire path of treatment assignment for each unit can be summarized by his/her first treated period (cohort), denoted by the cohort variable $G_i\in\{S,\dots,T,\infty\}$, where $G_i=\infty$ means the unit $i$ never gets treated. Let the cohort indicators $D_{ig}$ denote whether unit $i$ first receives treatment in period $g \in \{S, \ldots, T, \infty\}$, where $D_{i\infty}=1$ indicates that unit $i$ never receives treatment.
 
  We assume that never-treated units exist. The potential outcomes depend on cohorts and thus are denoted as $Y_{it}(g)$ for $g\in\{S,\dots,T\}$ and $Y_{it}(0)$ for $G_i=\infty$. Obviously, $\sum_{g=S}^T  D_{ig} + D_{i\infty}=1$. The realized outcome for unit $i$ at time $t$ is 
 $ Y_{it}=Y_{it}(0)+\sum_{g=S}^{T}D_{ig}\left(Y_{it}(g)-Y_{it}(0)\right)$.

We focus on the analysis of treatment effect heterogeneity by allowing the ATT to vary with the cohort $g$ ($g \neq \infty$) and the time period $t \geq g$:
\begin{align*}
\tau_{0}^{g,t}=\mathbb E_0[Y_t(g)-Y_t(0)|D_g=1],~\text{for}~ g=S,\dots,T~\text{and}~ t=g,\dots,T.
\end{align*}
The researcher observes an \textit{i.i.d.} sample of $\left\{(Y_{i1},\dots,Y_{iT}, D_{iS},\dots,D_{iT},D_{i\infty}, X_i^{\top})\right\}_{i=1}^n$, where $X_i$ is a vector of pre-treatment covariates as in Section \ref{sec:setup}.

Applying the identification strategy in \cite{callaway2021difference}, the ATT parameters $\tau_0^{g,t}$ for $g\in\mathcal G:=\{S,\dots ,T\}$ and $t=g,\dots,T$ can be identified under Assumption \ref{Ass:DiD_stagger} below.\footnote{\cite{callaway2021difference} propose two identification strategies, depending on whether the parallel trend assumption is imposed on the never-treated cohort or ``Not-Yet-Treated" cohorts. Here we consider the former version.} We only consider the just-identified situation as in \cite{callaway2021difference}. In principle, it is desirable to make use of the available information from the pre-treatment periods under strengthened parallel trends. This over-identified case has recently been examined by \cite{CSX2025Efficient}, where the proposed efficient estimation procedure explores the sequential conditional moment restrictions. Such a modeling framework has not yet been developed from a Bayesian perspective, and we leave it as an important future topic.

For the identification of the ATT, we follow the setup by \cite{callaway2021difference} and impose the following conditions, which correspond to their Assumptions 3, 4, and 6 (in their $\delta = 0$ case).

\begin{assumption}\label{Ass:DiD_stagger}
 For all $x$ in the support of $F_X$ and $g\in\mathcal G$ we have:\\
 		 (i)
 		 $\mathbb{E}_0\left[Y_t(g)\mid D_g=1, X=x\right] = \mathbb{E}_0\left[Y_t(0)\mid D_g=1, X=x\right]$ for all $t\in \{1,\dots,g-1\}$,\\
 (ii) 		 $\mathbb{E}_0\left[Y_t(0)-Y_1(0) \mid D_g=1, X=x\right]=\mathbb{E}_0\left[Y_t(0)-Y_1(0) \mid D_{\infty}=1, X=x\right]$ for all $t\in\{g,\ldots,T\}$,\\
 (iii) 
 $P_0\left(D_g=1\right) > \epsilon$
 		   and 
 		   $P_0\left(D_g=1 \mid D_g + D_{\infty}=1,  X=x \right) \leq 1-\epsilon$
 		  for some $\epsilon>0$.
 		  \end{assumption}
Assumption \ref{Ass:DiD_stagger}(i) is a ``no anticipation" assumption, Assumption \ref{Ass:DiD_stagger}(ii) is a conditional parallel trend assumption based on the never-treated cohort, and Assumption \ref{Ass:DiD_stagger}(iii) is an overlap restriction.
Under Assumption \ref{Ass:DiD_stagger}, \cite{callaway2021difference}  show that the ATT in the staggered entry case is identified by
\begin{equation}\label{att_0_stagger}
	\tau_{0}^{g,t}
	=\E_0\left[\Delta_g Y_t - m_{0}^{g,t}(X)\mid D_{g}=1\right],~\text{for}~ g\in\mathcal G~\text{and}~ t=g,\dots,T,
\end{equation}
where the difference operator $\Delta_g$ is defined by $\Delta_g Y_t := Y_t - Y_{g-1}$
and the conditional mean function $m_0^{g,t}(x):=  \mathbb{E}_0\left[\Delta_g Y_t\mid D_{\infty}=1, X=x\right]$.

The identification result in (\ref{att_0_stagger}) uses the cohort $g$ (i.e., $D_{g}=1$) as the treated group and the ``never treated" cohort ($D_{\infty}=1$)  as the control group. Using the transformed cross-sectional data $\left(\Delta_g Y_{it}, D_{ig}, X_{i} \right)$ for $i=1,\dots, n$ 
and following the notation in Section \ref{sec:bayes_frame}, we can write ATT for a given pair $(g,t)$ under a family of probability distributions $\{P_\eta:\eta\in\mathcal H\}$ as 
\begin{equation}\label{att_estimand_staggered}
	\tau_{\eta}^{g,t}:=\frac{\mathbb{E}_{\eta}[D_g\Delta_g Y_t-D_g m_{\eta}^{g,t}(X)]}{\mathbb{E}_{\eta}[D_g]},
\end{equation}
where $\mathbb{E}_{\eta}$ denotes the expectation with respect to the distribution of $\left(\Delta_g Y_t , D_g, X\right)$.
The Bayesian DiD procedures in Section \ref{sec:method_outline} can be applied in the staggered DiD case to obtain the posterior draws $\left\{(\tau_{\eta}^{g,t})^s: s=1,\ldots,S \right\}$. Specifically, this can be achieved by replacing $\Delta Y_i$, $D_i$, $m_{\eta}(\cdot)$, $\bar\pi_{\eta}$, and $\pi_{\eta}(\cdot)$ in Algorithm \ref{algorithm_1} or \ref{algorithm_2} by $\Delta_g Y_{it}$, $D_{ig}$, $ m_{\eta}^{g,t}(\cdot)$, $\bar\pi_{\eta}^g:=\mathbb{E}_{\eta}[D_g]$ and $\pi_{\eta}^g(\cdot):=P_\eta\left(D_g=1 \mid D_g + D_{\infty}=1,  X=\cdot \right)$, respectively, as defined in this section. 

The first resulting Bayesian estimator is denoted by $\tau_\eta^{g,t}$, while the second, double-robust Bayesian method is denoted by $\check{\tau}_\eta^{g,t}$ for a cohort $g \in \mathcal{G}$. The next result is an immediate implication of Theorem \ref{BvM:thm:standard} and Theorem \ref{thm:Debias}, and its proof is thus omitted.
\begin{corollary}\label{cor:stag}
Let Assumption \ref{Ass:DiD_stagger} hold, and suppose that for any $g \in \mathcal{G}$:
\begin{itemize}
    \item[(i)] If Assumptions \ref{Assump:NonDRRate}--\ref{Assump:NonDRPS} hold under the $g$--specific components, i.e., $(\Delta Y_i, D_i, m_{\eta}(\cdot), \bar\pi_{\eta}, \pi_{\eta}(\cdot))$ are replaced by $(\Delta_g Y_{it}, D_{ig}, m_{\eta}^{g,t}(\cdot), \bar\pi_{\eta}^g, \pi_{\eta}^g(\cdot))$, then the Bayesian method $\tau_\eta^{g,t}$ satisfies the BvM result in Theorem \ref{BvM:thm:standard}.
    \item[(ii)] If Assumptions \ref{Assump:NonDRSE}(i), \ref{Assump:Rate}, \ref{Assump:Donsker}, and \ref{Assump:Prior} hold under the $g$--specific components, then the doubly robust Bayesian method $\check{\tau}_\eta^{g,t}$ satisfies the BvM result in Theorem \ref{thm:Debias}. 
\end{itemize}
\end{corollary}
Corollary \ref{cor:stag} pertains to inference on cohort-specific ATTs and establishes BvM results for our two Bayesian methods, employing either standard Gaussian process priors or prior/posterior adjustments via the cohort-specific propensity score. We note that extending this framework to aggregate ATTs is highly non-trivial, as it requires a joint modeling of outcome variables across different cohorts and time periods. Hence, distinct prior and likelihood specifications in the Bayesian methodology, as well as prior/posterior adjustments of the doubly robust version, are needed. A thorough investigation is therefore left for future research.

\section{Conclusion}
This paper introduces new semiparametric Bayesian procedures that satisfy the Bernstein-von Mises results in the DiD setup. Our first proposal provides a Bayesian analog to the outcome regression in \cite{heckman1997matching}. Through simulations, we show that it performs well in models that are not overly complex and, since no propensity score specification is required, it is not sensitive to weaker overlap. Our second, doubly robust proposal incorporates prior/posterior corrections based on estimated propensity scores. In simulations it works well for complex models, i.e., when the number of covariates is large. Overall, both Bayesian methods exhibit remarkable finite-sample performance, while adapting to the functional form of the conditional mean function. Although the paper mainly focuses on the canonical DiD panel-data case, it also discusses extensions to clustered data and staggered interventions.

\bibliographystyle{econometrica}
\bibliography{references}

\appendix

\let\originalclearpage\clearpage
\let\clearpage\relax
 \section{Proofs of Main Results}\label{appendix:main:proofs}
In the Appendix, $C>0$ denotes a generic constant, whose value might change line by line. We introduce additional subscripts when there are multiple constant terms in the same display. We also show in the Supplementary Appendix \ref{appendix:lfd} that $\gamma_\eta$ determines the least favorable direction of Bayesian submodels. For simplicity of notation we write $\sum_i$ instead of $\sum_{i=1}^n$ below.

	Our theoretical results rely on a key decomposition of the frequentist estimator $\widehat{\tau}$ implied by asymptotic efficiency. The minimal asymptotic variance for estimating the ATT can be written in terms of the information norm as
\begin{align}\label{var:equivalence}
	P_0\bigl[\gamma_0^2\bigr]
	= P_0\bigl[\widetilde{\tau}_0^{\,2}\bigr]
	= \textsc v_0 ,
\end{align}
	which is used in the results below.
	In the following, we denote the log-likelihood based on $Z^{(n)}$ as
	\begin{align*}
	\ell_n(\eta)=\sum_i\log p_{\eta}(Z_i)=\ell_n^m(\eta^m)+\ell_n^f(\eta^f),
	\end{align*}
	where each term is the logarithm of the factors involving only $m$ or $f$. 
	Note that we only put a prior distribution on $\eta^m$ and $\eta^f$, and thus the consideration of the likelihood above is sufficient, as shown in the following proofs. 
	
Define the set $\mathcal{H}_n$ that contains $m_{\eta}$ with posterior probability going to 1. Recall the definition of the measurable sets $\mathcal H^m_n$ of functions $\eta$ such that $\Pi(\eta^m\in\mathcal{H}^m_n\mid Z^{(n)})\to_{P_0} 1$. We introduce the conditional prior $\Pi_n(\cdot):=\Pi(\cdot \cap \mathcal{H}^m_n)/\Pi(\mathcal{H}^m_n)$. Define $\upsilon_{\eta}=\bar\pi_\eta/\bar\pi_0$. 

As we show the conditional weak convergence via examining the convergence of the conditional Laplace transform, the following posterior Laplace transform of  $\sqrt{n}\left[\upsilon_{\eta} (\tau_\eta-\widehat{\tau})-b_{0, \eta}\right]$ for all $t\in\mathbb{R}$
	\begin{equation}\label{laplace:transform:def}
	I_n(t)= \mathbb{E}^{\Pi_n}\left[e^{t\sqrt{n}[\upsilon_{\eta}(	\tau_\eta-\widehat\tau)-b_{0, \eta}]}\mid Z^{(n)} \right],
	\end{equation}
	plays a crucial role in establishing the BvM theorem \citep{castillo2012gaussian,castillo2015bvm,ray2020causal}. See also Lemma \ref{lemma:WeakConv} in the Supplementary Appendix \ref{appendix:auxiliary:results}. 
Recall the ``bias term'' given in Theorem \ref{thm:BvM} is
\begin{align*}
	b_{0,\eta}:&=\frac{1}{n}\sum_i\gamma_0(D_i,X_i)[m_0(X_i)-m_{\eta}(X_i)]\\
	&=\frac{1}{n}\sum_i\left(\frac{D_i}{\bar\pi_0}-\frac{1-D_i}{\bar\pi_0} \frac{\pi_0(X_i)}{1-\pi_0(X_i)}\right)[m_0(X_i)-m_{\eta}(X_i)].
\end{align*}
The ``de-biasing term'' of our posterior correction is given by 
\begin{align*}
	\widehat{b}_{\eta}:&=\frac{1}{n}\sum_i\widehat{\gamma}(D_i,X_i)[\widehat{m}(X_i)-m_{\eta}(X_i)]\\
	&=\frac{1}{n}\sum_i\left(\frac{D_i}{\widehat{\bar{\pi}}}-\frac{1-D_i}{\widehat{\bar{\pi}}} \frac{\widehat\pi(X_i)}{1-\widehat \pi(X_i)}\right)[\widehat{m}(X_i)-m_{\eta}(X_i)].
\end{align*}
Because the expansions in the proof of both BvM theorems largely coincide, we keep the bias term explicit even in proving the non-doubly robust version where the bias term is asymptotically negligible.

\begin{proof}[\textbf{\textsc{Proof of Theorem \ref{BvM:thm:standard}.}}]
We begin with a useful decomposition of $\tau_\eta-\widehat{\tau}$. We may assume $\widehat\tau=\tau_0+\mathbb P_n[\widetilde{\tau}_0]$, which satisfies \eqref{def:est:chi}. Consequently, from the definition of the efficient influence function in \eqref{eif_att} we infer
	\begin{align*}
		\widehat \tau-\tau_0&=\frac{1}{\bar\pi_0 n}\sum_i\left(\frac{D_i-\pi_0(X_i)}{1-\pi_0(X_i)}\left(\Delta Y_i-m_{0}(X_i)\right)-D_i\tau_0\right).
	\end{align*}
	In addition, the definition of the Bayesian procedure in \eqref{att} implies
	\begin{align*}
		\tau_\eta-\tau_0=\frac{\mathbb{E}_{\eta}\left[D\left(\Delta Y-m_{\eta}(X)\right)-D\tau_0\right]}{\mathbb{E}_{\eta}[D]}.
	\end{align*}
Thus, using the notation $\upsilon_{\eta}=\bar\pi_\eta/\bar\pi_0$, we obtain the decomposition
\begin{align*}
 \upsilon_{\eta}\sqrt{n}(\tau_{\eta}-\widehat{\tau})&=\frac{\bar\pi_\eta}{\bar\pi_0}\sqrt{n}(\tau_{\eta}-\tau_0-(\widehat{\tau}-\tau_0))\\
 &=\frac{1}{\bar\pi_0}\sqrt{n}\,\E_{\eta}\left[D\left(\Delta Y-m_{\eta}(X)\right)-D\tau_0\right]\\
	&-\frac{1}{\bar\pi_0 \sqrt n}\sum_i\left(\frac{D_i-\pi_0(X_i)}{1-\pi_0(X_i)}\left(\Delta Y_i-m_{0}(X_i)\right)-D_i\tau_0\right)\\
	&+\underbrace{\frac{1}{\sqrt n}\sum_i\left(\frac{D_i-\pi_0(X_i)}{1-\pi_0(X_i)}\left(\Delta Y_i-m_{0}(X_i)\right)-D_i\tau_0\right)\left(1-\upsilon_{\eta}\right)}_{=:R_{n,\eta}}.
\end{align*}
Considering the second summand on the right hand side, we make use of the relation
\begin{align*}
		&\frac{1}{\bar\pi_0 n}\sum_i\left(\frac{D_i-\pi_0(X_i)}{1-\pi_0(X_i)}\left(\Delta Y_i-m_{0}(X_i)\right)-D_i\tau_0\right)\\
		&=\frac{1}{\bar\pi_0 n}\sum_i\Big(D_i(\Delta Y_i-m_0(X_i)-\tau_0) -(1-D_i)\frac{\pi_0(X_i)}{1-\pi_0(X_i)}\left(\Delta Y_i-m_0(X_i)\right)\Big)
	\end{align*}
and hence arrive at the following decomposition
\begin{align*}
	 \upsilon_{\eta}t\sqrt{n}(\tau_{\eta}-\widehat{\tau})-tR_{n,\eta}
	&=\underbrace{\frac{t}{\sqrt n}\sum_i\left(\E_\eta\left[\frac{D}{\bar\pi_0}(\Delta Y-m_\eta(X)-\tau_0)\right]-\frac{D_i}{\bar\pi_0}(\Delta Y_i-m_\eta(X_i)-\tau_0)\right)}_{=t\sqrt{n}(P_{\eta}-\mathbb{P}_n)\left[\frac{d}{\bar\pi_0}(\Delta y-m_\eta(x)-\tau_0)\right]}\\
&+\underbrace{\frac{t}{\sqrt n}\sum_i\frac{D_i}{\bar \pi_0}(m_0(X_i)-m_{\eta}(X_i))}_{\textcircled{a}}\\
&+\underbrace{\frac{t}{\sqrt n}\sum_i\frac{1-D_i}{\bar\pi_0} \frac{\pi_0(X_i)}{1-\pi_0(X_i)} \left(\Delta Y_i-m_0(X_i)\right).}_{\textcircled{b}}
\end{align*}

In order to show the conditional (on the observed data) convergence of the posterior distribution in the bounded Lipschitz distance, it is sufficient to show the pointwise convergence of the posterior Laplace transform for every $t$ in a neighborhood of $0$, by Theorem 1.13.1 of \cite{van1996empirical}. 
The Laplace transform given in \eqref{laplace:transform:def} can be written for all $t\in\mathbb{R}$ as
\begin{equation*}
	I_n(t)=\int\frac{\int_{\mathcal{H}_n}\left[e^{t\sqrt n \upsilon_{\eta}\left(	\tau_\eta-\widehat{\tau}\right)-tR_{n,\eta}-t\sqrt{n}b_{0,\eta}}\right]e^{\ell_n^m(\eta^m)}\mathrm{d}\Pi(\eta^m)}{\int_{\mathcal{H}_n}e^{\ell_n^m(\eta^{m\prime})}\mathrm{d}\Pi(\eta^{m\prime})}\mathrm{d}\Pi(F_{\eta}\mid Z^{(n)}_{\text{Treated}}).
%	 ~~~\forall t\in\mathbb{R}
\end{equation*}
We consider the perturbation via the least favorable direction for the log-likelihood part that depends on $m_{\eta}$. Specifically, we introduce
\begin{align*}
\eta_t^m:=\eta_t\left(\eta^m\right):=\eta^m-\frac{t}{\sqrt{n}} \gamma_{c,0} \quad\text{where}\quad \gamma_{c,0}:=-\frac{(1-d)}{\bar\pi_0}\frac{\pi_0(x)}{1-\pi_0(x)},
\end{align*}
which defines a perturbation of $\eta^m$ along the control arm of the least favorable direction $\gamma_{c,0}$.

We further evaluate for the Laplace transform for all $t\in\mathbb{R}$: 
\begin{equation*}
	I_n(t)=\int\frac{\int_{\mathcal{H}_n}\left[e^{t\sqrt n \upsilon_{\eta}\left(	\tau_\eta-\widehat{\tau}\right)-tR_{n,\eta}-t\sqrt{n}b_{0,\eta}}\right]e^{\ell_n^m(\eta^m)-\ell_n^m(\eta^m_t)}e^{\ell_n^m(\eta^m_t)}d\Pi(\eta^m)}{\int_{\mathcal{H}_n}e^{\ell_n^m(\eta^{m\prime})}\mathrm{d}\Pi(\eta^{m\prime})}\mathrm{d}\Pi(F_{\eta}\mid Z^{(n)}_{\text{Treated}}).
%	~~~\forall t\in\mathbb{R}
\end{equation*}
By Lemma \ref{lemma:likelihood} we further obtain for the likelihood functions uniformly for $\eta\in\mathcal H_n$:
\begin{align*}
	\ell_n^m(\eta^m)-\ell_n^m(\eta^m_t)&=-\underbrace{\frac{t}{\sqrt n}\sum_{i=1}^n\frac{1-D_i}{\bar\pi_0} \frac{\pi_0(X_i)}{1-\pi_0(X_i)} \left(\Delta Y_i-m_0(X_i)\right)}_{\textcircled{c}}\\
	&+\underbrace{\frac{t^2}{2}\E_0\left[ \frac{1-D}{\bar{\pi}_0^2} \frac{\pi_0^2(X)}{(1-\pi_0(X))^2}(\Delta Y-m_0(X))^2\right]}_{\textcircled{d}}\\
	&-\underbrace{\frac{t}{\sqrt n}\sum_{i=1}^n\frac{1-D_i}{\bar\pi_0} \frac{\pi_0(X_i)}{1-\pi_0(X_i)} (m_0(X_i)-m_{\eta}(X_i))}_{\textcircled{e}}+o_{P_0}(1).
\end{align*}
We immediately see that the term $\textcircled{b}$ cancels with $\textcircled{c}$, and $\textcircled{a}-\textcircled{e}$ cancels $t\sqrt{n}\,b_{0,\eta}$ in the expression for the Laplace transform $I_n(t)$. In addition, because all variables have been integrated out in the integral in the denominator, it is a constant relative to either $m_{\eta}$ or $F_\eta$. By Fubini's Theorem, the double integral without this normalizing constant is
\begin{align*}
	\int_{\mathcal{H}^m_n}\! \! \exp&\left(\frac{t^2}{2}\E_0\left[ \frac{1-D}{\bar{\pi}_0^2} \frac{\pi_0^2(X)}{(1-\pi_0(X))^2}(\Delta Y-m_0(X))^2\right]+\ell_n^m(\eta^m_t)\right)\\
	&\times \underbrace{\int \exp\left(t\sqrt{n}(P_{\eta}-\mathbb{P}_n)\left[\frac{d}{\bar\pi_0}(\Delta y-m_\eta(x)-\tau_0)\right]\right)\mathrm{d}\Pi(F_\eta\mid Z_{\text{Treated}}^{(n)})}_{\textcircled{f}}\mathrm{d}\Pi(\eta^m).
\end{align*}
Note that the posterior law of $F_{\eta}$ conditional on the observed data is equivalent to the Bayesian bootstrap measure. 
Using the envelope condition imposed in Assumption \ref{Assump:NonDRSE},  we may apply Lemma \ref{lemma:DP} to the $\textcircled{f}$ term so that for the conditional Laplace transform we have for all $t\in\mathbb{R}$:
\begin{eqnarray}
I_n(t)&=\exp\left( \frac{t^2}{2}\left(Var_0\left[\frac{D}{\bar\pi_0}(\Delta Y-m_0(X)-\tau_0)\right] + \E_0\left[
\frac{1-D}{\bar{\pi}_0^2} \frac{\pi_0^2(X)}{(1-\pi_0(X))^2}(\Delta Y-m_0(X))^2\right]\right)\right)\notag\\
&\qquad\times\underbrace{\frac{\int_{\mathcal{H}_n}e^{\ell_n^m(\eta^m_t)}\mathrm{d}\Pi(\eta^m)}{\int_{\mathcal{H}_n}e^{\ell_n^m(\eta^{m\prime})}\mathrm{d}\Pi(\eta^{m\prime})}}_{\textcircled{g}}\times \exp(o_{P_0}(1))\label{prior_vary} \\
&=\exp\left(\frac{t^2}{2}\textsc v_0\right)+o_{P_0}(1)\notag,
\end{eqnarray}
where the last line follows from that the term $\textcircled{g}=1+o_{P_0}(1)$, which is implied by Assumption \ref{Assump:NonDRPS} and Lemma \ref{lemma:prior:stability}. 

We apply Lemma \ref{lemma:WeakConv} by taking $S_n=\sqrt{n}\left[\upsilon_{\eta}(\tau_\eta-\widehat{\tau})-b_{0,\eta}\right]-R_{n,\eta}$ and the limiting law $L$ as the normal distribution $N(0,\textsc v_0)$. Thus, we have shown that the posterior distribution of $\sqrt{n}\left[\upsilon_{\eta}(\tau_\eta-\widehat{\tau})-b_{0,\eta}\right]-R_{n,\eta}$ converges to $N(0,\textsc v_0)$ in the bounded Lipschitz norm.
 Note that the bias term is asymptotically negligible given the stochastic equicontinuity in Assumption \ref{Assump:NonDRSE}. We have also shown the negligibility of $R_{n,\eta}$ in Lemma \ref{lemma:Remainder}. Hence, we apply Lemma \ref{lemma:Slutsky} to show the conditional convergence of $\sqrt{n}(\tau_\eta-\widehat{\tau})$, as $\upsilon_{\eta}$ converges to $1$, in $P_0$-probability conditional on the data, which concludes the proof.
	\end{proof}

\begin{proof}[Proof of Lemma \ref{lemma:MisRate}] We apply Lemma \ref{lemma:MisRateCluster} when the number of clusters equals the sample size and each cluster contains only one unit. The tail bound of the latent error (\ref{ErrorMoment}), as well as the prior mass (\ref{PriorMass}) and the control of covering number (\ref{CoveringNumber}), give us
			\begin{align*}
					\Pi\left(
					m_{\eta}:\Vert m_{\eta}-m_0\Vert_{ 2,F_0}^2\geq M\varepsilon_n^2
					\mid Z^{(n)}
					\right)
					\to_{P_0}0,
				\end{align*}
		for a sufficiently large constant $M>0$, which is the desired result.
	\end{proof}

\begin{proof}[Proof of Theorem \ref{thm:BvM}]
	Since the estimated least favorable direction $\widehat{\gamma}$ is based on observations that are independent of $Z^{(n)}$, we may apply Lemma 2 of \cite{ray2020causal}. That is, it suffices to handle the ordinary posterior distribution with $\widehat{\gamma}$ set equal to a deterministic function $\gamma_n$. Consequently, for the analysis of the conditional Laplace transform $I_n(t)$, we can follow the proof of Theorem \ref{BvM:thm:standard}.
	Further, the prior stability condition is satisfied by Assumption \ref{Assump:Prior} and the proof of Lemma B.2 from \cite{BLY2025supplement}.
	
In sum, we have shown that $\sqrt{n}\left[\upsilon_{\eta}(\tau_\eta-\widehat{\tau})-b_{0,\eta}\right]-R_{n,\eta}$ converges to the normal distribution $N(0,\textsc v_0)$ in bounded Lipschitz norm by Lemma \ref{lemma:WeakConv}. In Lemma \ref{lemma:Remainder}, we prove that $\sup_{\eta\in \mathcal{H}_n}\vert \sqrt{n}(\upsilon_{\eta}-1)b_{0,\eta}\vert =o_{P_0}(1)$ which implies the conditional weak convergence of $\sqrt{n}\upsilon_{\eta}(\tau_\eta-\widehat{\tau}-b_{0,\eta})$ to the same normal distribution (under $P_0$). Finally, we establish this result for  $\sqrt{n}(\tau_\eta-\widehat{\tau}-b_{0,\eta})$ by dropping the scaling factor $\upsilon_{\eta}$, due to Lemma \ref{lemma:Slutsky}. 
\end{proof}

	\begin{proof}[Proof of Theorem \ref{thm:Debias}]
		It is sufficient to show that
	\begin{equation*}
		\sup_{\eta\in\mathcal{H}_n}\left|b_{0,\eta}-\widehat{b}_{\eta}\right|=o_{P_0}(n^{-1/2}),
	\end{equation*}
	where $b_{0,\eta}=	\mathbb{P}_n[\gamma_0(m_0-m_{\eta})]$ and $\widehat{b}_\eta=	\mathbb{P}_n[\widehat \gamma(\widehat m - m_{\eta}) ]$.
	We make use of the decomposition
	\begin{equation}\label{dec:proof:feasible:drb}
		b_{0,\eta}-\widehat{b}_{\eta}=	\mathbb{P}_n[(\gamma_0-\widehat \gamma)(m_0-m_{\eta})]
		+ \mathbb{P}_n[\widehat \gamma(m_0- \widehat m) ].
	\end{equation}
	Consider the first summand on the right hand side of the previous equation.
From Assumption \ref{Assump:Donsker} we infer
\begin{align*}
		\sqrt n\sup_{\eta\in\mathcal{H}_n}&\left|	\mathbb{P}_n[(\gamma_0-\widehat \gamma)(m_0-m_{\eta})]\right|
		\leq 		\sup_{\eta\in\mathcal{H}_n}\left|\mathbb{G}_n[(\gamma_0-\widehat \gamma)(m_0-m_{\eta})]\right|\\
		&\qquad\qquad\qquad\qquad\qquad\qquad+\sqrt n\sup_{\eta\in\mathcal{H}_n}\left|P_0[(\gamma_0-\widehat \gamma)(m_0-m_{\eta})]\right|\\
		&\qquad\leq o_{P_0}(1)+	O_{P_0}(1)\times \sqrt{n}\Vert \pi_0 - \widehat{\pi}\Vert_{2, F_0}\sup_{\eta\in\mathcal{H}_n}\Vert m_\eta-m_0\Vert_{2, F_0}=o_{P_0}(1),
	\end{align*} 
using the Cauchy-Schwarz inequality and Assumption \ref{Assump:Rate}.
	Consider the second summand on the right hand side of \eqref{dec:proof:feasible:drb}. 
Another application of the Cauchy-Schwarz inequality and Assumption \ref{Assump:Rate} yields
	\begin{align*}
		 \mathbb{P}_n[\widehat \gamma(m_0 -\widehat m ) ]= \mathbb{P}_n[\gamma_0(m_0 -\widehat m )]+o_{P_0}(n^{-1/2}).
	\end{align*}
	Using Lemma \ref{lemma:negligible} we have $\mathbb{P}_n[\gamma_0(m_0 -\widehat m )]=o_{P_0}(n^{-1/2})$ which completes the proof.
\end{proof}

For the exponential family, we have the conditional mean as follows:
\begin{align*}
	\mathbb E_{\eta}[\Delta Y|D=0,X=x]=\frac{(A'\circ q^{-1})(\eta^m(x))}{a(q'\circ q^{-1})(\eta^m(x))}.
\end{align*}
Now the operator under consideration is $\Upsilon:= \frac{1}{a}A\circ q^{-1}$ and its derivative is given by $\Upsilon'= \frac{1}{a}(A'\circ q^{-1})/(q'\circ q^{-1})$. We can further simplify the expression to 
\begin{align*}
	\mathbb E_{\eta}[\Delta Y|D=0,X=x]=\Upsilon^{\prime}(\eta^m(x)),
\end{align*}
which is used in the proofs below. We write $L_n$ as some term which is a polynomial of $(\log n)$, whose exact value may change from line to line.

	\section{BvM under Low-level Conditions}\label{sec:sq:exp}

This section provides primitive conditions for the high-level assumptions in Sections \ref{sec:high} and \ref{sec:dr_high} to hold and derives the BvM Theorems under them. We focus on squared exponential process priors as an example of Gaussian process (GP) priors. Moreover, we consider specific smoothness classes to derive the explicit regularity conditions implied by our high-level assumptions. One of the main innovations in the nonparametric Bayesian literature is about the adaptivity. Nonparametric Bayesian methods based on flexible GP priors can adapt to unknown features of the underlying functions in theory and numerical results \citep{ghosal2017fundamentals}.

A Gaussian process is completely characterized by its mean and covariance functions \citep{williams2006gaussian}. Below we consider a GP prior, which has mean zero and the covariance function specified by $\mathbb{E}[W(s)W(t)]=\exp(-\Vert s-t\Vert_{\ell_2}^2)$. This is known as the squared exponential process prior, which is one of the most commonly used priors in applications; see \cite{williams2006gaussian} and \cite{murphy2023pml}. Following \cite{BLY2022}, we consider a rescaled Gaussian process $\big(W(a_nt):\,t\in [0,1]^p\big)$. Intuitively, $a_n^{-1}$ can be thought of as a bandwidth parameter. For a large $a_n$, the prior sample path $t\mapsto W(a_nt)$ is obtained by shrinking the original sample path $t\mapsto W(t)$. 

Let $\mathcal{C}^{s_m}([0,1]^p)$ denote a H\"older space with the smoothness index $s_m>0$. We illustrate our theory with the case where $m_0\in \mathcal{C}^{s_m}([0,1]^{p})$. Given such a H\"older-type smoothness condition, we choose
\begin{equation}\label{RescaleRate}
a_n\sim n^{1/(2s_m+p)}(\log n)^{-(1+p)/(2s_m+p)}.
\end{equation}			
The particular choice of $a_n$ mimics the corresponding kernel bandwidth based on the standard kernel smoothing method. Note that the minimax posterior contraction rate for the conditional mean function $m_\eta$ given by $\varepsilon_n=n^{-s_m/(2s_m+p)}(\log n)^{s_m(1+p)/(2s_m+p)}$; see Section 11.5 of \cite{ghosal2017fundamentals}.

\begin{proposition}[Unadjusted Squared Exponential Process Priors]\label{prop:exponential}
Suppose $m_0\in \mathcal{C}^{s_m}([0,1]^{p})$ and $\pi_0\in \mathcal{C}^{s_\pi}([0,1]^{p})$ under the smoothness conditions $\min(s_\pi,s_m)>p/2$. Consider the prior on $m$ given by $m(x) =q^{-1}\left(W^m(x)\right)$, where $W^m$ is the rescaled squared exponential process, with its rescaling parameter $a_n$ of the order in \eqref{RescaleRate}, combined with an independent Dirichlet process prior on $F$. Then, under Assumption \ref{Ass:unconfounded}, the posterior distribution for the ATT satisfies Theorem \ref{BvM:thm:standard}.
\end{proposition}

Proposition \ref{prop:exponential} makes explicit the smoothness requirements for the BvM Theorem to hold when standard Gaussian process priors are placed on the conditional mean function. This result shows that the smoothness of both the conditional mean function and the propensity score function must exceed $\dim(X)/2$. In situations where one is confident that these regularity conditions are met, no additional modifications to the Bayesian procedures are necessary to achieve the BvM result. 
 \begin{proposition}[Adjusted Squared Exponential Process Priors]\label{prop:exponential:dr}
The estimator $\widehat\gamma$ satisfies $\|\widehat\gamma\|_\infty=O_{P_0}(1)$ and $\|\widehat{\gamma}-\gamma_0\Vert_\infty= O_{P_0}\big((n/\log n)^{-s_\pi/(2s_\pi+p)}\big)$ for some $s_\pi>0$. Suppose $m_0\in \mathcal{C}^{s_m}([0,1]^{p})$ for some $s_m>0$ with $\sqrt{s_\pi \, s_m}>p/2$. Also, $\|\widehat{m}-m_0\Vert_{2, F_0}= O_{P_0}\big((n/\log n)^{-s_m/(2s_m+p)}\big)$. 
Consider a Dirichlet process prior on $F$ combined with the independent prior on $m$ given by $m(x) =q^{-1}\left(W^m(x) + \lambda\,\widehat \gamma(0,x)\right)$, where $W^m$ is the rescaled squared exponential process, with  rescaling parameter $a_n$ satisfying  \eqref{RescaleRate} and 
$\left(n/\log n\right)^{-s_m/(2s_m+p)}\lesssim u_n\varsigma_n$ for some deterministic sequence $u_n\to 0$, and  $\varsigma_n\lesssim 1$.
	Then,  under Assumption \ref{Ass:unconfounded}, the corrected posterior distribution for the ATT satisfies Theorem \ref{thm:BvM}.
\end{proposition}

Proposition \ref{prop:exponential:dr} requires $\sqrt{s_\pi \, s_m}>p/2$, which represents a trade-off between the smoothness requirement for $m_0$ and $\pi_0$. This corresponds to the \textit{double robustness} property; i.e., a lack of smoothness of the conditional mean function $m_0$ can be mitigated by exploiting the regularity of the propensity score $\pi_0$,  and vice versa.

The rest of the section presents proofs for Propositions \ref{prop:exponential} and \ref{prop:exponential:dr}.
\begin{proof}[Proof of Proposition \ref{prop:exponential}]
	Regarding the conditional mean function $m_{\eta}$, we consider the set $\mathcal{H}_{n}^m:= \left\{w:w\in \mathcal{B}^m_n,\Vert \Upsilon^\prime(w(\cdot))-m_0(\cdot)\Vert_{2, F_0}\leq \varepsilon_n\right\}$,
	where $\mathcal{B}_n^m$ is the set defined in (\ref{SieveSet}), that contains the Gaussian process with its posterior probability going to one.
 The posterior rate of contraction follows from the proof of Proposition 4.1 in \cite{BLY2022} without restricting the additional $\lambda$ used in the prior adjustment. The Donsker property is satisfied, following the calculation on Page 561 in the same proof from \cite{BLY2022}, if $s_{m}>p/2$.
	
	We show the prior stability by verifying Conditions (3.18) in Proposition 1 from \cite{ray2020causal}.
	Recall the definition of the ball in $\mathbb{H}^m$ centered at the true Riesz representer $\gamma_0$ given by 
\begin{align*}
\mathbb{H}^m(r_n):=\left\{h \in \mathbb{H}^m : \Vert h-\gamma_0\Vert_{\infty}\leq r_n\text{ and }\Vert h\Vert_{\mathbb{H}^m}\leq \sqrt{n}r_n\right\}
\end{align*}
for some rate $r_n$. 
We need to verify Assumption \ref{Assump:NonDRPS}, that is, there exists $\overline\gamma_n\in\mathbb{H}^m(\zeta_n)$ for  a sequence $\zeta_n=o(1)$ with $\sqrt{n}\varepsilon_n\zeta_n=o(1)$ where 
$\varepsilon_n=n^{-s_m/(2s_m+p)}L_n$ throughout the analysis. We need to consider two cases separately. 
	
	(I) If $s_{\pi}\geq s_{m}$ (meaning the Riesz representer is more regular than the conditional mean, hence it also belongs to $\mathcal{C}^{s_m}([0,1]^p)$ itself), we can simply take $\overline\gamma_n=\gamma_0$ and $\zeta_n=n^{-1/2}\Vert \gamma_0\Vert_{\mathbb{H}^m}$. Because the Donsker property already enforces $s_m>p/2$, it is easy to see that the condition $\sqrt{n}\varepsilon_n\zeta_n\to0$ is indeed satisfied.
	
	(II) If $s_{\pi}<s_m$ (meaning the Riesz representer is less regular than the conditional mean), we apply Lemma \ref{lemma:Decentering} with $\zeta_n=a_n^{-s_{\pi}}=n^{-s_\pi/(2s_m+p)}L_n$, so that $\Vert \overline\gamma_n\Vert_{\mathbb{H}^m}\leq C a_n^{p/2}$. It is straightforward to check that $\Vert \overline\gamma_n\Vert_{\mathbb{H}^m}\leq \sqrt{n}\zeta_n$ automatically holds if $s_{\pi}<s_m$. Finally, $\sqrt{n}\varepsilon_n\zeta_n\to0$ holds if and only if $n^{1/2-(s_\pi+s_m)/(2s_m+p)}L_n\to0$. The aforementioned condition holds if $s_\pi>p/2$.
\end{proof}

\begin{proof}[Proof of Proposition \ref{prop:exponential:dr}]
The posterior contraction follows from the proof of Proposition 4.1 in \cite{BLY2022}, if one restricts to the control group only.  Note that $\widehat\gamma$ is based on an auxiliary sample and hence we can treat $\widehat\gamma$ below as a deterministic function denoted by $\gamma_n$ satisfying the rate restrictions $\|\gamma_n\|_\infty=O(1)$ and $\|\gamma_n-\gamma_0\Vert_\infty= O\big((n/\log n)^{-s_\pi/(2s_\pi+p)}\big)$.  Regarding the conditional mean function $m_{\eta}$, we consider the set $\mathcal{H}_{n}^m:= \left\{w+\lambda\gamma_n:(w,\lambda)\in \mathcal{W}_{n}\right\}$, where for some constant $C>0$:
{\small \begin{align}\label{WnSet}
		\mathcal{W}_{n}:= \left\{(w,\lambda):w\in \mathcal{B}^m_n,|\lambda|\leq C\varsigma_n\sqrt{n}\varepsilon_n \right\}\cap \left\{(w,\lambda):\Vert \Upsilon^\prime(w+\lambda\gamma_n)-m_0\Vert_{2, F_0}\leq \varepsilon_n \right\},
\end{align}}%
where $\mathcal{B}_n^m$ is the set defined in (\ref{SieveSet}), that contains the Gaussian process with its posterior probability going to one.

We first verify Assumption \ref{Assump:Rate} with $\varepsilon_n=(n/\log n)^{-s_m/(2s_m+p)}$. The posterior contraction rate is shown in Lemma  C.3 of \cite{BLY2025supplement}. We next check the product rate condition, i.e., $\sqrt{n}\varepsilon_nr_n=o(1)$ for $r_n\sim  (n/\log n)^{-s_\pi/(2s_\pi+p)}$. This is satisfied if $2s_m/(2s_m+p)+2s_\pi/(2s_\pi+p)>1$, which can be rewritten as $\sqrt{s_\pi \, s_m}>p/2$.

We now verify Assumption \ref{Assump:Donsker}. It is sufficient to deal with the resulting empirical process $\mathbb{G}_n$.
From Lemma  C.5 in \cite{BLY2025supplement} we infer
\begin{align*}
	\E_0\sup_{\eta\in\mathcal{H}^m_n}\left|\mathbb{G}_n[(\gamma_n-\gamma_0)(m_\eta-m_0)] \right|&\leq 4\Vert \gamma_n-\gamma_0\Vert_{\infty}\E_0\sup_{\eta\in\mathcal{H}^m_n}\left|\mathbb{G}_n[m_\eta-m_0] \right|\\
	&\qquad+\Vert \gamma_n-\gamma_0\Vert_{2, F_0}\sup_{\eta\in\mathcal{H}_n}\Vert m_\eta-m_0\Vert_{2, F_0}\\
%	&\lesssim  (n/\log n)^{-s_\pi/(2s_\pi+p)}\E_0\sup_{\eta\in\mathcal{H}^m_n}\left|\mathbb{G}_n[m_\eta-m_0] \right|+(n/\log n)^{-s_\pi/(2s_\pi+p)}(n/\log n)^{-s_m/(2s_m+p)}\\
	&=  (n/\log n)^{-s_\pi/(2s_\pi+p)}\E_0\sup_{\eta\in\mathcal{H}^m_n}\left|\mathbb{G}_n[m_\eta-m_0] \right|+o(1)\\
	&=o(1),
\end{align*}
where the last line follows from the proof of Proposition 4.1 in \cite{BLY2022}.
Assumption \ref{Assump:Prior} (prior stability) follows from the proof on pages 561-562 in \cite{BLY2022}. 
\end{proof}
\let\clearpage\originalclearpage

\newpage
$\,$
\setcounter{page}{1}
\vskip 1cm

\begin{center}
{\LARGE Supplement to ``Semiparametric Bayesian Difference-in-Differences"\par\vspace{0.6\baselineskip}}
\end{center}
\begin{center}
{ \large Christoph Breunig}
  \qquad \quad{ \large Ruixuan Liu}
    \qquad \quad{ \large Zhengfei Yu}
\end{center}
\begin{center}
{ \large\today}
\end{center}

This supplementary appendix contains materials to support our main paper.
Appendix \ref{appendix:lfd} derives the least favorable direction for the ATT. Appendix \ref{appendix:misspecified} proves Lemma \ref{lemma:MisRateCluster}. Appendix \ref{appendix:crossfit} provides details for the cross-fitting algorithm.  
Appendix \ref{appendix:auxiliary:results} collects auxiliary results used for the derivations of our Bernstein-von Mises results.
Appendix \ref{details:implementation} provides details regarding the implementation.
Appendix \ref{appendix:simu} presents additional simulation evidence. 

In this supplement, $C>0$ denotes a generic constant, whose value might change line by line. We introduce additional subscripts when there are multiple constant terms in the same display. 

\addtocounter{section}{0}
%\appendix
\renewcommand{\thefigure}{\thesection\arabic{figure}}
\addtocounter{figure}{0}

\section{Least Favorable Direction}\label{appendix:lfd}
	Our prior correction through the Riesz representer $\gamma_0$ is motivated by the least favorable direction of Bayesian submodels. As we show below, this correction is indeed sufficient for our double-robust BvM theorem. We first provide least favorable direction (LFD) calculations of Bayesian submodels, which are closely linked to semiparametric efficiency derivations. We present the LFD for a more general class of likelihood functions known as the single-parameter exponential family that nests the normal likelihood as a special case. Let $\tilde{Z}=\left(D\Delta Y, D, X^\top\right)$ and use $f$ to denote its density function, as in Section \ref{sec:bayes_frame}.

Consider the one-dimensional submodel $t\mapsto \eta_t$ defined by the path
	\begin{align}\label{submodel}
	m_t(\cdot) =  q^{-1}(\eta^m+t\mathfrak{m})(\cdot )\quad\text{and}\quad 
	f_{t} (\cdot )= f(\cdot )e^{t\mathfrak{f}(\cdot)}\left(\int e^{t\mathfrak{f}(\tilde z)} f(\tilde z)\,\mathrm{d}\tilde z \right)^{-1},
	\end{align}
	for the given direction $(\mathfrak{m},\mathfrak{f})$ with $\int \mathfrak{f}(\tilde z) f(\tilde z)\,\mathrm{d}\tilde z =0$.
	The difficulty of estimating the parameter $\tau_{\eta_{t}}$ for the submodels depends on the direction $(\mathfrak{m},\mathfrak{f})$. Among them, let $ \xi_{\eta}= (\xi_{\eta}^{m},\xi_{\eta}^{f})$ be the \emph{least favorable direction} that is associated with the most difficult submodel, i.e., gives rise to the largest asymptotic optimal variance for estimating $\tau_{\eta_{t}}$. 
Let $p_{\eta_t}$ denote the joint density of $Z$ depending on $\eta_t:= (m_t, f_{t})$. Taking derivative of the logarithmic density $\log p_{\eta_t} (z)$ with respect to $t$ and evaluating at $t=0$ gives a score operator $B_\eta$, which we derive explicitly in the following proof. The least favorable direction is defined as the solution $\xi_\eta$ of the equation
$B_{\eta}\xi_{\eta}=\widetilde{\tau}_{\eta}$, see \citet[p.370]{ghosal2017fundamentals}, where $\widetilde{\tau}_{\eta}$ is the efficient influence function for estimation of the ATT given in \eqref{eif_att}.

\begin{lemma}\label{lemma:lfd}
Let Assumption \ref{Ass:unconfounded} be satisfied, then the least favorable direction for estimating the ATT in \eqref{att} is:
\begin{equation*}%\label{lfd:att}
	\xi_{\eta}(y,d,x)
	= \left(\frac{\gamma_\eta(0,x)}{a}, \frac{d\big(y-m_\eta(x)-\tau_\eta\big)}{\bar\pi_\eta}\right)
\end{equation*}
where the Riesz representer $\gamma_\eta$ is given in \eqref{riesz:def}.
\end{lemma}
\begin{proof}%[Proof of Lemma \ref{lemma:lfd}]
For the submodel defined in (\ref{submodel}), the definition of the joint density $p_\eta$ given \eqref{condpdf:eta} evaluated at the perturbation $\eta_t$ yields
	\begin{align*}
	\log p_{\eta_t}(y,d,x) &= 
%	d\log\Psi(\eta^{\pi} + t\mathfrak{p})(x)+ (1-d)\log (1-\Psi(\eta^{\pi} + t\mathfrak{p}))(x) \\
		 \log c(y) +ay(\eta^m+t\mathfrak{m})(x) -A(q^{-1}(\eta^m+t\mathfrak{m}))(x)\\
	& + t\mathfrak{f}(\tilde z) - \log\int e^{t\mathfrak{f}(\tilde z)} f(\tilde z)\,\mathrm{d}\tilde z + \log f(\tilde z).
	\end{align*}	
Taking derivative with respect to $t$ and evaluating at $t=0$ gives the score operator:
\begin{equation}\label{score_ate}
B_{\eta}(\mathfrak{m},\mathfrak{f}_X)(Z)= B_{\eta}^{m}\mathfrak{m}(Z) + B_{\eta}^{f}\mathfrak{f}_X(Z),
\end{equation}
where $B_{\eta}^{f}\mathfrak{f}_X(Z)= \mathfrak{f}_X(\tilde Z)$ and 
\begin{eqnarray*}
B_{\eta}^{m}\mathfrak{m}(Z) &=& (1-D)\left[a\Delta Y-\frac{A^\prime(m_{\eta}( X))}{q^\prime(m_{\eta}(X))}\right]\mathfrak{m}(X),\notag\\
&=& a(1-D)\left(\Delta Y-m_{\eta}( X)\right)\mathfrak{m}(X).
\end{eqnarray*}
In the last equation, we made use of the relation
\begin{eqnarray*}
A^\prime(m_{\eta}(x)) 
& =& q^\prime(m_{\eta}(x))\int ayc(y)\exp\left[q(m_{\eta}(x)) ay-A(m_{\eta}(x))\right]\mathrm{d}y\\
&=&  q^\prime(m_{\eta}(x))\,\mathbb{E}_{\eta}\left[a\Delta Y|D=0,X=x\right],
\end{eqnarray*}
which follows from the exponential family assumption. 
%In this case, there is a one-to-one correspondence between the conditional density function and the conditional mean function of the outcome given covariates. 
The efficient influence function for estimation of the ATT $\tau_\eta$ in \eqref{eif_att} is given by  $	\widetilde \tau_{\eta}(\Delta Y,D,X) = \gamma_\eta(D,X)(\Delta Y-m_\eta(X))-\frac{D}{\bar\pi_\eta}\tau_\eta$.
The score operator $B_{\eta}$ given in \eqref{score_ate} applied to 
\begin{equation*}%\label{lfd:att}
	\xi_{\eta}(y,d,x)
	= \left(\frac{\gamma_\eta(0,x)}{a}, \frac{d\big(y-m_\eta(x)-\tau_\eta\big)}{\bar\pi_\eta}\right)
\end{equation*}
yields  $B_{\eta}\xi_{\eta}=\widetilde{\tau}_{\eta}$. The result follows as $\widetilde{\tau}_{\eta}$ is the efficient influence function under $P_\eta$ given our assumptions following \cite{sant2020doubly}.
	\end{proof}	

\section{Posterior Contraction with Misspecified Likelihood and Clustered Data} \label{appendix:misspecified}
In this part, we present the formal result to establish the posterior contraction, when the working likelihood of the conditional distribution for the controlled outcome is standard normal. The true likelihood can be non-normal and admits the clustering structure. The asymptotic analysis lets the number of clusters $J\to\infty$, while the number of observations $n_j$ within each cluster remains bounded. Recall that $\mathcal{M}$ denotes the class of uniformly bounded functions mapping $\mathcal{X}$ to $\mathbb{R}$. We let $	N\left(\varepsilon,\mathcal{M},\Vert\cdot \Vert_{2,F_0}\right)$ be the covering number for any small $\varepsilon>0$.
\begin{lemma}\label{lemma:MisRateCluster}
	Let Assumption \ref{Ass:unconfounded} hold and consider the model \eqref{model:potential:control}, where $m_0\in\mathcal{M}$. Let the data consist of independent clusters $Z_j=\{Z_{jk}:k=1,\ldots,n_j\}$, $ j=1,\ldots,J$. Suppose that the
	cluster sizes are uniformly bounded, that is,	$\sup_j n_j \leq \bar n <\infty$. 	Suppose further that the unobservables $U_{jk}$ satisfy, for all $t\geq 0$,
	\begin{align}\label{ErrorMomentCluster}
		\sup_{j,k}\sup_{x\in\mathcal X}
		P_0\left(
		|U_{jk}|>t \mid D_{jk}=0, X_{jk}=x
		\right)
		\lesssim e^{-\underline{c} t^\beta}
	\end{align}
	for some constants $\underline{c}>0$ and $\beta>1$.
	If $\varepsilon_J$ is a sequence of positive numbers satisfying
	$\varepsilon_J\to 0$ and $J\varepsilon_J^2\to\infty$ such that, for some constant
	$c>0$ and all $J$,
	\begin{align}
		\Pi\left(
		m_{\eta}\in\mathcal{M}:
		P_0(m-m_0)^2\leq \varepsilon_J^2\right)
		&\geq e^{-cJ\varepsilon_J^2}, \label{PriorMass}
		\\
		N\left(
		\varepsilon_J,\mathcal{M},\Vert\cdot\Vert_{2,F_0}
		\right)
		&\leq e^{J\varepsilon_J^2}, \label{CoveringNumber}
	\end{align}
	then, we have
	\begin{align*}
		\Pi\left(
		m_{\eta}\in\mathcal{M}:
		P_0(m-m_0)^2\geq M\varepsilon_J^2
		\mid Z^{(J)}
		\right)
		\to_{P_0}0 ,
	\end{align*}
	for a sufficiently large constant $M>0$.
\end{lemma}
The conclusion in Lemma \ref{lemma:MisRate} also follows from the proof of Lemma \ref{lemma:MisRateCluster} by taking the number of clusters to be the same as the total sample size, when each cluster only has one observation. In this case, the working normal likelihood satisfies:
\begin{equation*}
	\log\frac{p_{c,\eta^m}}{p_{c,0}}(Z)=-\frac{1}{2}(m_{\eta}-m_0)^2(X)+U(m_{\eta}-m_{0})(X),
\end{equation*}
cf. Equation (4.1) in \cite{kleijn2006misInf}.

In our analysis, we maintain the assumption that the prior class of functions contains the true conditional mean $m_0$ and only examine the consequence of misspecifying the likelihood function. In other words, $m^*=m_0$, where $m^*$ stands for the function that minimizes the KL divergence of $p_{c,\eta^m}$ with respect to $p_{c,0}$. For general misspecified models, the complication in nonparametric Bayesian analysis is due to the covering number for testing under misspecification, which is introduced by \cite{kleijn2006misInf}. In the context of regression models with light-tailed errors, the problem becomes more tractable under reasonable assumptions with bounded conditional mean functions. Therein, one can obtain the posterior rate of contraction under standard conditions for the entropy number and prior mass, as in Theorem 4.1 of \cite{kleijn2006misInf}.

 \begin{proof}[Proof of Lemma \ref{lemma:MisRateCluster}]
 	The posterior is computed from the product working likelihood, whereas the
 	probability statements are taken under the true clustered law \(P_0\), under
 	which the clusters \(Z_1,\dots,Z_J\) are independent. Let \(P_0^c\) denote the distribution of the control sample, conditional on
 	\(D=0\), and let \(F_0^c\) denote the corresponding marginal distribution of
 	\(X\). 
 	
 	We first observe that, for every \(M>0\),
 	\begin{align*}
 		\E_0\left[e^{M |U_{jk}|}\mid D_{jk}=0, X_{jk}=x\right]
 		&= 1 + \int_0^\infty M e^{M t} \,
 		\P_0\left(|U_{jk}| > t \mid D_{jk}=0,  X_{jk}=x\right)\, \mathrm  dt .
 	\end{align*}
 	Consequently, under condition \eqref{ErrorMomentCluster}, there exists a
 	constant \(K>0\) such that
 	\begin{align*}
 		\E_0\left[e^{M |U_{jk}|}\mid D_{jk}=0, X_{jk}=x\right]
 		&\le 1 + M K \int_0^\infty e^{M t - \underline{c} t^\beta} \,\mathrm dt .
 	\end{align*}
 	Since \(\beta>1\), \(\underline{c} t^\beta\) dominates \(M t\) as \(t\to\infty\). Hence,
 	for every \(M>0\),
 	\begin{align*}
 		\sup_{j,k}\sup_{x\in\mathcal X}
 		\E_0\left[e^{M |U_{jk}|}\mid D_{jk}=0, X_{jk}=x\right]
 		<\infty .
 	\end{align*}
 	Moreover, the same argument implies that, for every \(k\in\mathbb N\) and every
 	\(M>0\),
 	\begin{align*}
 		\sup_{j,k}\sup_{x\in\mathcal X}
 		\E_0\left[|U_{jk}|^k e^{M |U_{jk}|}\mid D_{jk}=0, X_{jk}=x\right]
 		<\infty .
 	\end{align*}
 	
 	We verify the analogues of the inequalities in \((4.2)\) of
 	\cite{kleijn2006misInf} at the cluster level.  Since \(m_0\in\mathcal M\), the
 	pseudo-true regression function equals \(m_0\).	For cluster \(j\), define the working likelihood ratio
 	\begin{align*}
 	L_j(m_\eta,m_0)
 	=
 	\prod_{k=1}^{n_j}
 	\frac{p_{c,\eta^m}(Z_{jk})}{p_{c,0}(Z_{jk})}.
 	\end{align*}
 	Using $\E_0[U_{jk}\mid D_{jk}=0,X_{jk}]=0$, we obtain
 	\begin{align*}
 		P^c_0\left(\log L_j(m_{\eta},m_0)\right)
 		=
 		-\frac{1}{2}
 		\sum_{k=1}^{n_j}
 		P_0^c\left[(m-m_0)^2\right]
 		\leq
 		-\frac{1}{2}
 		P_0^c\left[(m-m_0)^2\right].
 	\end{align*}
 	Summing over independent clusters gives
 	\begin{align*}
 		\sum_{j=1}^J
 		P^c_0\left(\log L_j(m_{\eta},m_0)\right)
 		&\leq
 		-\frac{J}{2}
 		P_0^c\left[(m_{\eta}-m_0)^2\right],
 	\end{align*}
 	up to some multiplicative constant depending only on the bounded cluster
 	size.
 	
 	Next, using the inequality $
 	\left(\sum_{k=1}^{n_j} a_k\right)^2
 	\leq n_j \sum_{k=1}^{n_j} a_k^2$, and the uniform bound \(n_j\leq \bar n\), we obtain
 	\begin{align*}
 		P^c_0\left(\log L_j(m_0,m)\right)^2
 		&\leq
 		\bar n
 		\sum_{k=1}^{n_j}
 		P^c_0\left(
 		\log\frac{p_{c,0}(Z_{jk})}{p_{c,\eta^m}(Z_{jk})}
 		\right)^2
 		\\
 		&\lesssim
 		\sum_{k=1}^{n_j}
 		P_0^c\left[(m-m_0)^2\right]
 	\lesssim
 		P_0^c\left[(m-m_0)^2\right].
 	\end{align*}
 	Indeed, for each observation,
 	\begin{align*}
 		P^c_0\left(
 		\log\frac{p_{c,0}(Z_{jk})}{p_{c,\eta^m}(Z_{jk})}
 		\right)^2
 		&\leq
 		P_0^c\left[(m_{\eta}-m_0)^4\right]
 		+
 		2P_0^c\left[U_{jk}^2(m_{\eta}-m_0)^2\right]
 		\\
 		&=
 		P_0^c\left[(m_{\eta}-m_0)^4\right]
 		+
 		2P_0^c\left[
 		\E_0\left(U_{jk}^2\mid D_{jk}=0,X_{jk}\right)(m_{\eta}-m_0)^2
 		\right]
 		\\
 		&\lesssim
 		P_0^c\left[(m_{\eta}-m_0)^2\right],
 	\end{align*}
 	where we used the uniform boundedness of \(\mathcal M\) and
 	\[
 	\sup_{j,k}\sup_{x\in\mathcal X}
 	\E_0[U_{jk}^2\mid D_{jk}=0,X_{jk}=x]<\infty .
 	\]
 	
 	It remains to check the weighted square-log bound. For any \(\alpha\in(0,1)\),
 	boundedness of \(\mathcal M\) implies that, for a sufficiently large constant
 	\(C<\infty\),
 	\begin{align*}
 		\left(\frac{p_{c,\eta^m}(Z_{jk})}{p_{c,0}(Z_{jk})}\right)^\alpha
 		&\leq
 		\exp\left\{2\alpha(C^2+C|U_{jk}|)\right\}.
 	\end{align*}
 	Therefore,
 	\begin{align*}
 		& P^c_0\left[
 		\left(\log L_j(m_{\eta},m_0)\right)^2
 		L_j(m_{\eta},m_0)^\alpha
 		\right]
 		\\
 		&\quad \leq
 		\bar n
 		\sum_{k=1}^{n_j}
 		P_0\left[
 		\left(\log\frac{p_{c,\eta^m}(Z_{jk})}{p_{c,0}(Z_{jk})}\right)^2
 		\left(\frac{p_{c,\eta^m}(Z_{jk})}{p_{c,0}(Z_{jk})}\right)^\alpha
 		\prod_{\ell\neq k}
 		\left(\frac{p_{c,\eta^m}(Z_{j\ell})}{p_{c,0}(Z_{j\ell})}\right)^\alpha
 		\right].
 	\end{align*}
 	Since \(n_j\leq \bar n\) and \(\mathcal M\) is uniformly bounded, the remaining
 	within-cluster likelihood ratios are bounded by an exponential term involving
 	only finitely many \(|U_{j\ell}|\)'s. Hence, by the uniform exponential moment
 	bounds above,
 	\begin{align*}
 		 P_0\left[
 		\left(\log L_j(m_{\eta},m_0)\right)^2
 		L_j(m_{\eta},m_0)^\alpha
 		\right]
 		\lesssim
 		P_0^c\left[(m_{\eta}-m_0)^2\right].
 	\end{align*}
 	
 	In sum, the three inequalities in \((4.2)\) of \cite{kleijn2006misInf} continue to
 	hold after grouping observations into independent clusters, with constants that
 	may depend on the upper bound \(\bar n\) on the cluster size, but not on \(J\).
 	Since \(m_0\in\mathcal M\), the pseudo-true value is \(m^*=m_0\), and therefore
 	the additional orthogonality condition in Theorem 4.1 of
 	\cite{kleijn2006misInf} is automatically satisfied:
 	\begin{align*}
 		P_0^c\left[(m_{\eta}-m^*)(m^*-m_0)\right]
 		&=0 .
 	\end{align*}
 	
 	Applying Theorem 4.1 of \cite{kleijn2006misInf} to the independent clusters
 	\(Z_1,\ldots,Z_J\), the prior mass and entropy conditions with
 	\(J\varepsilon_J^2\) imply that, for a sufficiently large constant \(M>0\),
 	\begin{align*}
 		\Pi\left(
 		m\in\mathcal M:
 		P_0^c\left[(m_{\eta}-m_0)^2\right]\geq M\varepsilon_J^2
 		\mid Z^{(J)}
 		\right)
 		&\to_{P_0}0 .
 	\end{align*}
Note that the conditional density of covariates given $D=0$ is defined by
\begin{equation*}
	f_{X|D=0}(x):=\frac{P_0(D=0\mid X=x )f_X(x)}{P_0(D=0 )},
\end{equation*}
which means 
\begin{equation*}
	f_X(x)=\frac{1-\bar\pi_0}{1-\pi_0(x)}	f_{X|D=0}(x)
\end{equation*}
Under the overlapping condition in Assumption \ref{Ass:unconfounded}, we have $1-\bar\pi_0\leq 1$ and $(1-\pi_0(x))^{-1}\leq \epsilon^{-1}$. Therefore, the convergence in $L_2$ norm with respect to the conditional distribution given $D=0$ also holds with respect to the marginal distribution of covariates $F_X$. This proves the clustered version of the claim.
 \end{proof}

\section{Cross-Fitted Version}\label{appendix:crossfit}
Admittedly,  our cross-fitted procedure does not fall into the Bayesian paradigm in the strict sense. When we apply the Bayes' rule to obtain the fold-specific posterior draws of the conditional mean function, the other folds are used in a non-Bayesian way for pilot estimators. Also, when we move along different folds, their roles change. Nonetheless, the method remains pragmatically Bayesian, as we make use of the posterior draws assembled together after the cross-fitting process, and we build the corresponding credible set conditional on the observed data.  Algorithm \ref{algorithm_2_cf} describes the procedure

We denote the likelihood function in the $k$-th fold by $\ell_{n_k}(m^k_{\eta})$ for $k=1,2$. In similar spirit, the vectors $Z^{(n_k)}$ and $W^{(n_k)}$ denote the observable data and Bayesian bootstrap weight for the $k$--th fold. Regarding the pilot estimator, we denote the conditional mean and propensity score computed over the data excluding the $k$-th fold by $\widehat{m}^{(-k)}(\cdot)$ and $\widehat{\pi}^{(-k)}$. Thereafter, we define the estimated Riesz representer as 
\begin{eqnarray}\label{riesz_cf}
\widehat{\gamma}^{(-k)}(d,x):=\frac{d}{\hat{\bar\pi}^{(-k)}}-\frac{(1-d)\widehat{\pi}^{(-k)}(x)}{\hat{\bar\pi}^{(-k)}(1-\widehat{\pi}^{(-k)}(x))},
\end{eqnarray}
and the estimated bias term as $\widehat{b}^{k}_\eta =	\mathbb{P}_{n_{k}}[(\widehat m^{(-k)}-m^k_{\eta})\widehat{\gamma}^{(-k)}]$. We denote the frequentist DML estimator by $\widehat{\tau}^{k}$ computed in the $k$--th fold, which serves as the centering point in the BvM theorem.  

To outline the theoretical justification, we focus on the two-fold case without loss of generality. The extension to the general $K$--fold case, for any $K>2$, is straightforward. In this case, we write $\widehat{\gamma}^2:=\widehat{\gamma}^{(-1)}$ and $\widehat{\gamma}^1:=\widehat{\gamma}^{(-2)}$ to further lower the notational burden. Because the prior is created by plugging in pilot estimators computed over the other fold, the random variable $\tau_{\eta}^{1}$ depends on $Z^{(n_2)}$ via the propensity score adjusted prior $\Pi_{\widehat\gamma^{2}}(\cdot)$ by construction. Given the priors, we apply the Bayes' rule to obtain two sequences of independent draws for the posterior of the conditional mean, as well as the Bayesian bootstrap weights. The resulting posterior draws from the distributions $\mathcal{L}_{\Pi_{\widehat\gamma^{2}}}(\cdot\mid Z^{(n_k)})$, which encode the randomness of the conditional mean function $m_\eta$ and the bootstrap weights $W_n$, are made independent over different folds. To clarify the notion of independence, it is helpful to view the construction as producing a distribution for the ATT within each fold. Although each distribution depends on the full dataset, the Bayesian perspective allows us to interpret it as a posterior distribution specific to that fold, even though our correction incorporates information from other folds. One can then sample independently from each fold-specific distribution and average the resulting draws.

\begin{algorithm}[H]
\caption{Doubly robust Bayesian Procedure with $K$-fold cross-fitting}
\label{algorithm_2_cf}
\begin{algorithmic}
    \STATE \textbf{Input:} Data $Z_i=(\Delta Y_i,D_i,X_i^\top)^\top$ for $i=1,\dots,n$, number of posterior draws $S$.
     \STATE Take a $K$-fold random partition $(I_{k})_{k=1}^K$ of indices, $\left\{1,\dots,  n\right\}$,  each with size $n_k$.
  \FOR{$k=1,\ldots, K$}
   \STATE \textbf{Initial estimators:}   
   \STATE  (a) Use the subsample $\left\{(D_i,X_i^\top)\right\}_{i\notin I_{k}}$ to calculate the quantity $\hat{\bar\pi}^{(-k)}=\sum_{i\notin I_{k}}D_i/\sum_{j\neq k}n_j$,  construct the estimated propensity score function $x\rightarrow \widehat \pi^{(-k)}(x)$,  and obtain $\widehat \gamma^{(-k)}(d,x)$ following (\ref{riesz_cf}). 
    \STATE  (b) Use the subsample $\left\{(\Delta Y_i,X_i^\top): D_i=0\right\}_{i\notin I_{k}}$ to construct the estimated conditional mean function $x\rightarrow\widehat m^{(-k)} (x)$ following (\ref{m:est}).
    \STATE \textbf{Prior Specification:}  Set the adjusted prior:
\begin{align}\label{prior:ps_cf}
m_{\eta}(x) = q^{-1}\left(\eta^m(x)\right),  \quad \eta^m(x)=W^m(x) + \lambda\,\widehat \gamma^{(-k)}(0,x), 
\end{align}    
where $W^m$ is the Gaussian process independent of $\lambda \sim N(0,\varsigma_n^2)$,  $\varsigma_n=\nu\log n_{c,k}/(\sqrt{n_{c,k}}\, \Gamma_{n,k})$, $\nu$ is the hyperparameter in (\ref{SE_cov}), $n_{c,k}=\sum_{i\in I_{k}}(1-D_i)$, and $\Gamma_{n,k}=\sum_{i\in I_{k}}\vert\widehat\gamma^{(-k)}(0,X_i)(1-D_i)\vert/n_{c,k}$.

\textbf{Posterior Computation:} 
    \FORNN{$s=1,\ldots, S$}
        \STATE  (a)
        Generate the $s$-th draw of the posterior of $(m_\eta(X_i))_{i\in I_{k}}$ using the adjusted prior in (\ref{prior:ps_cf}) and the subsample $\left\{(\Delta Y_i,X_i^\top): D_i=0\right\}_{i\in I_{k}}$; denote it as $(m^{k,s}_{\eta}(X_i))_{i\in I_{k}}$.
        \STATE (b)  Draw Bayesian bootstrap weights $M^{k,s}_{i}=e^{k,s}_i/\sum_{j\in I_{k}} e^{k,s}_j$ where $e_i^{k,s} \stackrel{\text{iid}}{\sim} \textup{Exp}(1)$, $i\in I_{k}$.
      
        \STATE (c) Calculate the corrected posterior draw for the ATT: $\check{\tau}_\eta^{k,s}=\tau_\eta^{k,s}-\widehat{b}^{k,s}_{\eta}$
        where 
        \begin{equation}
        \tau_\eta^{k,s}=\frac{\sum_{i\in I_{k}} M_{i}^{k,s}D_i \big(\Delta Y_i-m_\eta^{k,s}(X_i)\big)}{\sum_{i\in I_{k}} M_{i}^{k,s} D_i},
        \end{equation}
and 
\begin{align}\label{recentering_term_cf}
\widehat{b}^{k,s}_{\eta}=\frac{1}{n_k}\sum_{i\in I_{k}} \widehat{\gamma}^{(-k)}(D_i,X_i) (\widehat m^{(-k)}-  m_{\eta}^{k,s})(X_i).
\end{align}
    \ENDFOR
     \ENDFOR
    \STATE \textbf{Output: $\{\check{\tau}_\eta^s=\frac{1}{K}\sum_{k=1}^K \check\tau_\eta^{k,s}:s=1,\ldots,S\}$} 
\end{algorithmic}
\end{algorithm}

	\begin{theorem}\label{thm:BvMCF}
	Let Assumptions \ref{Ass:unconfounded}, \ref{Assump:NonDRSE}(i), \ref{Assump:Rate}, \ref{Assump:Donsker}, and \ref{Assump:Prior} hold. Consider the propensity score adjusted prior \eqref{prior:ps_cf} on $\eta^m$ and an independent Dirichlet process prior on $F$. Then we have
	\begin{equation*}
		d_{BL}\left(	\mathcal{L}_{\Pi_{\widehat{\gamma}^2}}\left(\frac{\sqrt{n_1}}{\sqrt 2}(\tau_{\eta}^{1}-\widehat{\tau}^{1}-\widehat{b}^1_{\eta})|Z^{(n_1)}\right)\ast	\mathcal{L}_{\Pi_{\widehat{\gamma}^1}}\left(\frac{\sqrt{n_2}}{\sqrt 2}(\tau_{\eta}^{2}-\widehat{\tau}^{2}-\widehat{b}^2_{\eta})|Z^{(n_2)}\right), N(0,\textsc v_0) \right)\to_{P_0} 0.
	\end{equation*}
\end{theorem}

\begin{proof}
To establish the desired conditional convergence of the corresponding distributions, we consider the posterior Laplace transforms
defined for $k=1,2$ by:
\begin{align*}
	I_{n_k}(t)&=\mathbb{E}_{\Pi_{\widehat\gamma^{3-k}}}[e^{t\sqrt{n_k}/\sqrt{2}(\check{\tau}_{\eta}^{k}-\widehat\tau^k-\widehat{b}^k_{\eta})}|Z^{(n_k)}]\\
	&=\frac{\iint e^{t\sqrt{n_k}/\sqrt 2(\check{\tau}_{\eta}^{k}-\widehat\tau^k-\widehat{b}^k_{\eta}) } \mathrm{d}\Pi_W(W^{(n_k)}\mid Z^{(n_k)})e^{\ell_{n_k}(m^k_{\eta})}\mathrm{d}\Pi_{\widehat\gamma^{3-k}}(m^k_{\eta})}{\int e^{\ell_{n_k}(m^k_{\eta})} \mathrm{d}\Pi_{\widehat\gamma^{3-k}}(m^k_{\eta})}.
\end{align*}
The conditional Laplace transform $I_{n_k}(t)$ depends not only on the observation $Z^{(n_k)}$ itself, but also on $Z^{(n_{-k})}$ through the pilot estimation for $k=1,2$. In other words, both $I_{n_1}(t)$ and $I_{n_2}(t)$ are functions of the entire dataset $(Z^{(n_1)}, Z^{(n_2)})$. Despite this complication, we have
\begin{align*}
I_{n_1}(t)\times 	I_{n_2}(t)&=\prod_{k\in\{1,2\}}\frac{\iint e^{t\sqrt{n_k}/\sqrt 2(\check{\tau}_{\eta}^{k}-\widehat\tau^k-\widehat{b}^k_{\eta}) } \mathrm{d}\Pi_W(W^{(n_k)}\mid Z^{(n_k)})e^{\ell_{n_k}(m^k_{\eta})}\,\mathrm{d}\Pi_{\widehat\gamma^{3-k}}(m^k_{\eta})}{\int e^{\ell_{n_k}(m^k_{\eta})}\, \mathrm{d}\Pi_{\widehat\gamma^{3-k}}(m^k_{\eta})}\\
		&\to_{P_0} \exp(\textsc v_0t^2/4)\times \exp(\textsc v_0t^2/4)=\exp(\textsc v_0t^2/2),
\end{align*}
where the convergence in probability has already been established in the proof of Theorem \ref{thm:Debias}. This verifies convergence in probability of the product of the two posterior Laplace transforms for every $t$ in a neighborhood of zero, which in turn implies conditional convergence of the convoluted law of two fold-specific posteriors to the desired normal limiting distribution in the bounded Lipschitz distance.
\end{proof}

	\section{Auxiliary Results}\label{appendix:auxiliary:results}
\subsection{Useful Lemmas}

\subsubsection{Results on Conditional Weak Convergence}
We first present a useful result from part of Theorem 1.13.1 in \cite{van1996empirical} concerning conditional weak convergence. To do so, we introduce a sequence of random variables $S_n$, a subfield $\mathcal{H}_n$ of their associated $\sigma$-algebra, and a Borel probability measure $L$.
\begin{lemma}\label{lemma:WeakConv}
 The following two statements are equivalent:
(i) $d_{BL}(\mathcal{L}(S_n\mid \mathcal{H}_n),L)\overset{P_0}{\to} 0$; (ii) for every $t$ in some neighborhood of $0$, 
\begin{equation*}
	\mathbb{E}[e^{t S_n}\mid \mathcal{H}_n]\overset{P_0}{\to} \int e^{t x}dL(x)<\infty.
\end{equation*}
\end{lemma}

We now state a conditional Slutsky result, which coincides with Lemma 10 in \cite{yiu2023corrected} and is included here for completeness.
\begin{lemma}\label{lemma:Slutsky}
Let $Z^{(n)}=(Z_1,\ldots,Z_n)$ be \textit{i.i.d.} variables from a distribution $P_0$ on a Polish sample space $(\mathcal{Z},\mathcal{A})$. Suppose that $(P_n)_n$ is a sequence of random probability measures on $(\mathbb{R}^2,\mathcal{B}(\mathbb R^2))$ such that $P_n$ is $\sigma(Z^{(n)})$-measurable for each $n$.
Let $(U_n,V_n)$ be variables each taking values in $\mathbb R$ with $(U_n,V_n)|P_n\sim P_n$ and denote the marginals by $P_n^U$ and $P_n^V$ for $U_n$ and $V_n$ respectively. Suppose that 
\begin{eqnarray*}
	d_{BL}(P_n^U,P^U)\rightarrow^{P_0}0\\
	d_{BL}(P_n^V,\delta_{\{c\}})\rightarrow^{P_0}0,
\end{eqnarray*}
where $P^U$ is a fixed probability measure on $(\mathbb{R},\mathcal{B}(\mathbb R))$, and $c$ is a fixed constant in $\mathbb R$. Then we have
\begin{eqnarray*}
	d_{BL}(\mathcal L(U_n+V_n| P_n),\mathcal L(U_n+c| P_n))\rightarrow^{P_0}0\\
	d_{BL}(\mathcal L(U_nV_n| P_n),\mathcal L(cU_n| P_n))\rightarrow^{P_0}0.
\end{eqnarray*}
\end{lemma}

	We now state the following generalization of Lemma 1 from \cite{ray2020causal}, where the function $g(\cdot)$ may depend on random variables beyond the covariates $X$. A close inspection of their proof shows that the argument remains valid when $g(\cdot)$ is a function of $Z = (Y, D, X^\top)^\top$.
\begin{lemma}\label{lemma:DP}
	Suppose $\mathcal{G}_n$ is a sequence of separable classes of measurable functions, such that
	\begin{equation*}
		\sup_{g\in\mathcal{G}_n}\left|\frac{1}{n}\sum_{i=1}^ng(Z_i)-\mathbb{E}_0[g(Z)]\right|\to_{P_0} 0,
	\end{equation*}
and there exists an envelope function $G_n$ such that $\mathbb{E}_0[G_n^{2+\delta}]=O(1)$, for some $\delta>0$. Then for every $t$ in a sufficiently small neighborhood of $0$, 
	\begin{equation*}
		\sup_{g\in\mathcal{G}_n}
		\left|\mathbb{E}_0\left[\exp\left(t\sqrt{n}\sum_{i=1}^n (M_{i}-1/n)g(Z_i)\right)\Bigm\vert Z^{(n)}\right]-\exp\left(t^2Var_0(g(Z))/2\right) \right|\to_{P_0} 0.
	\end{equation*}
\end{lemma}
\subsubsection{Results on Gaussian Processes}
Consider a mean-zero Gaussian random element $W$ in a separable Banach space $\mathbb B$ defined on a probability space $(\Omega,\mathcal U, P)$ and $\mathbb H^m$ its RKHS. The dual space $\mathbb{B}^*$ of the Banach space $\mathbb{B}$ consists of all continuous and linear maps $b^*:\mathbb{B}\mapsto \mathbb{R}$. Define a map $U$ by $U(Sb^*)=b^*(W)$, $b^*\in\mathbb B^*$. By the definition of RKHS the map $S\mathbb B^*:\mathbb H\mapsto \mathbb{L}_2(\Omega,\mathcal U, P)$ is an isometry. Let $U:\mathbb H\mapsto \mathbb{L}_2(\Omega,\mathcal U, P)$ be its extension to the full RKHS.
If $W$ is a mean-zero Gaussian random element in a separable Banach space and $h$ is an element of its RKHS, then by the Cameron-Martin Theorem, the distributions $P^{W+h}$ and $P^W$ of $W+h$ and $W$ on $\mathbb B$ are equivalent with Radon-Nikodym density
\begin{equation}\label{CMThm}
	\frac{dP^{W+h}}{dP^W}(W)=\exp\left(Uh-\frac{1}{2}\Vert h\Vert_{\mathbb H}^2\right), ~~~\text{almost surely}.
\end{equation}
Regarding the uncorrected prior, we consider the Gaussian process prior $W^m$ for the conditional mean as Borel-measurable maps in the Banach space $C([0,1]^p)$, equipped with the uniform norm $\Vert\cdot\Vert_{\infty}$. Such a process also determines a reproducing kernel Hilbert space (RKHS) ($\mathbb{H}^m,\Vert\cdot\Vert_{\mathbb{H}^m}$) and a so-called concentration function $\phi_{\eta_0^m}$, defined as, for $\varepsilon>0$,
\begin{equation}
	\phi_{\eta_0^m}(\varepsilon):= \inf_{h\in\mathbb{H}^m:\Vert h-\eta_0^m\Vert_{\infty}<\varepsilon}\Vert h\Vert_{\mathbb{H}^m}^2-\log P(\Vert W^m\Vert_{\infty}<\varepsilon).
\end{equation}

The posterior contraction rate $\varepsilon^m_n$ for such a Gaussian process prior is determined by the solution of the equation:
\begin{equation}
	\phi_{\eta_0^m}(\varepsilon^m_n)\sim n(\varepsilon^m_n)^2.
\end{equation}
Each Gaussian process comes with an intrinsic Hilbert space determined by its covariance kernel. This space is critical in analyzing the rate of contraction for its induced posterior. Consider a Hilbert space $\mathbb{H}$ with inner product $\langle\cdot,\cdot\rangle_{\mathbb{H}}$ and associated norm $\Vert\cdot\Vert_{\mathbb{H}}$. $\mathbb{H}$ is an Reproducing Kernel Hilbert Space (RKHS) if there exists a symmetric, positive definite function $k:\mathcal{X}\times \mathcal{X}\mapsto \mathbb{R}$, called a kernel, that satisfies two properties: (i) $k(\cdot,\bm{x})\in\mathbb{H}$ for all $\bm{x}\in \mathcal{X}$; and (ii) $f(\bm{x})=\langle f,k(\cdot,\bm{x})\rangle_{\mathbb{H}}$ for all $\bm{x}\in \mathcal{X}$ and $f\in \mathbb{H}$. It is well-known that every kernel defines a RKHS and every RKHS admits a unique reproducing kernel.

Let $\mathbb{H}^{a_n}_1$ be the unit ball of the RKHS for the rescaled squared exponential process and let $\mathbb{B}_1^{s_m,p}$ be the unit ball of the H\"older class $\mathcal{C}^{s_m}([0,1]^p)$ in terms of the supremum norm $\Vert\cdot\Vert_{\infty}$. 
We introduce a class of functions $\mathcal{B}^m_n$ which is shown to contain the Gaussian process $W$ which sufficiently large probability, and is given by
\begin{equation}\label{SieveSet}
	\mathcal{B}^m_n:= \varepsilon_n \mathbb{B}_1^{s_m,p}+M_n\mathbb{H}_1^{a_n},
\end{equation}
where $a_n=n^{1/(2s_m+p)}(\log n)^{-(1+p)/(2s_m+p)}$, $\varepsilon_n=n^{-s_m/(2s_m+p)}(\log n)^{s_m(1+p)/(2s_m+p)}$, and $M_n=-2\Phi^{-1}(e^{-Cn\varepsilon_n^2})$. 
For notational simplicity, we suppress the dependence of the rescaled Gaussian process on the rescaling parameter $a_n$.

\begin{lemma}[Lemma 11.56 in \cite{ghosal2017fundamentals}]\label{lemma:Decentering}
Consider the rescaled squared exponential process with rescaling factor $a$. For any $s>0$ and $w\in \mathcal{C}^{s}([0,1]^p)$, there exist constants $C_1$ and $C_2$ (depending only on $w$) such that
\begin{equation}
	\inf_{h:\Vert h-w\Vert_{\infty}\leq C_1a^{-s}}\Vert h\Vert_{\mathbb{H}^a}^2\leq C_2 a^p.
\end{equation}
\end{lemma}

\subsection{Expansions}
Recall the definition of the score operator
\begin{equation*}
B_{\eta}^m \mathfrak{m}(Z)=(1-D)(\Delta Y-m_{\eta}(X))\mathfrak{m}(X).
\end{equation*}
The least favorable direction for the conditional mean in the control group is $\gamma_{\eta}(0,x)=-\frac{1}{\bar\pi_\eta}\frac{\pi_\eta(x)}{1-\pi_\eta(x)}$. To simplify the notation in the following derivation, we also write $\gamma_{c,\eta}(d,x)=(1-d)\gamma_{\eta}(0,x)$ to signify this relationship to the control group.
Given any $\eta^m$, the perturbation we consider is as follows:
\begin{equation*}
	\eta^m_t(x):=\eta^m(x)-t\gamma_{\eta_0}(0,x)/\sqrt{n}.
\end{equation*}
Below we denote the conditional density function $p_{c,\eta}(y,d,x)=f^{1-d}_{(\Delta Y|D,X),\eta}(y,0,x)$.
From the proof of Lemma \ref{lemma:lfd} we observe
\begin{align*}
	\mathbb E_{\eta}[\Delta Y\mid D=0,X=x]=\frac{(A'\circ q^{-1})(\eta^m(x))}{a(q'\circ q^{-1})(\eta^m(x))}.
\end{align*}
Now the operator under consideration is $\Upsilon= \frac{1}{a}A\circ q^{-1}$ and its derivative is given by $\Upsilon'= \frac{1}{a}(A'\circ q^{-1})/(q'\circ q^{-1})$.
For the exponential family under consideration, the first and second order cumulants (conditional on covariates) are:
\begin{align}\label{ExpMoments}
	\mathbb E_{\eta}[\Delta Y\mid D=0,X=x]=\Upsilon^{\prime}(\eta^m(x)),~~~Var_{\eta}(\Delta Y\mid D=0,X=x)=\Upsilon^{(2)}(\eta^m(x)).
\end{align}
The conditional variance formula also shows the convexity of $\Upsilon(\cdot)$. Related proofs can be found page 19 in Appendix F of \cite{BLY2025supplement}.
	\begin{lemma}\label{lemma:Taylor}
		Let Assumptions \ref{Ass:unconfounded} and \ref{Assump:Rate} hold. Then, we have uniformly for $\eta\in\mathcal H_n$:
		\begin{align*}
			\log p_{c,\eta^m}-\log p_{c,\eta^m_t} 	&=\frac{t}{\sqrt{n}}[\gamma_{c,0}(m_0-m_\eta)]+\frac{t^2}{2n}\left[ \Upsilon^{(2)}(\eta_0^m)\gamma_{c,0}^2\right]+O_{P_0}(n^{-3/2}).
			\end{align*}
	\end{lemma}
\begin{proof}
	For this purpose, we use the notation $g(u)=\log p_{c, \eta_{c,u}^m}$ for $u\in[0,1]$.
Specifically, in the one-parameter exponential family case, we have 
\begin{equation*}
	\log p_{c, \eta_u^m}(y,d,x)=(1-d)\left[y\eta_{u}^m(x)-\Upsilon(\eta_u^m(x))+\log c(y)\right].
\end{equation*}
By the definition of $\Upsilon(\cdot)$, we can obtain the first to third order derivatives of $g$ as
\begin{align*}
	g^\prime(0)&=\frac{t}{\sqrt{n}}\gamma_{c,0}\rho^{\Upsilon^\prime(\eta^m)}=\frac{t}{\sqrt{n}}\gamma_{c,0}\rho^{m_\eta},\\
	g^{(2)}(0)&=\frac{t^2}{n}\gamma^2_{c,0}\Upsilon^{(2)}(\eta^m), ~~~g^{(3)}(\tilde{u})=\frac{t^3}{n^{3/2}}\gamma^3_{c,0}\Upsilon^{(3)}(\eta_{\tilde{u}}^m),
\end{align*}
where $\tilde{u}$ is some intermediate value between $0$ and $1$. In the above calculation, we have made use of (\ref{ExpMoments}).
\end{proof}
	\begin{lemma}\label{lemma:likelihood}
	Let Assumptions \ref{Ass:unconfounded} and \ref{Assump:Rate} hold. Then, we have uniformly for $\eta\in\mathcal H_n$:
	\begin{align*}
	\ell_n^m(\eta^m)-\ell_n^m(\eta^m_t)=&\frac{t}{\sqrt n}\sum_{i=1}^n\frac{1-D_i}{\bar\pi_0} \frac{\pi_0(X_i)}{1-\pi_0(X_i)} \left(\Delta Y_i-m_0(X_i)\right)\\
	&+\frac{t^2}{2}\E_0\left[ \frac{1-D}{\bar\pi_0^2} \frac{\pi_0^2(X)}{(1-\pi_0(X))^2}(\Delta Y-m_0(X))^2\right]\\
	&-\frac{t}{\sqrt n}\sum_{i=1}^n\frac{1-D_i}{\bar\pi_0} \frac{\pi_0(X_i)}{1-\pi_0(X_i)} (m_0(X_i)-m_{\eta}(X_i))+o_{P_0}(1).
\end{align*}
	\end{lemma}
	\begin{proof}
	We start with the following decomposition:
\begin{align*}
		\ell_n^m(\eta^m)-\ell_n^m(\eta^m_t)=&t\mathbb{G}_n[\gamma_{c,0}\rho^{m_0}]+ \sqrt{n}\mathbb{G}_n[\log p_{c, \eta^m}-\log p_{c,\eta^m_t}-\frac{t}{\sqrt{n}}\gamma_{c,0}\rho^{m_0}]\\
		&	+ n P_0[\log p_{c,\eta^m}-\log p_{c,\eta^m_t}],
	\end{align*}
	where $\gamma_{c,0}(d,x)=-\frac{1-d}{\bar\pi_0} \frac{\pi_0(x)}{1-\pi_0(x)} $ and $\rho^{m_0}(Z)=\Delta Y-m_0(X)$.
		Then, we apply the expansion in Lemma \ref{lemma:Taylor} so that
	\begin{align*}
	&	\sqrt{n}\mathbb{G}_n[\log p_{c,\eta^m}-\log p_{c, \eta^m_t}-\frac{t}{\sqrt{n}}\gamma_{c,0}\rho^{m_0}]\\
		=&t\mathbb{G}_n[\gamma_{c,0}(m_\eta-m_0)]+\frac{t^2}{2}(\mathbb{P}_n-P_0)[\gamma^2_{c,0}\Upsilon^{(2)}(\eta^m)]+ o_{P_0}(n^{-1/2}),
	\end{align*}
	uniformly in $\eta^m\in\mathcal{H}^m_n$. The second term on the right hand side vanishes because of the $P_0$-Glivenko-Cantelli (GC) property and the permanence GC theorem, i.e., Theorem 2.10.5 in \cite{van1996empirical}.
	Then, we infer for the stochastic equicontinuity term that
	\begin{align*}
	\sqrt{n}\mathbb{G}_n[\log p_{c,\eta^m}-\log p_{c, \eta^m_t}-\frac{t}{\sqrt{n}}\gamma_{c,0}\rho^{m_0}]
	=t\mathbb{G}_n[\gamma_{c,0}(m_\eta-m_0)]+o_{P_0}(1),
	\end{align*}
uniformly in $\eta^m\in\mathcal{H}^m_n$. We can thus write
	\begin{align*}
	\ell_n^m(\eta^m)-\ell_n^m(\eta^m_t)&=t\mathbb{G}_n[\gamma_{c,0}\rho^{m_0}]+t\mathbb{G}_n[\gamma_{c,0}(m_0-m_\eta)]+n P_0[\log p_{c, \eta^m}-\log p_{c,\eta^m_t}]+o_{P_0}(1),
	\end{align*}
uniformly in $\eta^m\in\mathcal{H}^m_n$ and we control $n P_0[\log p_{\eta^m}-\log p_{\eta^m_t}]$ in the remainder of the proof. 
Specifically, we apply Lemma \ref{lemma:prob:expansion} and obtain uniformly for $\eta\in\mathcal H_n$, 
\begin{align*}
	n P_0[\log p_{c,\eta^m}-\log p_{c,\eta^m_t}] 	&=P_0[\gamma_{c,0}(m_0-m_\eta)]+t^2P_0\left[ \Upsilon^{(2)}(\eta_0^m)\gamma_{c,0}^2\right]+o_{P_0}(1).
\end{align*}
Now using that
	\begin{align*}
	\frac{1}{\sqrt n}\sum_i\gamma_{c,0}(D_i,X_i)(m_0(X_i)-m_{\eta}(X_i))=\mathbb{G}_n[\gamma_{c,0}(m_0-m_\eta)]
	-\sqrt n P_0[\gamma_{c,0}(m_0-m_\eta)]
	\end{align*}
	the result follows.
	\end{proof}	
	
		\begin{lemma}\label{lemma:prob:expansion}
		Let Assumptions \ref{Ass:unconfounded} and \ref{Assump:Rate} hold. Then, we have uniformly for $\eta\in\mathcal H_n$:
\begin{align*}
	n P_0\log \left(\frac{p_{c,\eta^m}}{ p_{c,\eta^m_t}}\right) &=t\sqrt{n} P_0[\gamma_{c,0}(m_0-m_\eta)]+t^2P_0\left[ \Upsilon^{(2)}(\eta_0^m)\gamma_{c,0}^2\right]+o_{P_0}(1).
\end{align*}
\end{lemma}
	\begin{proof}
First, we note that  $P_0[\gamma_{c,0} \rho^{m_0}]=0$ and
\begin{align*}
	P_0&\left(B_{\eta_0}^m \gamma_{c,0}\right)^2=  \mathbb{E}_0\left[\left(B_{\eta_0}^m\left(-\frac{1-D}{\bar\pi_0} \frac{\pi_0(x)}{1-\pi_0(x)}\right)\right)^2\right] \\
	& =\mathbb{E}_0\left[(\Delta Y-m_0( X))^2 \frac{1-D}{\bar\pi_0^2} \frac{\pi_0^2(X)}{\left(1-\pi_0(X)\right)^2}\right] \\
	%& =\mathbb{E}_0\left[\frac{1-D}{\pi_0^2} \frac{\pi_0^2(X)}{\left(1-\pi_0(X)\right)^2} m_0(0, X)\left(1-m_0(0, X)\right)\right] \\
	& =P_0\left[ \Upsilon^{(2)}(\eta_0^m)\gamma_{c,0}^2\right]
\end{align*}
using $Var_{\eta}(\Delta Y\mid D=0,X=x)=\Upsilon^{(2)}(\eta^m(x))$ as in (\ref{ExpMoments}).
Recall the function $g(u)=\log p_{c, \eta_{c,u}^m}$ for $u\in[0,1]$ in Lemma \ref{lemma:Taylor}. Based on the expansion therein, and the posterior convergence of $\eta^m$, it can be expressed as
\begin{align*}
	nP_0g^{(2)}(0)=&t^2\mathbb{E}_0[\gamma_0^2(0,X)\Upsilon^{(2)}(\eta_0^m(X))]+o_{P_0}(1)\\
%	=t^2\mathbb{E}_0[\gamma_0^2(0,X)Var_0(Y|D=0,X)]+o_{P_0}(1)\\
	=&t^2\mathbb{E}_0[\gamma_0^2(0,X)\left((\Delta Y-m_0(X)\right)^2]+o_{P_0}(1)=t^2P_0(B_{c,0}^m\gamma_{c,0})^2+o_{P_0}(1),
\end{align*}
where the score operator $B_0^m=B_{\eta_0}^m$ is given in the proof of Lemma \ref{lemma:lfd}. 
Consequently, we obtain, uniformly for $\eta\in\mathcal H_n$, 
\begin{align*}
	n P_0[\log p_{c, \eta^m}-\log p_{c, \eta^m_t}] &= -n (P_0g^\prime(0)+ P_0g^{(2)}(0))+o_{P_0}(1)\\
	&=t\sqrt{n} P_0[\gamma_{c,0}(m_0-m_\eta)]+ t^2P_0(B_{c,0}^m\gamma_{c,0})^2+o_{P_0}(1),
\end{align*}
which leads to the desired result.
	\end{proof}

The next lemma is about smaller order terms in the proof of our BvM theorems. Note that the posterior law of $F_{\eta}$ coincides with the Bayesian bootstrap. Here, the negligibility of those terms refers to the randomness with respect to the Bayesian bootstrap weights (for which we use $P_M$ to highlight this dependence), conditional on the observed data. We refer readers to Page 2891 in \cite{cheng2010bootstrap} for comprehensive discussion about disentangling the sources of randomness coming from the observed data and the Bayesian bootstrap weights. Formally, we define $\Delta_n=o_{P_M}(1)$ in $P_0$-probability if, for any small positive $\epsilon$ and $\delta$, $P_0\left(P_{M|Z^{(n)}}(|\Delta_n|>\epsilon )>\delta \right)\to 0$
In addition, we define $\Delta_n=O_{P_M}(1)$ in $P_0$-probability if, for any small positive $\delta$, there exists a large enough $C$ such that
	$P_0\left(P_{M|Z^{(n)}}(|\Delta_n|>C )>\delta \right)\to 0$.
For the next result, recall the definition of the remainder term $R_{n,\eta}$ given in the proof of Theorem \ref{BvM:thm:standard} by
\begin{align*}
		R_{n,\eta}=\sqrt{n}\mathbb{P}_n\left[\left(\frac{D-\pi_0(X)}{1-\pi_0(X)}\right)\left(\Delta Y-m_{0}(X)\right)-D\tau_0\right]\left(1-\upsilon_{\eta}\right),
	\end{align*}
where $\upsilon_{\eta}=\mathbb{E}_{\eta}[D]/\bar\pi_0$.	

\begin{lemma}\label{lemma:Remainder}
Under Assumption \ref{Assump:NonDRSE}(i), it holds
\begin{enumerate}
\item[(i)] $\sup_{\eta\in \mathcal{H}_n}\vert R_{n,\eta}\vert=o_{P_M}(1)$ in $P_0$-probability and 
\item[(ii)] $\sup_{\eta\in \mathcal{H}_n}\vert\sqrt{n}(\upsilon_{\eta}-1)b_{0,\eta}\vert=o_{P_M}(1)$ in $P_0$-probability.
\end{enumerate}
\end{lemma}
\begin{proof}
	The uniformity of $\eta\in\mathcal{H}_n$ related to the term $\upsilon_{\eta}$ is innocuous, as the posterior law of $F_{\eta}$ is equivalent to the Bayesian bootstrap measure, which no longer depends on $\eta$. That is, we can write
	\begin{equation*}
		\mathbb{E}_{\eta}[D]=\int \mathrm{d}F_{\eta}(y,1,x)=\sum_{i=1}^{n}M_{i}D_i, ~~\text{with}~~M_{i}=e_i/\sum_{i=1}^{n}e_i,~~\text{for} ~~e_i\stackrel{iid}{\sim} \textup{Exp}(1),
	\end{equation*}
	conditional on the observed data $Z^{(n)}$.
	
Proof of $(i)$. Because that $P_0\left[\left(\frac{D-\pi_0(X)}{1-\pi_0(X)}\right)\left(\Delta Y-m_{0}(X)\right)-D\tau_0\right]=0$, we can write
	\begin{equation*}
	R_{n,\eta}=\mathbb{G}_n\left[\left(\frac{D-\pi_0(X)}{1-\pi_0(X)}\right)\left(\Delta Y-m_{0}(X)\right)-D\tau_0\right]\times\left(1-\frac{\mathbb{E}_{\eta}[D]}{\bar\pi_0}\right).
	\end{equation*}
By the moment condition for the envelope function in Assumption \ref{Assump:NonDRSE} (i), the first term is $O_{P_0}(1)$ and the second term is $O_{P_M}(1)$ in $P_0$-probability. By the relationship in (71) of \cite{cheng2010bootstrap}, the remainder term $R_{n,\eta}$ converges to zero in $P_0$-probability, conditional on the data. 

	Proof of $(ii)$. For the second part, conditional on the observed data $Z^{(n)}$, we have
		\begin{align}\label{CLTDenominator}
			\sqrt{n}(\upsilon_{\eta}-1)&=\frac{\sqrt{n}}{\bar\pi_0}\left(\sum_{i=1}^{n}M_{i}D_i-\pi_0\right)\nonumber\\
			&=\frac{1}{\bar\pi_0}\left(\mathbb{G}_n^*[D]+\mathbb{G}_n[D] \right)=O_{P_M}(1)~~~\text{in}~~P_0-\text{probability},
		\end{align}
		where $\mathbb{G}_n^*$ denotes the Bayesian bootstrap weighted analog of $\mathbb{G}_n$.
	In addition, the definition of the bias term $b_{0,\eta}$ yields
	\begin{equation*}
		b_{0,\eta}=\frac{1}{n}\sum_{i=1}^{n}\gamma_0(D_i,X_i)[m_0(X_i)-m_{\eta}(X_i)]=(\mathbb{P}_n-P_0)[\gamma_0(m_0-m_{\eta})],
	\end{equation*} 
	where the second equation follows from $\mathbb{E}_0[\gamma_0(D,X)\mid X]=0$.
By the $P_0$-Glivenko-Cantelli property of the conditional mean function imposed in Assumption \ref{Assump:NonDRSE}(i), we have $	\sup_{\eta\in \mathcal{H}_n}|b_{0,\eta}|=o_{P_0}(1)$, which, combined with (\ref{CLTDenominator}), concludes the proof.
\end{proof}
\subsection{Prior Stability of GP Priors}\label{sec:GPStable}
In this section, we verify the prior stability using standard Gaussian process priors, which is used in the proof of Theorem \ref{BvM:thm:standard}. Here we follow the strategy in Section 5.3 of \cite{ray2020causal}. We first approximate $\eta_{t}^m$ by an element in the RKHS $\mathbb{H}$ and then apply the Cameron-Martin theorem in (\ref{CMThm}); see Proposition I.20 in \citep{ghosal2017fundamentals}. 
\begin{lemma}\label{lemma:prior:stability}
Under the conditions in Assumption \ref{Assump:NonDRPS}, we have
	\begin{equation}
\frac{\int_{\mathcal{H}_n}e^{\ell_n^m(\eta^m_t)}\mathrm{d}\Pi(\eta^m)}{\int_{\mathcal{H}_n}e^{\ell_n^m(\eta^{m\prime})}\mathrm{d}\Pi(\eta^{m\prime})} \rightarrow^{P_0} 1.
\end{equation} 	
\end{lemma}
\begin{proof}
 Let $\overline\gamma_n\in\mathbb{H}^m(\zeta_n)$, as stated in our Assumption \ref{Assump:NonDRPS}. Also, we set $\eta_{n,t}=\eta^m-t\overline\gamma_n/\sqrt{n}$. By the Cameron-Martin theorem, the distribution $\Pi_{n,t}$ of $\eta_{n,t}$ if $\eta^m$ is distributed according to the prior $\Pi$ has the Radon-Nikodym density  
\begin{equation}\label{E9_density}
	\frac{d\Pi_{n,t}}{d\Pi}(\eta^m)=\exp\left(tU_n(\eta^m)/\sqrt{n}-t^2\Vert \overline\gamma_n\Vert_{\mathbb H^m}^2/(2n)\right),
\end{equation}
where $U_n(\cdot)$ is a centered Gaussian variable with variance $\Vert \overline\gamma_n\Vert_{\mathbb H^m}^2$.

By the Gaussian tail bound, we have
\begin{equation}
	\Pi(\eta^m:\left|U_n(\eta^m)\right|>M\sqrt{n}\varepsilon_n\Vert\overline\gamma_n\Vert_{\mathbb H^m} )\leq 2\exp(-M^2n\varepsilon_n^2/2).
\end{equation}
As a result, the posterior measure of the set in the display tends to 0 in probability for large enough $M$ by Lemma 4 of \cite{ray2020causal}. Hence the set 
\begin{equation*}
	B_n:=\left\{\eta^m:\left|U_n(\eta^m)\right|\leq M\sqrt{n}\varepsilon_n\Vert\overline\gamma_n\Vert_{\mathbb H^m}  \right\}\cap\mathcal{H}_n
\end{equation*}
also satisfies $\Pi(B_n|Z^{(n)})\to1$ in probability. Considering (\ref{E9_density}) on the set $B_n$ and using Assumption \ref{Assump:NonDRPS}, we have
\begin{equation*}
	\left|\log\frac{d\Pi_{n,t}}{d\Pi}(\eta^m)\right|\leq M|t|\sqrt{n}\varepsilon_n\zeta_n+\frac{t^2\zeta_n^2}{2}\to 0.
\end{equation*}
By applying Lemma 3 in \cite{ray2020causal} with $A_n=B_n$, $\xi_0=\gamma_0$ and $w_n$ a sufficiently large fixed constant, we have
\begin{equation*}
	\sup_{\eta^m\in B_n}\left|\ell^m_n(\eta_{n,t})-\ell^m_n(\eta_t^m)\right|=o_{P_0}(1)
\end{equation*}
By the change of variable $\eta^m-t\overline\gamma_n/\sqrt{n}\mapsto v$, we have
\begin{equation*}
	\frac{\int_{B_n}e^{\ell^m_n(\eta^m_t)}d\Pi(\eta^m)}{\int_{B_n}e^{\ell^m_n(\eta^m)}d\Pi(\eta^m)}=\frac{\int_{B_n}e^{\ell^m_n(\eta_{n,t})}d\Pi(\eta^m)}{\int_{B_n}e^{\ell^m_n(\eta^m)}d\Pi(\eta^m)}e^{o_{P_0}(1)}=\frac{\int_{B_{n,t}}e^{\ell^m_n(v)}d\Pi_{n,t}(v)}{\int_{B_n}e^{\ell^m_n(\eta^m)}d\Pi(\eta^m)}e^{o_{P_0}(1)},
\end{equation*}
where $B_{n,t}=B_n-t\overline\gamma_n/\sqrt{n}$. We can replace $\Pi_{n,t}$ in the numerator by $\Pi$ at the cost of anther multiplicative $1+o_{P_0}(1)$ term. This makes the quotient into the ratio $\Pi(B_{n,t}|Z^{(n)})/\Pi(B_{n}|Z^{(n)})$. It has already been proved that $\Pi(B_n|Z^{(n)})=1-o_{P_0}(1)$, so it is sufficient to prove the same result for the numerator, i.e., $\Pi(B_{n,t}|Z^{(n)})=1-o_{P_0}(1)$. Note that
\begin{align*}
	B_{n,t}^c=&\left\{v:v+t\overline\gamma_n/\sqrt{n}\notin \mathcal{H}^m_n \right\}\cap \left\{v:\Vert \Upsilon^\prime(v+t\overline\gamma_n/\sqrt{n})-m_0(v)\Vert_{2,F_0}>\varepsilon_n \right\}\\
	&\cap \left\{v:|U_n(v+t\overline\gamma_n/\sqrt{n})|>M\sqrt{n}\varepsilon_n\Vert\overline\gamma_n\Vert_{\mathbb H^m} \right\}.
\end{align*}
The posterior probability of the first set tends to zero in probability by assumption. Considering the second term, we make use of the smoothness of the link function to get
\begin{equation*}
	\Vert \Upsilon^\prime(\eta^m+t\overline\gamma_n\sqrt{n})-\Upsilon^\prime(\eta^m)\Vert_{2,F_0}\lesssim \frac{\Vert\overline\gamma_n\Vert_{2, F_0}}{\sqrt{n}}\lesssim\frac{1}{\sqrt{n}}.
\end{equation*}
Therefore, the second set is contained in $\{\eta^m:\Vert \Upsilon^\prime(\eta^m)-m_0\Vert_{2, F_0}>\varepsilon_n-C/\sqrt{n} \}$, which has posterior probability $o_{P_0}(1)$.

When it comes to the third term, note that 
\begin{equation*}
	U_n(\eta^m+t\overline\gamma_n/\sqrt{n})\sim \mathbb{N}(-t\Vert\overline\gamma_n\Vert_{\mathbb H^m}^2/\sqrt{n},\Vert\overline\gamma_n\Vert_{\mathbb H^m}^2),
\end{equation*}
if $\eta^m$ is distributed according to the GP prior. Because the mean $t\Vert\overline\gamma_n\Vert_{\mathbb H^m}^2/\sqrt{n}$ of this Gaussian variable is negligible relative to its standard deviation, we can utilize the Gaussian tail bound to show $\Pi(|U_n(\eta^m+t\overline\gamma_n/\sqrt{n})|>M\sqrt{n}\varepsilon_n\Vert\overline\gamma_n\Vert_{\mathbb H^m})$ is exponentially small.
\end{proof}

\begin{lemma}\label{lemma:negligible} Let Assumptions \ref{Ass:unconfounded} and \ref{Assump:Rate} be satisfied. Then, we have
\begin{align*}
\sqrt n\,\mathbb P_n[\gamma_0 (\widehat m-m_0)]=o_{P_0}(1).
\end{align*}
\end{lemma}
\begin{proof}
The estimator $\widehat m$ is based on an auxiliary sample and hence it is sufficient to consider deterministic functions $m_n$ with the same rates of convergence as $\widehat m$.
We compute
\begin{align*}
&\mathbb E_0\left[\Big(\frac{1}{\sqrt n}\sum_{i=1}^n\gamma_0(D_i, X_i)(m_n-m_0)(X_i)\Big)^2\mid X_1, \ldots, X_n\right]\\
& = \frac{1}{n}\sum_{i, i'}(m_n-m_0)(X_i)(m_n-m_0)(X_{i'}) \mathbb E_0\left[\gamma_0(D_i, X_i) \gamma_0(D_{i'},X_{i'})\mid X_i, X_{i'}\right]\\
& = \frac{1}{n}\sum_{i=1}^n(m_n-m_0)^2(D_i,X_i) Var_0(\gamma_0(D_i, X_i)|X_i),
\end{align*}
using that 
\begin{align*}
\E_0[\gamma_0(D,X)\mid X]=\frac{\pi_0(X)}{\bar\pi_0}-\frac{1-\pi_0(X)}{\bar\pi_0} \frac{\pi_0(X)}{1-\pi_0(X)}=0
\end{align*}
Now overlap as imposed in Assumption \ref{Ass:unconfounded}(iii) implies
\begin{align*}
Var_0(\gamma_0(D_i, X_i)|X_i)=\frac{\pi_0(X)}{\bar\pi_0^2}+\frac{\pi_0^2(X)}{(1-\pi_0(X))\bar\pi_0^2} \lesssim 1.
\end{align*}
Consequently,  we obtain for the unconditional squared expectation that
\begin{align*}
\mathbb E_0\left[\Big(\frac{1}{\sqrt n}\sum_{i=1}^n\gamma_0(D_i, X_i)(m_n-m_0)(X_i)\Big)^2\right]
\lesssim  \Vert m_n-m_0\Vert_{2, F_0}^2=o(1)
\end{align*}
by Assumption \ref{Assump:Rate}, which implies the desired result. 
\end{proof}
\section{Computational Details in Algorithms}\label{details:implementation}

Recall that the prior placed on $m(x)$ in Algorithm \ref{algorithm_1} is a Gaussian process with mean $\mu$ and squared exponential (SE) covariance function \citep[p.83]{williams2006gaussian}
$K\left(x,x^\prime\right):= \nu^2 \exp\left(-\sum_{l=1}^{p}a_{ln}^{2}(x_{l}-x^\prime_{l})^2/2\right)$.  In implementation, hyperparameters $\mu$, $\nu^2$, $a_{1n},\ldots,a_{pn}$ and $\sigma^2$ (the variance of the noise $U$) are determined by maximizing the marginal likelihood. When the hyperparameters $a_n$ are as stated in Proposition \ref{prop:exponential:dr}, the convergence rate of $\widehat{m}$ is $O_{P_0}\big((n/\log n)^{-s_m/(2s_m+p)}\big)$. This can be shown by combining Theorems 11.22, 11.55 and 8.8 of \cite{ghosal2017fundamentals}. 
In Algorithm \ref{algorithm_2}, the adjusted prior placed on $m(x)$ has covariance kernel $K_c\left(x,x^\prime\right) = K\left(x,x^\prime\right)  + \varsigma_n^2\widehat \gamma(0,x)\,\widehat \gamma(0,x^\prime),$ cf. related constructions from \cite{ray2019debiased}, \cite{ray2020causal}, and \cite{BLY2022}.
The parameter $\varsigma_n$, representing the standard deviation of $\lambda$, controls the weight of the prior adjustment relative to the standard Gaussian process.
The choice $\varsigma_n=\nu\log n_c/(\sqrt{n_c}\, \Gamma_n)$ in Algorithm \ref{algorithm_2} satisfies the rate condition in Assumption \ref{Assump:Prior} with probability approaching one. It is similar to that suggested by \citet[page 6]{ray2019debiased}, which is proportional to $1/(\sqrt{n}\,\Gamma_n)$. The factor $\Gamma_n$ normalizes the second term (the adjustment term) of $K_c$  to have the same scale as the unadjusted covariance $K$.

We describe how Step (a) of Posterior Computation in Algorithm \ref{algorithm_2} is conducted. The corresponding step in Algorithm \ref{algorithm_1} immediately follows by replacing the adjusted kernel function $K_c$ by the original kernel function $K$.
Let $\boldsymbol y_0$ be the vector of $\left\{\Delta Y_i: D_i=0 \right\}$,
$\boldsymbol X_0\in \mathbb{R}^{n_c\times p}$ be the matrix of data $\left\{X_i: D_i=0\right\}$ and $\boldsymbol X \in \mathbb{R}^{n\times p}$ be the matrix of data $\left\{X_i: i=1,\cdots,n\right\}$. Let $\boldsymbol m_{n_c}$ and $\boldsymbol m_{n}$ be the $n_c$-vector and $n$-vector of the function $m(x)$ evaluated at $\boldsymbol X_0$ and $\boldsymbol X$ respectively:
\begin{equation*}
\boldsymbol m_{n_c} = \left(m(X_i): D_i=0\right)^{\top}, \text{and}\quad \boldsymbol m_n=\left[m(X_1),\dots,m (X_n)\right]^{\top}.
\end{equation*}
For matrices $\boldsymbol X$ and $\boldsymbol X_0$, we define $K_c(\boldsymbol X,\boldsymbol X_0)$ as a $n\times n_c$ matrix whose $(i,j)$-th element is $K_c(X_{i},X_{0j})$, where $X_{i}$ is the $i$-th row of $\boldsymbol X$ and $X_{0j}$ is the $j$-th row of $\boldsymbol X_0$. Analogously, $K_c(\boldsymbol X_0,\boldsymbol X_0)$ is an $n_c\times n_c$ matrix with the $(i,j)$-th element being $K_c(X_{0i},X_{0j})$, and $K_c(\boldsymbol X,\boldsymbol X)$ is  a $n\times n$ matrix with the $(i,j)$-th element being $K_c(X_{i},X_{j})$.

Given the GP prior with mean $\mu$ and covariance kernel $K_c$, the posterior of $\boldsymbol m_n$ has a Gaussian distribution with the mean $\bar{\boldsymbol m}_n$ and covariance V$(\boldsymbol m_n)$ specified as follows \citep[p.16]{williams2006gaussian}:
\begin{eqnarray*}
\bar{\boldsymbol m}_n&=& \mu\boldsymbol 1_{n} +  K_c(\boldsymbol X,\boldsymbol X_0)\left[K_c(\boldsymbol X_0,\boldsymbol X_0)+\sigma^2 \boldsymbol I_{n_c}\right]^{-1}\boldsymbol (y_0-\mu\boldsymbol 1_{n_c}),\\
\text{V}(\boldsymbol m_n) &=& K_c(\boldsymbol X,\boldsymbol X)-K_c(\boldsymbol X,\boldsymbol X_0)\left[K_c(\boldsymbol X_0,\boldsymbol X_0)+ \sigma^2 \boldsymbol I_{n_c}\right]^{-1}K_c^\top(\boldsymbol X,\boldsymbol X_0).
\end{eqnarray*}
We use the MATLAB toolbox $\mathtt{GPML}$ for implementation.\footnote{The $\mathtt{GPML}$ toolbox can be downloaded from \url{http://gaussianprocess.org/gpml/code/matlab/doc/}.}

\section{Additional Simulation Results}\label{appendix:simu}
\renewcommand{\thetable}{A\arabic{table}}
    \setcounter{table}{0}
This section provides additional simulation results. We first present the finite sample performance of an alternative approach with a Bayesian interpretation. Then we examine the performance of the double-robust Bayesian proposal under varying values of the tuning parameter $\varsigma_n$. We also consider scenarios where the conditional distribution of $\Delta Y$ given $(D, X^\top)$ does not belong to the single-parameter exponential family. Finally, we consider the cases where the propensity score may take extreme values.

\subsection{Comparison with A Bayesian Bootstrap Method}\label{sec:alternativeBootstrap}
%\textcolor{red}{I think this part is good to show the posterior draws of GP for the conditional mean function actually played some role for the finite sample performance, not just the Bayesian bootstrap.}
As an alternative to our Bayesian DiD, which imposes a Gaussian process prior on the conditional mean function, one may apply the Bayesian bootstrap to the entire procedure.\footnote{We thank an anonymous referee for pointing out this approach.} To illustrate,  let us consider the outcome regression approach using Bayesian bootstrap.  Assume a parametric model $m_0(x)=x^\prime\beta_0$.  For $s=1,\dots, S$,  let $\left\{M_i^s, i=1,\dots, n\right\}$ be  the Bayesian bootstrap weights generated in Algorithm 1.  For each $s$,  run the following Bayesian bootstrap weighted OLS:
\[
\hat{\beta}^s = \arg\min_{\beta}  \sum_{i=1}^n M_i^s (1-D_i)\left(\Delta Y_i - X_i^\prime \beta\right)^2
\]
and obtain a Bayesian bootstrap distribution of the ATT $\left\{\hat{\tau}^{s}, s=1,\dots,S\right\}$ where
\[
\hat{\tau}^{s}=\frac{\sum_{i} M_{i}^{s} D_{i}\left(\Delta Y_{i}-X_i^\prime \hat\beta^s\right)}{\sum_{i} M_{i}^s D_{i}} .
\]
The empirical quantiles of the distribution $\left\{\hat{\tau}^{s}, s=1, \dots,S\right\}$ are then used to construct confidence intervals for the ATT.  Its frequentist validity is guaranteed by a proper modification of  \cite{chen2009efficient},  which covers the weighted bootstrap for a general class of semiparametric models.
The Bayesian bootstrap scheme also has a Bayesian interpretation, but it differs from our Bayesian proposals in modeling the conditional mean function $m_0(x)$.  Our Bayesian proposals explicitly place a Gaussian process prior on it and calculate its posterior using an \textit{i.i.d.} normal model,  whereas the Bayesian bootstrap scheme approximates its posterior by repeatedly solving the least square problem in lieu of \cite{chamberlain2003bayesian}.

We apply the Bayesian bootstrap scheme to the OR, IPW and DR approaches and report the results in Table \ref{tab:simu_n1K_bb}.  As shown, the Bayesian bootstrap scheme version of these frequentist methods performs very similarly to their counterparts in Table~\ref{tab:simu_n1K} in Section \ref{sec:simulations} that use plug-in standard errors. On the other hand, the Bayesian bootstrap scheme performs differently from our Bayesian proposals as shown in Table~\ref{tab:simu_n1K}.
%This indicates that the uncertainty quantification in our approach is Bayesian in character rather than a resampling approximation.

\begin{table}[H]
\centering
\caption{Simulation results for Bayesian bootstrap scheme under DGPs 1 to 4. Sample size $n=1000$.}
\vskip.15cm
{\small 
\label{tab:simu_n1K_bb}
\setlength{\tabcolsep}{3pt}
\renewcommand{\arraystretch}{0.75}

\begin{tabular}{l cccc| cccc}
\toprule
Method
& \multicolumn{4}{c}{DGP 1}
& \multicolumn{4}{c}{DGP 2} \\
\cmidrule(lr){2-5}\cmidrule(lr){6-9}
& Bias & RMSE & CP & CIL
& Bias & RMSE & CP & CIL \\
\midrule
DR
  & $-0.003$ & $0.107$ & $0.936$ & $0.403$
  & $-0.000$ & $0.107$ & $0.935$ & $0.399$ \\
OR
  & $-0.002$ & $0.102$ & $0.945$ & $0.393$
  & $-0.001$ & $0.103$ & $0.936$ & $0.392$ \\
IPW
  & $-0.040$ & $1.140$ & $0.949$ & $4.273$
  & $-0.784$ & $1.244$ & $0.850$ & $3.647$ \\
\midrule
Method
& \multicolumn{4}{c}{DGP 3}
& \multicolumn{4}{c}{DGP 4} \\
\cmidrule(lr){2-5}\cmidrule(lr){6-9}
& Bias & RMSE & CP & CIL
& Bias & RMSE & CP & CIL \\
\midrule
DR
  & $-0.029$ & $1.008$ & $0.952$ & $3.825$
  & $-2.498$ & $2.689$ & $0.286$ & $3.821$ \\
OR
  & $-1.347$ & $1.822$ & $0.805$ & $4.773$
  & $-5.149$ & $5.301$ & $0.011$ & $4.992$ \\
IPW
  & $-0.000$ & $1.461$ & $0.939$ & $5.439$
  & $-3.889$ & $4.156$ & $0.240$ & $5.680$ \\
\bottomrule
\end{tabular}}
\end{table}

\subsection{Sensitivity with respect to the tuning parameter $\varsigma_n$}\label{appendix:sensivity:setup}
Table \ref{tab:simu_sigma_n1K} evaluates the sensitivity of the finite sample performance of DR Bayes with respect to the standard deviation $\varsigma_n$ that controls the strength of the prior correction term. We scale $\varsigma_n$ in Algorithm \ref{algorithm_2} by a factor
%set $\varsigma_n = c_{\varsigma} \times \nu(\log n_c)/(\sqrt{n_c}\,\Gamma_n)$ with 
$c_{\varsigma}\in\{1/5, 1/2, 2, 5\}$.  Note that $c_{\varsigma}=1$ corresponds to the simulation results of DR Bayes in Table \ref{tab:simu_n1K}.  Table \ref{tab:simu_sigma_n1K}  shows that the performance of DR Bayes is stable to the choice of $c_{\varsigma}$. 

\begin{table}[H]
\centering
\caption{The effect of prior adjustment strength $\varsigma_n$ on DR Bayes under DGPs~1 to~4. Sample size $n=1000$.}
\vskip.15cm
{\small 
\label{tab:simu_sigma_n1K}
\setlength{\tabcolsep}{3pt}
\renewcommand{\arraystretch}{0.75}
\begin{tabular}{lccccc| cccc}
\toprule
Method & $c_{\varsigma}$
  & \multicolumn{4}{c}{DGP 1}
  & \multicolumn{4}{c}{DGP 2} \\
\cmidrule(lr){3-6}\cmidrule(lr){7-10}
  & & Bias & RMSE & CP & CIL
  & Bias & RMSE & CP & CIL \\
\midrule
DR Bayes$^{\text{Logit}}$ & $1/5$
  & $-0.003$ & $0.111$ & $0.952$ & $0.447$
  & $-0.000$ & $0.110$ & $0.956$ & $0.433$ \\
 & $1/2$
  & $-0.003$ & $0.111$ & $0.952$ & $0.447$
  & $-0.000$ & $0.110$ & $0.956$ & $0.433$ \\
 & $2$
  & $-0.003$ & $0.111$ & $0.952$ & $0.447$
  & $-0.000$ & $0.110$ & $0.955$ & $0.433$ \\
 & $5$
  & $-0.003$ & $0.111$ & $0.951$ & $0.447$
  & $-0.000$ & $0.110$ & $0.954$ & $0.433$ \\
\cline{1-10}
DR Bayes$^{\text{RF}}$ & $1/5$
  & $-0.002$ & $0.147$ & $0.965$ & $0.596$
  & $-0.001$ & $0.148$ & $0.961$ & $0.611$ \\
 & $1/2$
  & $-0.002$ & $0.147$ & $0.965$ & $0.596$
  & $-0.001$ & $0.148$ & $0.961$ & $0.611$ \\
 & $2$
  & $-0.002$ & $0.147$ & $0.965$ & $0.596$
  & $-0.001$ & $0.148$ & $0.960$ & $0.611$ \\
 & $5$
  & $-0.002$ & $0.147$ & $0.965$ & $0.596$
  & $-0.001$ & $0.148$ & $0.962$ & $0.611$ \\
\midrule
Method & $c_{\varsigma}$
  & \multicolumn{4}{c}{DGP 3}
  & \multicolumn{4}{c}{DGP 4} \\
\cmidrule(lr){3-6}\cmidrule(lr){7-10}
  & & Bias & RMSE & CP & CIL
  & Bias & RMSE & CP & CIL \\
\midrule
DR Bayes$^{\text{Logit}}$ & $1/5$
  & $\phantom{-}0.003$ & $0.832$ & $0.953$ & $3.210$
  & $-1.102$ & $1.361$ & $0.752$ & $3.328$ \\
 & $1/2$
  & $\phantom{-}0.010$ & $0.833$ & $0.956$ & $3.308$
  & $-1.081$ & $1.344$ & $0.775$ & $3.432$ \\
 & $2$
  & $\phantom{-}0.012$ & $0.834$ & $0.958$ & $3.357$
  & $-1.072$ & $1.336$ & $0.790$ & $3.496$ \\
 & $5$
  & $\phantom{-}0.013$ & $0.834$ & $0.958$ & $3.361$
  & $-1.071$ & $1.336$ & $0.791$ & $3.501$ \\
\cline{1-10}
DR Bayes$^{\text{RF}}$ & $1/5$
  & $-0.001$ & $1.066$ & $0.972$ & $3.941$
  & $-0.703$ & $1.446$ & $0.911$ & $4.183$ \\
 & $1/2$
  & $-0.001$ & $1.068$ & $0.972$ & $3.950$
  & $-0.690$ & $1.441$ & $0.913$ & $4.193$ \\
 & $2$
  & $-0.001$ & $1.069$ & $0.972$ & $3.952$
  & $-0.687$ & $1.440$ & $0.915$ & $4.195$ \\
 & $5$
  & $-0.001$ & $1.069$ & $0.972$ & $3.952$
  & $-0.687$ & $1.440$ & $0.915$ & $4.195$ \\
\bottomrule
\end{tabular}}
\end{table}

\subsection{Sensitivity with respect to error distributions}\label{appendix:sensivity:error}
We consider scenarios when the error terms $\epsilon_{1}$, $\epsilon_{2}(0)$ and $\epsilon_{2}(1)$ in the simulation designs deviate from the standard normal distribution.  Table \ref{tab:simu_chi_n1K} presents the results for the designs where the error terms follow a $\chi^2$ distribution with three degrees of freedom  (normalized to have a mean of zero and unit variance). Table \ref{tab:simu_hetero_n1K} considers the case of heteroskedastic errors where $\epsilon_{2}(d)\sim N(0, e(x))$ and $e(x)=\sum_{j=1}^4x_j^2/8$, for $d\in\{0,1\}$. As shown, both Bayesian and frequentist methods exhibit performance similar to that observed in Table \ref{tab:simu_n1K},  which considers normal errors. Thus, the results suggest that our Bayesian methods, which assume \textit{i.i.d. } normal errors, continue to deliver strong finite-sample performance even when the underlying error distribution deviates from standard normality. This is consistent with the theoretical results on posterior behavior under misspecification established by \cite{kleijn2006misInf}.

\begin{table}[H]
\centering
\caption{Simulation results under DGPs~1 to~4 with $\chi^2(3)$ errors. Sample size $n=1000$.}
\vskip.15cm
{\small 
\label{tab:simu_chi_n1K}
\setlength{\tabcolsep}{3pt}
\renewcommand{\arraystretch}{0.75}
\begin{tabular}{l cccc| cccc}
\toprule
Method
  & \multicolumn{4}{c}{DGP 1}
  & \multicolumn{4}{c}{DGP 2} \\
\cmidrule(lr){2-5}\cmidrule(lr){6-9}
  & Bias & RMSE & CP & CIL
  & Bias & RMSE & CP & CIL \\
\midrule
OR Bayes
  & $-0.001$ & $0.108$ & $0.940$ & $0.406$
  & $0.002$ & $0.110$ & $0.936$ & $0.409$ \\
DR Bayes$^{\text{Logit}}$
  & $-0.001$ & $0.114$ & $0.945$ & $0.446$
  & $0.000$ & $0.113$ & $0.943$ & $0.432$ \\
DR Bayes$^{\text{RF}}$
  & $0.002$ & $0.156$ & $0.965$ & $0.596$
  & $0.001$ & $0.257$ & $0.957$ & $0.643$ \\
DR Bayes$^{\text{Logit}}$(FS)
  & $-0.001$ & $0.110$ & $0.932$ & $0.419$
  & $0.001$ & $0.111$ & $0.929$ & $0.406$ \\
DR Bayes$^{\text{RF}}$(FS)
  & $-0.002$ & $0.140$ & $0.964$ & $0.587$
  & $0.002$ & $0.154$ & $0.963$ & $0.610$ \\
\cline{1-9}
DR
  & $-0.001$ & $0.110$ & $0.926$ & $0.408$
  & $0.001$ & $0.110$ & $0.930$ & $0.404$ \\
OR
  & $-0.000$ & $0.106$ & $0.934$ & $0.396$
  & $0.002$ & $0.106$ & $0.930$ & $0.394$ \\
IPW
  & $0.023$ & $1.129$ & $0.943$ & $4.311$
  & $-0.809$ & $1.240$ & $0.838$ & $3.628$ \\
DML$^{\text{RF-RF}}$
  & $-0.178$ & $0.712$ & $0.995$ & $4.460$
  & $0.493$ & $0.994$ & $0.991$ & $5.015$ \\
DML$^{\text{RF-NN}}$
  & $-0.015$ & $0.149$ & $0.968$ & $0.656$
  & $-0.018$ & $0.178$ & $0.961$ & $0.737$ \\
DML$^{\text{RF-GP}}$
  & $-0.362$ & $0.963$ & $0.967$ & $4.687$
  & $0.149$ & $1.125$ & $0.990$ & $5.882$ \\
\midrule
Method
  & \multicolumn{4}{c}{DGP 3}
  & \multicolumn{4}{c}{DGP 4} \\
\cmidrule(lr){2-5}\cmidrule(lr){6-9}
  & Bias & RMSE & CP & CIL
  & Bias & RMSE & CP & CIL \\
\midrule
OR Bayes
  & $-0.053$ & $0.419$ & $0.993$ & $2.256$
  & $-0.550$ & $0.750$ & $0.979$ & $3.231$ \\
DR Bayes$^{\text{Logit}}$
  & $-0.014$ & $0.798$ & $0.962$ & $3.308$
  & $-1.060$ & $1.328$ & $0.778$ & $3.402$ \\
DR Bayes$^{\text{RF}}$
  & $-0.060$ & $0.920$ & $0.970$ & $3.934$
   & $-0.782$  & $1.591$ & $0.893$ & $4.203$ \\
DR Bayes$^{\text{Logit}}$(FS)
  & $-0.057$ & $0.420$ & $0.965$ & $1.774$
  & $-0.546$ & $0.746$ & $0.844$ & $2.083$ \\
DR Bayes$^{\text{RF}}$(FS)
  & $-0.056$ & $0.441$ & $0.990$ & $2.252$
  & $-0.499$ & $0.717$ & $0.936$ & $2.628$ \\
\cline{1-9}
DR
  & $-0.077$ & $1.010$ & $0.946$ & $3.855$
  & $-2.528$ & $2.713$ & $0.266$ & $3.851$ \\
OR
  & $-1.389$ & $1.861$ & $0.807$ & $4.818$
  & $-5.174$ & $5.334$ & $0.019$ & $5.051$ \\
IPW
  & $-0.019$ & $1.428$ & $0.947$ & $5.467$
  & $-3.923$ & $4.203$ & $0.243$ & $5.721$ \\
DML$^{\text{RF-RF}}$
  & $0.273$ & $0.871$ & $0.987$ & $4.484$
  & $-0.842$ & $1.242$ & $0.960$ & $5.071$ \\
DML$^{\text{RF-NN}}$
  & $0.032$ & $0.809$ & $0.983$ & $3.906$
  & $-0.985$ & $1.448$ & $0.860$ & $4.648$ \\
DML$^{\text{RF-GP}}$
  & $0.244$ & $1.029$ & $0.992$ & $5.308$
  & $-0.888$ & $1.650$ & $0.939$ & $7.377$ \\
\bottomrule
\end{tabular}}
\end{table}

\begin{table}[H]
\centering
\caption{Simulation results under DGPs~1 to~4 with heteroskedastic errors. Sample size $n=1000$.}
\vskip.15cm
{\small 
\label{tab:simu_hetero_n1K}
\setlength{\tabcolsep}{3pt}
\renewcommand{\arraystretch}{0.75}
\begin{tabular}{l cccc| cccc}
\toprule
Method
  & \multicolumn{4}{c}{DGP 1}
  & \multicolumn{4}{c}{DGP 2} \\
\cmidrule(lr){2-5}\cmidrule(lr){6-9}
  & Bias & RMSE & CP & CIL
  & Bias & RMSE & CP & CIL \\
\midrule
OR Bayes
  & $-0.001$ & $0.091$ & $0.950$ & $0.356$
  & $0.002$ & $0.091$ & $0.944$ & $0.357$ \\
DR Bayes$^{\text{Logit}}$
  & $-0.003$ & $0.105$ & $0.939$ & $0.388$
  & $-0.001$ & $0.101$ & $0.935$ & $0.378$ \\
DR Bayes$^{\text{RF}}$
  & $-0.005$ & $0.134$ & $0.955$ & $0.510$
  & $-0.002$ & $0.137$ & $0.954$ & $0.527$ \\
DR Bayes$^{\text{Logit}}$(FS)
  & $-0.002$ & $0.096$ & $0.948$ & $0.364$
  & $0.001$ & $0.094$ & $0.936$ & $0.351$ \\
DR Bayes$^{\text{RF}}$(FS)
  & $-0.003$ & $0.139$ & $0.966$ & $0.528$
  & $-0.001$ & $0.128$ & $0.962$ & $0.529$ \\
\cline{1-9}
DR
  & $-0.002$ & $0.095$ & $0.940$ & $0.355$
  & $0.001$ & $0.092$ & $0.943$ & $0.354$ \\
OR
  & $-0.004$ & $0.088$ & $0.954$ & $0.341$
  & $-0.002$ & $0.086$ & $0.950$ & $0.339$ \\
IPW
  & $-0.009$ & $1.149$ & $0.950$ & $4.326$
  & $-0.793$ & $1.245$ & $0.851$ & $3.631$ \\
DML$^{\text{RF-RF}}$
  & $-0.178$ & $0.692$ & $1.000$ & $4.484$
  & $0.496$ & $0.968$ & $0.991$ & $4.922$ \\
DML$^{\text{RF-NN}}$
  & $-0.016$ & $0.152$ & $0.959$ & $0.628$
  & $-0.022$ & $0.180$ & $0.969$ & $0.719$ \\
DML$^{\text{RF-GP}}$
  & $-0.382$ & $0.950$ & $0.963$ & $4.750$
  & $0.142$ & $1.033$ & $0.998$ & $5.785$ \\
\midrule
Method
  & \multicolumn{4}{c}{DGP 3}
  & \multicolumn{4}{c}{DGP 4} \\
\cmidrule(lr){2-5}\cmidrule(lr){6-9}
  & Bias & RMSE & CP & CIL
  & Bias & RMSE & CP & CIL \\
\midrule
OR Bayes
  & $-0.041$ & $0.416$ & $0.992$ & $2.274$
  & $-0.529$ & $0.723$ & $0.977$ & $3.224$ \\
DR Bayes$^{\text{Logit}}$
  & $0.023$ & $0.853$ & $0.957$ & $3.305$
  & $-1.047$ & $1.305$ & $0.806$ & $3.437$ \\
DR Bayes$^{\text{RF}}$
  & $-0.000$ & $1.051$ & $0.970$ & $3.932$
  & $-0.700$ & $1.443$ & $0.907$ & $4.173$ \\
DR Bayes$^{\text{Logit}}$(FS)
  & $-0.046$ & $0.417$ & $0.965$ & $1.766$
  & $-0.526$ & $0.719$ & $0.852$ & $2.088$ \\
DR Bayes$^{\text{RF}}$(FS)
  & $-0.050$ & $0.426$ & $0.992$ & $2.290$
  & $-0.489$ & $0.696$ & $0.936$ & $2.603$ \\
\cline{1-9}
DR
  & $-0.028$ & $1.006$ & $0.951$ & $3.852$
  & $-2.497$ & $2.687$ & $0.284$ & $3.853$ \\
OR
  & $-1.347$ & $1.821$ & $0.813$ & $4.815$
  & $-5.149$ & $5.301$ & $0.009$ & $5.044$ \\
IPW
  & $0.001$ & $1.460$ & $0.942$ & $5.460$
  & $-3.888$ & $4.155$ & $0.241$ & $5.710$ \\
DML$^{\text{RF-RF}}$
  & $0.280$ & $0.831$ & $0.992$ & $4.466$
  & $-0.851$ & $1.203$ & $0.966$ & $5.026$ \\
DML$^{\text{RF-NN}}$
  & $0.050$ & $0.815$ & $0.986$ & $3.921$
  & $-1.016$ & $1.435$ & $0.856$ & $4.597$ \\
DML$^{\text{RF-GP}}$
  & $0.257$ & $1.009$ & $0.992$ & $5.308$
  & $-0.910$ & $1.597$ & $0.940$ & $7.270$ \\
\bottomrule
\end{tabular}}
\end{table}

\subsection{Sensitivity with respect to weaker overlap}
This section considers a scenario where the propensity score is more likely to take values close to $1$, bringing the overlap condition close to being violated.  
In our simulation DGPs 1 to 4,  we double the scaling factor in function $g(\cdot)$ from $0.75$ to $1.5$. Specifically,  let $g(w) = 1.5(-w_1 + 0.5w_2 - 0.25w_3 - 0.1w_4)$.  This increases the probability of the propensity score taking values close to $1$. We discard observations with estimated propensity scores outside the range $(0, 1 - t]$, where the trimming threshold $t = 0.01$. When compared with Table \ref{tab:simu_n1K}, the performance of all methods deteriorates in Table \ref{tab:simu_overlap_t001_n1K}, especially under DGP~4, but the relative performance among Bayesian methods, frequentist parametric methods and DML methods exhibit a similar pattern as observed in Table \ref{tab:simu_n1K}. Take DGP~4 as example, all methods except OR Bayes significantly undercover. Among them, DR Bayes$^{\text{RF}}$ and its full sample version still exhibit substantially better coverage performance relative to the frequentist parametric methods and DML$^{\text{RF-NN}}$.  At the same time,  DR Bayes$^{\text{RF}}$ also produces shorter confidence intervals than DML$^{\text{RF-RF}}$ and DML$^{\text{RF-GP}}$.

\begin{table}[H]
\centering
\caption{Simulation results under DGPs~1 to~4 with weaker overlapping. Sample size $n=1000$. Trimming threshold for propensity score estimates  $t=0.01$.}
\vskip.15cm
{\small 
\label{tab:simu_overlap_t001_n1K}
\setlength{\tabcolsep}{3pt}
\renewcommand{\arraystretch}{0.75}
\begin{tabular}{l cccc| cccc}
\toprule
Method
  & \multicolumn{4}{c}{DGP 1}
  & \multicolumn{4}{c}{DGP 2} \\
\cmidrule(lr){2-5}\cmidrule(lr){6-9}
  & Bias & RMSE & CP & CIL
  & Bias & RMSE & CP & CIL \\
\midrule
OR Bayes
  & $-0.000$ & $0.127$ & $0.953$ & $0.505$
  & $0.000$ & $0.125$ & $0.952$ & $0.509$\\
DR Bayes$^{\text{Logit}}$
  & $-0.001$ & $0.171$ & $0.961$ & $0.738$
  & $-0.005$ & $0.156$ & $0.960$ & $0.669$\\
DR Bayes$^{\text{RF}}$
  & $-0.006$ & $0.227$ & $0.974$ & $0.927$
  & $-0.005$ & $0.241$ & $0.964$ & $0.954$\\
DR Bayes$^{\text{Logit}}$(FS)
  & $-0.002$ & $0.161$ & $0.941$ & $0.631$
  & $-0.003$ & $0.145$ & $0.939$ & $0.558$\\
DR Bayes$^{\text{RF}}$(FS)
  & $-0.009$ & $0.215$ & $0.969$ & $0.910$
  & $-0.026$ & $0.679$ & $0.968$ & $1.457$\\
\cline{1-9}
DR
  & $-0.002$ & $0.156$ & $0.910$ & $0.543$
  & $-0.003$ & $0.152$ & $0.923$ & $0.536$\\
OR
  & $-0.005$ & $0.119$ & $0.943$ & $0.465$
  & $-0.003$ & $0.116$ & $0.953$ & $0.460$\\
IPW
  & $-0.090$ & $2.720$ & $0.916$ & $9.137$
  & $-1.321$ & $2.373$ & $0.848$ & $6.965$\\
DML$^{\text{RF-RF}}$
  & $-0.587$ & $1.398$ & $0.958$ & $5.929$
  & $1.358$ & $1.976$ & $0.878$ & $6.738$\\
DML$^{\text{RF-NN}}$
  & $-0.022$ & $0.233$ & $0.941$ & $0.882$
  & $-0.038$ & $0.302$ & $0.944$ & $1.059$\\
DML$^{\text{RF-GP}}$
  & $-1.574$ & $2.240$ & $0.789$ & $7.017$
  & $-0.171$ & $1.964$ & $0.956$ & $9.753$\\
\midrule
Method
  & \multicolumn{4}{c}{DGP 3}
  & \multicolumn{4}{c}{DGP 4} \\
\cmidrule(lr){2-5}\cmidrule(lr){6-9}
  & Bias & RMSE & CP & CIL
  & Bias & RMSE & CP & CIL \\
\midrule
OR Bayes
  & $-0.005$ & $0.612$ & $0.992$ & $3.944$
  & $-1.515$ & $1.788$ & $0.974$ & $6.933$\\
DR Bayes$^{\text{Logit}}$
  & $-0.080$ & $1.550$ & $0.947$ & $6.002$
  & $-2.521$ & $3.025$ & $0.622$ & $6.046$\\
DR Bayes$^{\text{RF}}$
  & $-0.207$ & $1.733$ & $0.929$ & $5.922$
  & $-1.608$ & $2.673$ & $0.773$ & $6.488$\\
DR Bayes$^{\text{Logit}}$(FS)
  & $-0.020$ & $0.602$ & $0.968$ & $2.631$
  & $-1.506$ & $1.799$ & $0.537$ & $3.168$\\
DR Bayes$^{\text{RF}}$(FS)
  & $-0.020$ & $0.657$ & $0.983$ & $3.386$
  & $-1.393$ & $1.689$ & $0.730$ & $4.016$\\
\cline{1-9}
DR
  & $-0.200$ & $1.207$ & $0.920$ & $4.326$
  & $-3.983$ & $4.150$ & $0.047$ & $4.077$\\
OR
  & $-3.689$ & $4.005$ & $0.363$ & $6.178$
  & $-10.865$ & $10.971$ & $0.000$ & $5.968$\\
IPW
  & $-0.089$ & $2.748$ & $0.928$ & $9.515$
  & $-7.844$ & $8.197$ & $0.111$ & $8.641$\\
DML$^{\text{RF-RF}}$
  & $1.117$ & $1.655$ & $0.896$ & $5.712$
  & $-2.434$ & $3.055$ & $0.718$ & $7.854$\\
DML$^{\text{RF-NN}}$
  & $-0.007$ & $1.413$ & $0.961$ & $5.903$
  & $-2.931$ & $3.598$ & $0.548$ & $7.937$\\
DML$^{\text{RF-GP}}$
  & $-0.206$ & $1.639$ & $0.960$ & $7.740$
  & $-3.230$ & $4.271$ & $0.700$ & $12.239$\\
\bottomrule
\end{tabular}}
\end{table}

%
%\bibliographystyle{econometrica}
%\bibliography{Bayes_bib}

%\end{document}	

\end{document}